%% file: main.tex
\documentclass[envcountsame,runningheads,orivec]{llncs}
\usepackage{graphicx,multirow,tabularx,makecell}
\usepackage{hyperref,microtype}
\usepackage[compact]{titlesec}
\usepackage{paralist}
\RequirePackage{symbolic-bisim}

\RequirePackage[disable,textsize=tiny]{todonotes}
\RequirePackage{mathpartir}
\RequirePackage{cleveref}
\crefname{section}{Sec.}{Sec.(s)}
\crefname{appendix}{Appendix}{Appendices}
\crefname{figure}{Fig.}{Fig.(s)}
\crefname{theorem}{Thm.}{Thm.(s)}
\crefname{proposition}{Prop.}{Prop.(s)}
\crefname{corollary}{Cor.}{Cor.(s)}
\crefname{lemma}{Lem.}{Lem.(s)}
\crefname{example}{Ex.}{Ex.(s)}
\crefname{remark}{Rem.}{Rem.(s)}

\usepackage{versions}
\excludeversion{exclproof}
\excludeversion{technical}
\newcommand\mytechnical[1]{}

\newcommand\cutout[1]{}

\begin{document}
\title{From Bounded Checking to Verification of Equivalence via Symbolic Up-to Techniques
\thanks{This publication has emanated from research supported in part by a grant from Science Foundation Ireland under Grant number 13/RC/2094\_2. For the purpose of Open Access, the author has applied a CC BY public copyright licence to any Author Accepted Manuscript version arising from this submission.}}
\titlerunning{From Bounded Checking to Verification of Equivalence}
%
%
\author{Vasileios Koutavas\inst{1}%
 \and
Yu-Yang Lin\inst{1}%
$^{\star}$%
\and
Nikos Tzevelekos\inst{2}%
}

\authorrunning{V. Koutavas, Y.-Y. Lin, N. Tzevelekos}
%

\institute{Trinity College Dublin, Ireland
\email{\{Vasileios.Koutavas,linhouy\}@tcd.ie} \and
Queen Mary University of London, UK
\email{nikos.tzevelekos@qmul.ac.uk}}

\maketitle              
\begin{abstract}
We present a bounded equivalence verification technique for higher-order programs with local state.
  This technique combines fully abstract \emph{symbolic environmental bisimulations} similar to symbolic game semantics, novel \emph{up-to techniques}, and lightweight \emph{state invariant annotations}.
  This yields an equivalence verification technique with no false positives or negatives. The technique is bounded-complete, in that all inequivalences are automatically detected given large enough bounds. Moreover, several hard equivalences are proved automatically or after being annotated with state invariants.
We realise the technique in a tool prototype called \Hobbit{} and benchmark it with an extensive set of new and existing examples. 
  \Hobbit{} can prove many classical equivalences including all Meyer and Sieber examples.

\keywords{Contextual equivalence  \and bounded model checking \and symbolic bisimulation \and up-to techniques \and operational game semantics.}
\end{abstract}

  \section{Introduction}
  \label{sec:intro}
  \input{intro.tex}

  \section{High-Level Intuitions}
  \label{sec:intuitions}
  \input{intuitions.tex}

  \section{Language and Semantics}
  \label{sec:lang}
  \input{lang.tex}

  \section{LTS with Symbolic Higher-Order Transitions}
  \label{sec:lts}
  \input{lts.tex}

  \section{Up-to Techniques}
  \label{sec:up-to}
  \input{upto.tex}

  \section{Symbolic First-Order Transitions}
  \label{sec:sym}

\input{sym.tex}

  \section{Up to State Invariants}
  \label{sec:up-to-gen}
  \input{generalisation.tex}

  \section{Implementation and Evaluation}
  \label{sec:imp}
  \input{imp.tex}

  \section{Comparison with Existing Tools}
  \label{sec:comp}
  \input{syteci.tex}

  \section{Conclusion}
  \label{sec:rel}
  \input{rel.tex}


  \bibliographystyle{splncs04}

  \bibliography{references}

  \clearpage
  \emph{\footnotesize This appendix is provided for the benefit of the reviewers, and will not appear in a final version of this paper.}
  \appendix
\includeversion{technical}

\input{typing.tex}
  \input{SC.tex}
  \input{enhancements.tex}

\begin{technical}
  \section{Soundness of $(\bisimil)$} \label{sec:soundness}
  \input{soundness.tex}
  \input{completeness}

\end{technical}
  \input{examples}
\end{document}

%% file: intro.tex

Contextual equivalence 
is a relation over program expressions which guarantees that related expressions are interchangeable in any program context.
It encompasses verification properties like safety and termination.
It has attracted considerable attention from the semantics community (cf.~the 2017 Alonzo Church Award), and has found its main applications in
 the verification of cryptographic protocols~\cite{cryptoverif}, compiler correctness~\cite{700CC}
and regression verification~\cite{regressionverif,rvt,reve,symdiff}.

In its full generality, 
contextual equivalence is hard as it requires reasoning about the behaviour of all program contexts,
and becomes even more difficult in languages with higher-order features (e.g.\ callbacks) and local state.
Advances in bisimulations~\cite{KoutavasW06,EnvBisim07,BiernackiLP19}, logical relations~\cite{AhmedDR09,HurDNV12,JaberTabareau15} and game semantics~\cite{Laird07,MurawskiT16,Dimovski14,LinT20}
have offered powerful theoretical techniques for hand-written proofs of contextual equivalence in higher-order languages with state.
However, these advancements have yet to be fully integrated in verification tools for contextual equivalence in programming languages, especially in the case of bisimulation techniques. 
Existing tools~\cite{hector,coneqct,syteci} only tackle carefully delineated language fragments.

In this paper we aim to push the frontier further by  proposing a bounded model checking technique for contextual equivalence
for the entirety of a higher-order language with local state (\cref{sec:lang}).
This technique, realised in 
a tool called
\Hobbit,\footnote{Higher Order Bounded BIsimulation Tool (\Hobbit), {\scriptsize\url{https://github.com/LaifsV1/Hobbit}}.}
automatically detects inequivalent program expressions given sufficient bounds,
and proves hard equivalences automatically or semi-automatically.

Our technique
uses a labelled transition system (LTS) for open expressions
in order to express equivalence as a bisimulation.
The LTS
is symbolic both for
higher-order arguments (\cref{sec:lts}), similarly to symbolic game models~\cite{Dimovski14,LinT20} and derived proof techniques~\cite{BiernackiLP19,JaberTabareau15},
and
first-order ones (\cref{sec:sym}), following established techniques (e.g.~\cite{ClarkeEtal04}).
This enables the definition of a fully abstract \emph{symbolic environmental bisimulation}, the bounded exploration of which is the task of the \Hobbit\ tool.
Full abstraction guarantees that our tool finds all inequivalences given sufficient bounds, and only reports true inequivalences.
As is corroborated by our experiments, this makes \Hobbit\ a practical inequivalence detector, 
similar to traditional bounded model checking~\cite{ArminEtal99} 
which has been proved an effective bug detection technique in industrial-scale C code~\cite{ClarkeEtal04,CordeiroEtal19,SchrammelKBMTB15}.

However, while proficient in bug finding, bounded model checking can rarely prove the absence of errors, and in our setting prove an equivalence: a bound is usually reached before all---potentially infinite---program runs  are explored.
{Inspired by hand-written equivalence proofs, we address this challenge by proposing two key technologies:
new \emph{bisimulation up-to techniques}, and lightweight user guidance in the form of \emph{state invariant annotations}.}
Hence we increase significantly the number of equivalences proven by \Hobbit, including for example all classical equivalences due to Meyer and Sieber~\cite{MeyerS88}.

Up-to techniques~\cite{PousS11} are specific to bisimulation and concern the reduction of the size of bisimulation relations, oftentimes turning infinite transition systems into finite ones by focusing on a core part of the relation.
Although extensively studied in the theory of bisimulation, up-to techniques have not been used in practice in an equivalence checker. 
We specifically propose three novel up-to techniques: \emph{up to separation} and \emph{up to re-entry} (\cref{sec:up-to}), dealing with infinity in the LTS due to the higher-order nature of the language,
and \emph{up to state invariants} (\cref{sec:up-to-gen}),
dealing with infinity due to state updates.
Up to separation allows us to reduce the knowledge of the context the examined program expressions are running in, similar to a frame rule in separation logic.
%
Up to re-entry removes the need of exploring unbounded nestings of higher-order function calls under specific conditions.
Up to state invariants allows us to abstract parts of the state and make finite the number of explored configurations by introducing state invariant predicates in configurations.

State invariants are common in equivalence proofs of stateful programs,
both in handwritten (e.g.~\cite{KoutavasW06}) and tool-based proofs. In the latter they are expressed manually in annotations (e.g.~\cite{reve}) or automatically inferred (e.g.~\cite{syteci}).
In \Hobbit\ we follow the manual approach, leaving heuristics for automatic invariant inference for future work.
An important feature of our annotations is the ability to express relations between the states of the two compared terms, enabled by the up to state invariants technique.
This leads to finite bisimulation transition systems in examples where concrete value semantics are infinite state.

%



The above technology, combined with standard up-to techniques, transform \Hobbit\
from a bounded checker into an equivalence prover able to reason about infinite behaviour in a finite manner in
a range of examples,
including classical example equivalences (e.g.~all in \cite{MeyerS88}) and some
that previous work on up-to techniques would cannot algorithmically decide~\cite{BiernackiLP19} (cf.~\cref{ex:meyer-sieber-e6}).
We have benchmarked \Hobbit\ on examples from the literature and newly designed ones (\cref{sec:imp}). 
Due to the undecidable nature of contextual equivalence, up-to techniques are not exhaustive:
no set of up-to techniques is guaranteed to finitise all examples. 
Indeed there are a number of examples where the bisimulation transition system is still infinite and \Hobbit\ reaches the exploration bound.
For instance, \Hobbit\ is not able to prove examples with inner recursion and well-bracketing properties, which we leave to future work.
Nevertheless, our approach provides a contextual equivalence tool for a higher-order language with state that can prove many equivalences and inequivalences which previous work could not handle due to syntactic restrictions and other limitations (\cref{sec:comp}).


\paragraph{Related work}
Our paper marries techniques from environmental bisimulations up-to~\cite{KoutavasW06,EnvBisim07,PousS11,BiernackiLP19} with the work on fully abstract game models for higher-order languages with state~\cite{Laird07,Dimovski14,LinT20}.
The closest to our technique is that of Biernacki et al.~\cite{BiernackiLP19}, which introduces up-to techniques for a similar symbolic LTS to ours, albeit with symbolic values restricted to higher-order types,
resulting in infinite LTSs in examples such as \cref{ex:swap},
and with inequivalence decided outside the bisimulation by (non-)termination, precluding the use up-to techniques in examples such as \cref{ex:meyer-sieber-e6}.
Close in spirit is the line of research on logical relations~\cite{AhmedDR09,HurDNV12,JaberTabareau15} which provides a powerful tool for hand-written proofs of contextual equivalence.
Also related are the tools \tool{Hector}~\cite{hector} and \tool{Coneqct}~\cite{coneqct}, and \tool{SyTeCi}~\cite{syteci}, { based on game semantics and step-indexed logical relations respectively} (cf.\ \cref{sec:comp}).

%% file: intuitions.tex
%
Contextual equivalence requires that two program expressions lead to the same observable result \emph{in any program context} these may be fed in. This quantification is hard to work with e.g.\ due to redundancy in program contexts.
Alternatively, we can translate programs into a semantic model that is \emph{fully abstract}, i.e.\ it assigns to program expressions the same denotation just if these are contextually equivalent. Thus doing, contextual equivalence is reduced to semantic equality. 

The semantic model we use is that of Game Semantics~\cite{Laird07}. We model programs as
formal interactions between two \emph{players}: a \emph{Proponent} (corresponding to the program) and an \emph{Opponent} (standing for any program context). Concretely, these interactions are
sets of traces produced from a Labelled Transition System (LTS),
the nodes and labels of which are called \emph{configurations} and \emph{moves} respectively. The LTS captures the interaction of the program with its environment, which is realised via function applications and returns: moves can be \emph{questions} (i.e.\ function applications) or \emph{answers} (returns), and belong to proponent or opponent. E.g.\ a program calling an external function will issue a proponent question, while the return of the external function will be an opponent answer. In the examples that follow, 
moves that correspond to the opponent shall be underlined.

 \begin{example}\label{ex:fst}
Consider the following expression of type $(\mathsf{unit}\to\mathsf{unit})\to\mathsf{int}$.

\centerline{
  $N =  \text{\lstinline{fun f -> f (); 0}}$
}\smallskip

\noindent
Evaluating $N$ leads to a function, call it $g$,  being returned (i.e.\ $g$ is $\lambda f.f();0$).
When $g$ is called with some input $f_1$, it will always return 0 but in the process it may call the external function $f_1$. The call to $f_1$ may immediately return or it may call $g$ again (i.e.\ reenter), and so on. The LTS for $N$ is as in \cref{fig:LTS} (top).
\end{example}

\begin{figure}[t]
  \begin{center}
\scalebox{0.7}{
\begin{tikzpicture}[node distance=2cm]
\node (s0) [pnode] {$N$};
\node (s1) [onode, right of=s0] {};
\node (s2) [pnode, right=1.75cm of s1] {$f_1();0$};
\node (s3) [onode, right=1.75cm of s2] {}; 
\node (s3b) [pnode, below=0.5cm of s2] {$0$};
\node (s4) [pnode, right=1.75cm of s3] 
	{$f_2();0$};
\node (s5) [onode, right=1.75cm of s4] 
	{};
\node (s5b) [pnode, below=0.5cm of s4] {$0$};
\node[draw=none] (end) [right=1.75cm of s5] {$\cdots$};

\draw [->] 
(s0) edge node[anchor=south]{$\pret{g}$} (s1)
(s1) edge node[anchor=south]{$\oapp{g}{f_1}$} (s2)
(s2) edge node[anchor=south]{$\papp{f_1}{()}$} (s3)
(s3) edge node[anchor=south]{$\oapp{g}{f_2}$} (s4)
(s3) edge node[anchor=east]{$\oret{()}$} (s3b)
(s4) edge node[anchor=south]{$\papp{f_2}{()}$} (s5)
(s5) edge node[anchor=east]{$\oret{()}$} (s5b)
(s3b) edge node[anchor=east]{$\pret{0}$} (s1)
(s5b) edge node[anchor=east]{$\pret{0}$} (s3)
(s5) edge node[anchor=south]{$\oapp{g}{f_3}$} (end)
;
\end{tikzpicture}}\smallskip
\hrule\smallskip
\scalebox{0.67}{
\begin{tikzpicture}[node distance=2cm]
\node (s0) [pnode] {$M$};
\node (s1) [onode, right of=s0] {};
\node (s2) [pnode, right=1.75cm of s1] {$x_1: 0$};
\node (s3) [onode, right=1.75cm of s2]
	{$x_1 : 0$};
\node (s3b) [pnode,below=0.5cm of s2] {$x_1:0$};
\node[draw=none] (s3c) [left=1.75cm of s3b] {$\cdots$};
\node (s4) [pnode, right=1.75cm of s3]
	{$x_1 : 0$\\$x_2: 0$};
\node (s5) [onode, right=1.75cm of s4]
	{$x_1 : 0$\\$x_2: 0$};
\node (s5b) [pnode,below=0.35cm of s4] {$x_1 : 0$\\$x_2: 0$};
\node[draw=none] (s5c) [left=1.75cm of s5b] {$\cdots$};
\node[draw=none] (end) [right=1.75cm of s5] {$\cdots$};

\draw [->] 
(s0) edge node[anchor=south]{$\pret{g}$} (s1)
(s1) edge node[anchor=south]{$\oapp{g}{f_1}$} (s2)
(s2) edge node[anchor=south]{$\papp{f_1}{()}$} (s3)
(s3) edge node[anchor=south]{$\oapp{g}{f_2}$} (s4)
(s3) edge node[anchor=east]{$\oret{()}$} (s3b)
(s4) edge node[anchor=south]{$\papp{f_2}{()}$} (s5)
(s5) edge node[anchor=east]{$\oret{()}$} (s5b)
(s5) edge node[anchor=south]{$\oapp{g}{f_3}$} (end)
(s3b) edge node[anchor=south]{$\pret{0}$} (s3c)
(s5b) edge node[anchor=south]{$\pret{0}$} (s5c)
;
\end{tikzpicture}
}\smallskip
\hrule\smallskip
\scalebox{0.6}{
\begin{tikzpicture}[node distance=2cm]
\node (s0) [pnode] {$N$};
\node (s1) [onode, right of=s0] {$C_1$};
\node (s2) [pnode, right=1.7cm of s1] {};
\node (s3) [onode, right=1.7cm of s2]
	{};
\node (s3b) [pnode,below=0.8cm of s2] {};
\node (s3c) [onode,left=1.7cm of s3b] {};
\node[draw=none] (s3d) [left=1.3cm of s3c] {};
\node[draw=none] (s4) [right=1.3cm of s3]
	{};

\draw [->] 
(s0) edge node[anchor=south]{$\pret{g}$} (s1)
(s1) edge node[anchor=south]{$\oapp{g}{f_1}$} (s2)
(s2) edge node[anchor=south]{$\papp{f_1}{()}$} (s3)
(s3) edge node[anchor=south]{$\oapp{g}{f_2}$} (s4)
(s3) edge node[anchor=east]{$\oret{()}$} (s3b)
(s3b) edge node[anchor=south]{$\pret{0}$} (s3c)
;
\draw [->,strike through] (s3c) -- node[anchor=south]{$\oapp{g}{f_2}$} (s3d);
\draw [->,strike through] (s3) -- node[anchor=south]{$\oapp{g}{f_2}$} (s4);
\end{tikzpicture}\vrule\quad
\begin{tikzpicture}[node distance=2cm]
\node (s0) [pnode] {$M$};
\node (s1) [onode, right of=s0] {$C_1$};
\node (s2) [pnode, right=1.7cm of s1] {$x_1: 0$};
\node (s3) [onode, right=1.7cm of s2]
	{$x_1 : 0$};
\node (s3b) [pnode,below=0.8cm of s2] {$x_1:0$};
\node (s3c) [onode,left=1.7cm of s3b] {$x_1:0$};
\node[draw=none] (s3d) [left=1.3cm of s3c] {};
\node[draw=none] (s4) [right=1.3cm of s3]
	{};

\draw [->] 
(s0) edge node[anchor=south]{$\pret{g}$} (s1)
(s1) edge node[anchor=south]{$\oapp{g}{f_1}$} (s2)
(s2) edge node[anchor=south]{$\papp{f_1}{()}$} (s3)
(s3) edge node[anchor=south]{$\oapp{g}{f_2}$} (s4)
(s3) edge node[anchor=east]{$\oret{()}$} (s3b)
(s3b) edge node[anchor=south]{$\pret{0}$} (s3c)
;
\draw [->,strike through] (s3c) -- node[anchor=south]{$\oapp{g}{f_2}$} (s3d);
\draw [->,strike through] (s3) -- node[anchor=south]{$\oapp{g}{f_2}$} (s4);
\end{tikzpicture}
}\vspace{-4mm}
\end{center}
\caption{Sample LTS's modelling expressions in Section~2. Note that in these diagrams we include in each configuration only part of their components, for better exposition.}\label{fig:LTS}
\end{figure}
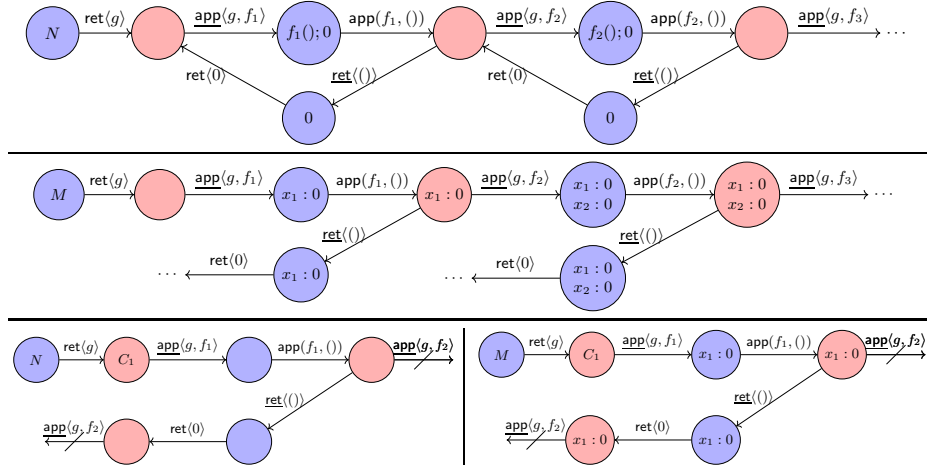

Given two expressions $M,N$, checking their equivalence will amount to checking bisimulation equivalence of their (generally infinite) LTS's.
%
Our checking routine performs a bounded analysis that aims to either find a finite counterexample and thus prove inequivalence, or build a bisimulation relation that shows the equivalence of the expressions. 
The former case is easier as it is relatively rapid to explore a bisimulation graph up to a given depth.
The latter one is harder, as the target bisimulation can be infinite. 
To tackle part of this infinity, we use three novel \textit{up-to techniques} for environmental bisimulation. 

Up-to techniques roughly assert that if a core set of configurations in the bisimulation graph explored can be proven to be part of a relation satisfying a definition that is more permissive than standard bisimulation, then a superset of configurations
forms a proper bisimulation relation.
This has the implication that a bounded analysis can be used to explore a finite part of the bisimulation graph to verify potentially infinitely many configurations. As there can be no complete set of up-to techniques, the pertaining 
question is how useful they are in practice. In the remainder of this section we present the first of our up-to techniques, called \textit{up to separation},  via an example equivalence.
The intuition behind this technique comes from Separation Logic and amounts to saying that functions that access separate regions of the state can be explored independently. As a corollary, a function that manipulates only its own local references may be explored independently of itself, i.e. it suffices to call it once. 

\begin{example}\label{ex:sep}
Consider the following pair of expressions, where $N$ is from \cref{ex:fst}.

  \centerline{
  $M = \text{\lstinline{fun f -> ref x = 0 in f (); !x}}$
  \qquad\qquad
  $N =  \text{\lstinline{fun f -> f (); 0}}$
}\smallskip

\noindent
  The LTS corresponding to $M$ and $N$ are shown in \cref{fig:LTS} (middle and top). Regarding $M$, we can see that opponent is always allowed to reenter the proponent function $g$, which creates a new reference $x_n$ each time. This makes each configuration unique, which prevents us from finding cycles and thus finitise  the bisimulation graph.
Moreover, both the LTS for $M$ and $N$ are infinite because of the stack discipline they need to adhere to when O issues reentrant calls. 

  With separation, however, we could prune the two LTS's as in \cref{fig:LTS} (bottom). 
We denote the configurations after the first opponent call as $C_1$.   
Any opponent call after $C_1$ leads to a configuration which differs from $C_1$ either by a state component that is not accessible anymore and can thus be {separated}, or by a stack component that can be similarly separated. Hence, the LTS's that we need to consider are finite and thus the expressions are proven equivalent.
\end{example}


%% file: lang.tex
\begin{figure*}[t]
  \footnotesize
  \[\begin{array}{r@{}r@{\,}c@{\,}l}
    \textsc{\Loc:} & l,k &
    \multicolumn{2}{l}{
      \hspace{3em}
    \textsc{\Var:}  x,y,z       
      \hspace{3em}
    \textsc{\Con:}  c 
    }\\
    \textsc{\Typ:} & T           & \mis & \Bool \mor \Int \mor \Unit \mor T \arrow T \mor T_1 * \ldots * T_n \\
      \textsc{\Exp: } & e,M,N           & \mis & v \mor (\vec e)\mor \arithop{\vec e} \mor \app e e \mor \cond{e}{e}{e} \mor \new l v e \mor \deref l \mor l \asgn e
 \mor \elet{(\vec x)}{e}{e} \\                                             
    \textsc{\Val:} & u,v       & \mis & c \mor x\mor \lam x e   \mor (\vec v)                                                                            \\ 
    \textsc{\EC:} & E        & \mis & \hole_T \morcondensed (\vec v,E,\vec e) \morcondensed \arithop{\vec v,E,\vec e} \morcondensed \app E e \morcondensed \app v E \morcondensed l \asgn E \morcondensed \cond E e e \morcondensed \elet{(\vec x)}{E}{e} \\
    \textsc{\Cxt:} & D           & \mis & \hole_{i,T} \mor e \mor (\vec D) \mor \arithop{\vec D} \mor \app D D \mor l \asgn D \mor \cond D D D \mor \lam x D
                    \\& & & \mor \new l D D  \mor \elet{(\vec x)}{D}{D}       \\
    \textsc{\St:} & s,t         & \in  & \Loc \overset{fin}{\rightharpoonup} \Val                                                               \\
  \end{array}\]
  \vspace{-1em}
  \[\begin{array}{l@{\;\,}l@{\;\,}ll}
    \redconf{s}{\arithop{\vec c}}    & \redbase & \redconf{s}{w}                      & \text{if } \mathop{op}^{\textsf{arith}}(\vec c) = w \\
    \redconf{s}{\app{(\lam[f] x e)}  v} & \redbase & \redconf{s}{e\sub{x}{v}\sub{f}{\lam[f] x e}}  \\
    \redconf{s}{\elet{(\vec x)}{(\vec v)}{e}} & \redbase & \redconf{s}{e\sub{\vec x}{\vec v}}  \\
    \redconf{s}{\new l v e}          & \redbase & \redconf{s\stupd{l}{v}}{e}          & \text{if } l \not\in\dom{s}\\
    \redconf{s}{\deref l}            & \redbase & \redconf{s}{v}                      & \text{if } s(l)=v \\
    \redconf{s}{l \asgn v}           & \redbase & \redconf{s\stupd{l}{v}}{()}   \\
    \redconf{s}{\cond{c}{e_1}{e_2}}  & \redbase & \redconf{s}{e_i}                    & \text{if } (c,i) \in \{(\true,1), (\false,2)\} \\
    \redconf{s}{E\hole[e]}           & \red     & \redconf{s'}{E\hole[e']}            & \text{if } \redconf{s}{e} \redbase \redconf{s'}{e'}
  \end{array}\]

  \hrule
  \caption{Syntax and reduction semantics of the language \lang.}\label{fig:lang}
\end{figure*}

We develop our technique for the language \lang, a simply typed lambda calculus with local state whose syntax and reduction semantics are shown in \cref{fig:lang}.
Expressions (\Exp) include the standard lambda expressions with recursive functions ($\lam x e$), together with location creation ($\new l v e$), dereferencing ($\deref l$), and assignment ($l \asgn e$), as well as standard base type constants ($c$) and operations ($\arithop{\vec e}$).
Locations are mapped to values, including function values, in a store (\St). We write \emptyS for the empty store
and let $\fl{\chi}$ denote the set of free locations in $\chi$.

The language \lang is simply-typed with typing judgements of the form $\typing{\Delta}{\Sigma}{e}{T}$, where $\Delta$ is a type environment (omitted when empty), $\Sigma$ a store typing and $T$ a value type (\Typ); $\Sigma_s$ is the typing of store $s$.
The rules of the type system are standard and omitted here (Appendix~\ref{sec:typing}). 
Values consist of boolean, integer, and unit constants, functions and arbitrary length tuples of values.
To keep the presentation of our technique simple we do not include reference types as value types, effectively keeping all locations local. 
Exchange of locations between expressions can be encoded using get and set functions. 
In \cref{ex:meyer-sieber-e6} we show the encoding of a classic equivalence with location exchange between expressions and their context.
Future work extensions to our technique to handle location types can be informed from previous work~\cite{Laird07,syteci}.

The reduction semantics is by small-step transitions between configurations containing a store and an expression, $\redconf{s}{e} \red \redconf{s'}{e'}$,
defined using single-hole evaluation contexts ($\EC$) over a base relation $\redbase$.
Holes $\hole_T$ are annotated with the type $T$ of closed values they accept, which we may omit to lighten notation.
Beta substitution of $x$ with $v$ in $e$ is written as $e\sub{x}{v}$.
We write $\redconf{s}{e}\trm$ to denote $\redconf{s}{e} \red^* \redconf{t}{v}$ for some $t$, $v$.
We write $\vec \chi$ to mean a syntactic sequence, and assume standard syntactic sugar from the lambda calculus.
In our examples we assume an ML-like syntax and implementation of the type system, which is also the concrete
syntax of \Hobbit.


We consider environments $\Gamma \in \Nat \xrightharpoonup{\mathsf{fin}} \Val$ which map natural numbers to closed values.
The concatenation of two such environments $\Gamma_1$ and $\Gamma_2$, written $\Gamma_1,\Gamma_2$ is defined when $\dom{\Gamma_1}\cap\dom{\Gamma_2}=\emptyset$.
We write $(\maps{i_1}{v_1}, \ldots, \maps{i_n}{v_n})$ for a concrete environment mapping $i_1, \ldots, i_n$ to $v_1, \ldots, v_n$, respectively.
When indices are unimportant we omit them and treat $\Gamma$ environments as lists.

General contexts $D$ contain multiple, non-uniquely indexed holes $\hole_{i,T}$, where $T$ is the type of value that can replace the hole. Notation $D\hole[\Gamma]$ denotes the context $D$ with each hole $\hole_{i,T}$ replaced with $\Gamma(i)$, provided that  $i\in\dom{\Gamma}$ and $\Sigma \vdash \Gamma(i):T$, for some $\Sigma$.
We omit hole types where possible and indices when all holes in $D$ are annotated with the same $i$. In the latter case we write $D\hole[v]$ instead of $D\hole[(\maps{i}{v})]$
and allow to replace all holes of $D$ with a closed expression $e$, written $D\hole[e]$.
We assume the Barendregt convention for locations, thus replacing context holes avoids location capture.
Standard contextual equivalence~\cite{Morris68} follows.

\begin{definition}[Contextual Equivalence]\label{def:cxt-equiv}
Expressions $\vdash e_1:T$ and $\vdash e_2:T$ are \emph{contextually equivalent}, written as $e_1 \cxteq e_2$, when for all contexts $D$ such that $\vdash D\hole[e_1]: \Unit$ and $\vdash D\hole[e_2]: \Unit$ we have
$
    \redconf{\emptyS}{D\hole[e_1]}\trm ~\text{iff}~
    \redconf{\emptyS}{D\hole[e_2]}\trm
$.
\end{definition}



%% file: lts.tex
  Our Labelled Transition System (LTS) has symbolic transitions for both higher-order and first-order transitions.
  For simplicity we first present our LTS with symbolic higher-order and concrete first-order transitions.
  We develop our theory and most up-to techniques on this simpler LTS.
  We then show its extension with symbolic first-order transitions and develop up to state invariants which relies on this extension.
%
  We extend the syntax with abstract function names $\alpha$:
$$
\textsc{\Val: } \quad u,v,w  \, \mis \, c \mor \lam x e\mor (\vec v) \mor \alpha
$$
We assume that $\alpha$'s are annotated by the type of function they represent, written $\alpha_{T\arrow T'}$, and omitted where possible;
$\an{\chi}$ is the set of abstract names in $\chi$.

\begin{figure*}[t] 

  \[\begin{array}{@{}lllll@{}}
    \irule![PropApp][propappf]{
      (D,\Gamma') \in \ulpatt(v)
    }{\nbox{
      \conf{A}{\Gamma}{K}{s}{E[\app {\alpha} v]}
      \trans{\lpropapp{\alpha}{D}} 
      \conf{A}{\Gamma,\Gamma'}{E\hole,K}{s}{\noe}
    }}
    \\
    \irule![PropRet][propretf]{
      (D,\Gamma') \in \ulpatt(v)
    }{
      \conf{A}{\Gamma}{K}{s}{v}
      \trans{\lpropret{D}}
      \conf{A}{\Gamma,\Gamma'}{K}{s}{\noe}
    }
    \\
     \irule![OpApp][opappf]{
       \nbox{
       \Sigma_s \vdash \Gamma(i) : T \arrow T'
       \\\text{and }
        (D,\vec \alpha) \in \ulpatt(T)
       \\\text{and }
        \app {\Gamma(i)} {D\hole[\vec \alpha]} \funred e
      }
    }{
      \conf{A}{\Gamma}{K}{s}{\noe} \trans{\lopapp{i}{D\hole[\vec \alpha]}} \conf{A\uplus\vec\alpha}{\Gamma}{K}{s}{e}
    }
    \\
    \irule![OpRet][opretf]{
        (D,\vec \alpha) \in \ulpatt(T)
    }{\\
      \conf{A}{\Gamma}{E\hole_T,K}{s}{\noe} \trans{\lopret{D\hole[\vec\alpha]}}
      \conf{A\uplus\vec\alpha}{\Gamma}{K}{s'}{E\hole[D[\vec\alpha]]}
    }
    \\
    \irule![Tau][tau]{
      \redconf{s}{e} \red \redconf{s'}{e'}
    }{
      \conf{A}{\Gamma}{K}{s}{e} \trans{\tau} \conf{A}{\Gamma}{K}{s}{e'}
    }
    \\
    \irule![Response][dummy]{
      \eta \not= \lterm
    }{
      C \trans{\eta} \botconf
    }
    \\
    \irule![Term][term]{}{
      \conf{A}{\Gamma}{\noe}{s}{\noe} \trans{\lterm} \botconf
    }

  \end{array}\]
  \hrule
  \caption{The Labelled Transition System.}\label{fig:lts}
\end{figure*}

We define our LTS (shown in \cref{fig:lts}) by opponent and proponent call and return transitions, based on Game Semantics~\cite{Laird07}.
Proponent transitions are the moves of an expression interacting with its context.
Opponent transitions are the moves of the context surrounding this expression.
These transitions are over proponent and opponent configurations $\conf{A}{\Gamma}{K}{s}{e}$ and $\conf{A}{\Gamma}{K}{s}{\cdot}$, respectively.
In these configurations:
\begin{itemize}
  \item $A$ is a set of abstract function names been used so far in the interaction;
  \item $\Gamma$ is an environment indexing proponent functions known to opponent;\footnote{thus, $\Gamma$ is encoding the environment of Environmental Bisimulations (e.g.~\cite{KoutavasW06})}
  \item $K$ is a stack of proponent continuations, created by nested proponent calls;
  \item $s$ is the store containing proponent locations;
  \item $e$ is the expression reduced in proponent configurations; $\hat e$ denotes $e$ or $\cdot$.
\end{itemize}
In addition, we introduce a special configuration $\botconf$ which is used in order to represent expressions that cannot perform given transitions (cf.~Remark~\ref{rem:bots}). We let a \emph{trace} be a sequence of app and ret moves (i.e.\ labels), as defined in  \cref{fig:lts}.


For the LTS to provide a fully abstract model of the language, it is necessary that functions which are passed as arguments or return  values from proponent to opponent be abstracted away, as the actual syntax of functions is not directly observable in \lang. This is achieved by deconstructing such values $v$ to:
\begin{itemize}
\item an \emph{ultimate pattern} $D$ (cf.~\cite{LassenL07}), which is a context obtained from $v$ by replacing each function in $v$ with a distinct numbered hole; together with
\item an environment $\Gamma$ whose domain is the indices of these holes, and $D\hole[\Gamma]=v$.
\end{itemize}
We let $\ulpatt(v)$ contain all such pairs $(D,\Gamma)$ for $v$;
e.g.:
$
  \ulpatt((\lambda x. e_1, 5)) = \{(\,(\hole_i, 5),\,[\maps{i}{\lambda x. e_1}]\,)
  \where \text{for any } i\}.
$

Ultimate pattern matching is extended to types through the use of symbolic function names:
$\ulpatt(T)$ is the largest set of pairs $(D,{\vec \alpha})$ such that $\vdash D\hole[{\vec \alpha}] : T$, where $\vec\alpha$ is an environment with indices omitted, and $D$ does not contain functions.

In  \cref{fig:lts}, proponent application and return transitions (\iref{propappf}, \iref{propretf}) use ultimate pattern matching for values and accumulate the functions generated by the proponent in the $\Gamma$ environment of the configuration, leaving only their indices on the label of the transition itself.
Opponent application and return transitions (\iref{opappf}, \iref{opretf}) use ultimate pattern matching for types to generate opponent-generated values which can only contain abstract functions.
This eliminates the need for quantifying over all functions in opponent transitions but still includes infinite quantification over all base values. Symbolic first-order values in \cref{sec:sym} will obviate the latter.

At opponent application the following preorder performs a beta reduction when opponent applies a concrete function.
This technicality is needed for soundness.

\begin{definition}[\funred]
  For application $\app v u$ we write $\app v u \funred e$ to mean
  $e = \app \alpha u$, when $v=\alpha$; and $e=e'\sub{x}{u}\sub{f}{\lam x e}$, when $v=\lam x e'$.
\end{definition}
In our LTS, $C$ ranges over configurations and $\eta$ over transition labels;
$\wtrans{\eta}$ means ${\trans{\tau}}{}^*$, when $\eta=\tau$, and $\wtrans{\tau}\trans{\eta}\wtrans{\tau}$ otherwise.
Standard weak (bi-)simulation follows.
\begin{definition}[Weak Bisimulation]
  Binary relation \bisim{R} is a \emph{weak simulation} when
  for all $C_1 \bisim*{R} C_2$ and $C_1 \trans{\eta} C_1'$, there exists $C_2'$
  such that $C_2 \wtrans{\eta} C_2'$ and $C_1' \bisim*{R} C_2'$.
  If \bisim{R}, $\bisim{R}^{-1}$ are weak simulations then
  \bisim{R} is a \emph{weak bisimulation}.
  Similarity $(\simil)$ and bisimilarity $(\bisimil)$ are the largest weak 
  simulation and bisimulation, respectively.
  \defqed
\end{definition}


\begin{remark}\label{rem:bots}
  Any proponent configuration that cannot match a standard bisimulation transition challenge can trivially respond to the challenge by transitioning into $\botconf$ by the \iref{dummy} rule in \cref{fig:lts}.
By the same rule, this configuration can trivially perform all transitions except a special termination transition, labelled with $\lterm$.
%
However, regular configurations that have no pending proponent calls ($K=\cdot$), can perform the special termination transition (\iref{term} rule), signalling the end of a \emph{complete trace}, i.e.\ a completed computation.
This mechanism allows us to encode complete trace equivalence, which coincides with contextual equivalence~\cite{Laird07}, as bisimulation equivalence.
In a bisimulation proof,
  if a proponent configuration is unable to match a bisimulation transition with a regular transition, it can still transition to \botconf where it can simulate every transition of the other expression, apart from \trans{\lterm} leading to a complete trace. 
%

Our mechanism for treating unmatched transitions has the benefit of enabling us to use the standard definition of bisimulation over our LTS.
This is in contrast to previous work \cite{BiernackiLP19,JaberTabareau15}, where termination/non-termination needed to be proven independently or baked in the simulation conditions.
More importantly, our approach  allows us to use bisimulation up-to techniques even when 
one of the related configurations diverges, which is not possible in previous symbolic LTSs \cite{Laird07,JaberTabareau15,BiernackiLP19}, and is necessary in examples such as \cref{ex:meyer-sieber-e6}.
\end{remark}
\begin{technical} 
\begin{lemma}
  Suppose $e_1$, $e_2$ are doubly closed expressions. Then
  \[
  \conf{\emptyA}{\emptyG}{\emptyK}{\emptyS}{e_1} \simil \conf{\emptyA}{\emptyG}{\emptyK}{\emptyS}{e_2}
  \qquad\text{iff}\qquad
  \conf{\emptyA}{\emptyG}{\emptyK}{\emptyS}{\lam \_ e_1} \simil \conf{\emptyA}{\emptyG}{\emptyK}{\emptyS}{\lam \_ e_2}
  \]
\end{lemma}
\begin{proof}
  We start by showing that
  \[
  \conf{\emptyA}{\emptyG}{\emptyK}{\emptyS}{\lam \_ e_1}
  \simil
  \conf{\emptyA}{\emptyG}{\emptyK}{\emptyS}{\lam \_ e_2}
  \quad\text{iff}\quad
  \conf{\emptyA}{\maps{i}{\lam \_ e_1}}{\emptyK}{\emptyS}{\noe}
  \simil
  \conf{\emptyA}{\maps{i}{\lam \_ e_2}}{\emptyK}{\emptyS}{\noe}
\]
  This is straightforward because 
  $\conf{\emptyA}{\emptyG}{\emptyK}{\emptyS}{\lam \_ e_k} \trans{\lpropret{j}}
  \conf{\emptyA}{\maps{j}{\lam \_ e_k}}{\emptyK}{\emptyS}{\noe}$ ($k=1,2$ and any $j$)
  is the only non-\iref{dummy} transition from this configuration and the resulting configurations are related in
  $\utperm{\simil}$.
  It thus suffices to show
  \[
  \conf{\emptyA}{\emptyG}{\emptyK}{\emptyS}{e_1} \simil \conf{\emptyA}{\emptyG}{\emptyK}{\emptyS}{e_2}
  \qquad\text{iff}\qquad
  \conf{\emptyA}{\maps{i}{\lam \_ e_1}}{\emptyK}{\emptyS}{\noe}
  \simil
  \conf{\emptyA}{\maps{i}{\lam \_ e_2}}{\emptyK}{\emptyS}{\noe}
  \]

  We prove the left-to-right direction 
  by constructing
  \[
      \bisim R \defeq \{(\conf{\emptyA}{\maps{i}{\lam \_ e_1}}{\emptyK}{\emptyS}{\noe},
            \conf{\emptyA}{\maps{i}{\lam \_ e_2}}{\emptyK}{\emptyS}{\noe})\}
  \]
  and showing \bisim{R} is a weak bisimulation up to separation, garbage collection,
  and similarity; i.e.,
\[
  \bisim{R} \progress (\utsepconj{} \sqcup \utgc{}
  \sqcup \constf{\simil})^{\omega}(\bisim R)
\]
  Let $C_1 \bisim*{R} C_2$ and $C_1 \trans{\eta} C_1'$.
  A transition from $C_1$ derived by the \iref{dummy} rule can be easily matched with the same transition from $C_2$ and the residual configurations are related in $\constf{\simil}(\bisim{R})$.
  The only other possibilities for the transition is to be derived from the rules \iref{opappc} and \iref{opappf}.
  Both cases are almost identical and we thus show only the latter where
  $\eta=\lopapp{i}{\alpha}$ and 
  $C_1' = \conf{\alpha}{\maps{i}{\lam \_ e_2}}{\emptyK}{\emptyS}{e_1}$.
  This transition can be matched by
  $C_2 \trans{\lopapp{i}{\alpha}} \conf{\alpha}{\maps{i}{\lam \_ e_2}}{\emptyK}{\emptyS}{e_2}=C_2'$.
  Note that $\alpha\not\in\an{e_1,e_2}$ as they are closed expressions.
  To relate the configurations requires the separation and garbage collection up-to techniques:
\[
  \begin{array}{@{}cl@{}}
    \multicolumn{2}{@{}c@{}}{
  C_3=\conf{\empty}{\emptyG}{\emptyK}{\emptyS}{e_1}
    \mathrel{\constf{\simil}(\bisim R)}
  \conf{\alpha}{\emptyG}{\emptyK}{\emptyS}{e_2}=C_4
    \qquad
  C_5=\conf{\emptyA}{\maps{i}{\lam \_ e_1}}{\emptyK}{\emptyS}{\noe}
    \bisim*{R}
  \conf{\emptyA}{\maps{i}{\lam \_ e_2}}{\emptyK}{\emptyS}{\noe}=C_6
}
  \\
    \timplies&
    C_3 \sepconj C_5=
  \conf{\emptyA}{\maps{i}{\lam \_ e_2}}{\emptyK}{\emptyS}{e_1}
    \utsepconj*{(\constf{\simil}\sqcup\utId{})(\bisim{R})}
  \conf{\emptyA}{\maps{i}{\lam \_ e_2}}{\emptyK}{\emptyS}{e_2}
    =C_4 \sepconj C_6
  \\
    \timplies&
    C_1'=
  \conf{\alpha}{\maps{i}{\lam \_ e_2}}{\emptyK}{\emptyS}{e_1}
    \utgc*{\utsepconj*{(\constf{\simil}\sqcup\utId{})(\bisim{R})}}
  \conf{\alpha}{\maps{i}{\lam \_ e_2}}{\emptyK}{\emptyS}{e_2}
  =C_2'
  \end{array}
\]

In the right-to-left direction we show the relation
      $\bisim{S}\defeq\{(\conf{\emptyA}{\emptyG}{\emptyK}{\emptyS}{e_1}, \conf{\emptyA}{\emptyG}{\emptyK}{\emptyS}{e_2})\}$
  to be a simulation up to weakening, garbage collection and similarity; i.e.,
  $\bisim{S} \progress (\utweak{}\sqcup\utgc{}
  \sqcup \constf{\simil})^\omega(\bisim R)$.
\end{proof}
\end{technical} 

\begin{definition}[Bisimilar Expressions]
Expressions $\vdash e_1: T$ and $\vdash e_2: T$ are bisimilar, written $e_1 \bisimil e_2$, when 
  $\conf{\emptyA}{\emptyG}{\emptyK}{\emptyS}{e_1} \bisimil
  \conf{\emptyA}{\emptyG}{\emptyK}{\emptyS}{e_2}$.
  \defqed
\end{definition}


\begin{theorem}[Soundness and Completeness]\label{theorem:SC}
  $e_1 \bisimil e_2$ iff $e_1 \cxteq e_2$ (see proof in \cref{sec:sound-complete,sec:soundness,sec:completeness}).
\end{theorem}
As a final remark, the LTS presented in this section is finite state only for a small number of trivial equivalence examples, such as the following one.
\begin{example}
  The following two implementations of conjunction ($\Bool * \Bool \arrow \Bool$)
  have a finite transition system according to the rules of \cref{fig:lts}.

  $M = \text{\lstinline{fun xy -> let (x,y) = xy in if x then y else false}}$

  $N =  \text{\lstinline{fun xy -> let (x,y) = xy in x && y}}$
\end{example}
However, even simple modifications to this example, such as allocation of a fresh location within one of the functions,
leads to infinite state transition systems. The following section addresses multiple sources of infinity in the transition systems through bisimulation up-to techniques.


%% file: upto.tex
We start by the definition of a sound up-to technique.
\begin{definition}[Weak Bisimulation up to $f$]
  \bisim{R} is a \emph{weak simulation up to $f$} when
  for all $C_1 \bisim*{R} C_2$ and $C_1 \trans{\eta} C_1'$, there is $C_2'$
  with $C_2 \wtrans{\eta} C_2'$ and $C_1' \mathrel{f(\bisim*{R})} C_2'$.
  If \bisim{R}, $\bisim{R}^{-1}$ are weak simulations up to $f$ then
  \bisim{R} is a \emph{weak bisimulation up to $f$}.
\end{definition}
\begin{definition}[Sound up-to technique]
  A function $f$ is a \emph{sound up-to technique} when for any $\bisim{R}$ which is a simulation up to $f$ we have $R \subseteq (\simil)$.
\end{definition}

\Hobbit\ employs the standard techniques: 
up to identity, up to garbage collection, up to beta reductions and up to name permutations
(see Appendix~\ref{app:simple}). Here we present two novel up-to techniques: up to separation and up to reentry.

\subsubsection{Up to Separation}

Our experience with \Hobbit\ has shown that one of the most effective up-to techniques for finitising bisimulation transition systems is the novel \emph{up to separation}
which we propose here.
The intuition of this technique is that if different functions operate on disjoint parts of the store, they can be explored in disjoint parts of the bisimulation transition system.
Taken to the extreme, a function that does not contain free locations can be applied only once in a bisimulation test as two copies of the function will not interfere with each other, even if they allocate new locations after application.
To define up to separation we need to define a separating conjunction for configurations.

\begin{definition}[Stack Interleaving]
  Let $K_1$, $K_2$ be lists of evaluation contexts from \EC (\cref{fig:lang}); we define the interleaving operation $K_1 \interleave[\vec k] K_2$ inductively,
  and write $K_1 \interleave K_2$ to mean $K_1 \interleave[\vec k] K_2$ for unspecified $\vec k$. We let \
$\emptyK \interleave[\emptyK] \emptyK = \emptyK$ \  and:
\begin{align*}
  E_1,K_1 \interleave[(1,\vec k)] K_2 &= E_1,(K_1 \interleave[\vec k] K_2) 
  &\quad
    K_1 \interleave[(2,\vec k)] E_2,K_2 &= E_2,(K_1 \interleave[\vec k] K_2)\,.
  \end{align*}
\end{definition}

\begin{definition}[Separating Conjuction]
  Let
  $C_1 =\conf{A_1}{\Gamma_1}{K_1}{s_1}{\hat e_1}$ and
  $C_2 =\conf{A_2}{\Gamma_2}{K_2}{s_2}{\hat e_2}$ be well-formed configurations.
  We define:
  \begin{itemize}
    \item$
    C_1 \sepconj[1][\vec k] C_2 \defeq \conf{A_1\cup A_2}{\Gamma_1,\Gamma_2}{K_1\interleave[\vec k] K_2}{s_1, s_2}{\hat e_1}
    $ when $\hat e_2 = \noe$
  \item$
    C_1 \sepconj[2][\vec k] C_2 \defeq \conf{A_1\cup A_2}{\Gamma_1,\Gamma_2}{K_1\interleave[\vec k] K_2}{s_1, s_2}{\hat e_2}
    $ when $ \hat e_1 = \noe$
  \end{itemize}
  provided
  $\dom{s_1}\cap\dom{s_2} = \emptyset$.
  We let $C_1 \sepconj C_2$ denote $\exists i,\vec k \st C_1 \sepconj[i][\vec k] C_2$.
\end{definition}

The function $\utsepconj{}$ provides the up to separation technique; it is defined as:
%
  \[\begin{array}{@{}c@{}}
    \irule*[UpTo$\sepconj$][uptosepconj]{
      C_1
      \bisim*{R}
      C_2
      \\
      C_3
      \bisim*{R}
      C_4
    \mytechnical{
      \\\\
      \dom{C_1.\Gamma}=\dom{C_2.\Gamma}
      \\\\
      \dom{C_3.\Gamma}=\dom{C_4.\Gamma}
    }
    }{
      C_1 \sepconj[i][\vec k] C_3 \utsepconj*{\bisim{R}} C_2 \sepconj[i][\vec k] C_4
    }
    \quad
    \irule*[UpTo${\sepconj}\bot_L$][uptosepconjbotl]{
      C_1 \bisim*{R} \botconf
      \\
      C_3 \bisim*{R} C_4
    }{
      C_1 \sepconj C_3 \utsepconj*{\bisim{R}} \botconf
    }
    \quad
    \irule*[UpTo${\sepconj}\bot_R$][uptosepconjbotr]{
      C_1 \bisim*{R}  C_2
      \\
      C_3 \bisim*{R}\botconf
    }{
      C_1 \sepconj C_3 \utsepconj*{\bisim{R}} \botconf
    }
   \end{array}\]
Its soundness follows by extending \cite{PousS11,PousCompanion} with 
a more powerful proof obligation
\begin{technical}
  $\utsepconj{}\comp{\constf{\gfp\WP}}\progress(\utsepconj{}\cup g)^\omega\comp{\constf{\gfp\WP}}$
  where $g=\utId{}\sqcup\utperm{}$
\end{technical}
  (see \cref{sec:uptosep-proof}).
  
\begin{lemma}\label{lem:utsepconj-sound}
  Function \utsepconj{} is a sound up-to technique.
\end{lemma}
\begin{exclproof}
  By \cref{lem:prfs}~(\ref{lem:prfs-3}) it suffices to show
  $\utsepconj{}\comp{\constf{\gfp\WP}}\progress(\utsepconj{}\cup g)^\omega\comp{\constf{\gfp\WP}}$,
  where $g=\utId{}\sqcup\utperm{}$.
  We need to show
  $\utsepconj{}\comp{\constf{\gfp\WP}}\comp\WP(\bisim{R}) \subseteq 
  {\WP}\comp (\utsepconj{}\cup g)^\omega\comp{\constf{\gfp\WP}}(\bisim{R})$.
  Because $\constf{\gfp\WP}$ is the constant function mapping its argument to $(\simil)$, we need to show
  $\utsepconj{}(\simil) \subseteq 
  {\WP}\comp (\utsepconj{}\cup g)^\omega(\simil)$.
  We proceed by induction on the derivation of $C_1 \utsepconj*{\simil} C_2$.

  Let $C_1\utsepconj*{\simil}C_2$, 
  and $C_1\trans{\eta}C_1'$.
  If $C_1\trans{\eta}C_1'$ is produced by rule \iref{dummy} then the proof is trivial as $C_2$ can perform the same transition and $C_1'=\botconf \utId*{\simil} \botconf=C_2'$. We thus consider only non-\iref{dummy} transitions from $C_1$.
  We proceed by case analysis on the derivation
  $C_1\utsepconj*{\simil}C_2$. There are three cases: \iref{uptosepconj}, \iref{uptosepconjbotl} and \iref{uptosepconjbotr}.
  We only show the proof for the first case; the last two cases are handled in a similar but simpler fashion.

  W.l.o.g. we let $i=1$ and we have
  \begin{align*}
    C_1 &= C_3 \sepconj[i][\vec k] C_5 = \conf{A_3\cup A_5}{\Gamma_3,\Gamma_5}{K_1\interleave[\vec k]K_3}{s_1, s_3}{\hat e_3}\\
    C_2 &= C_4 \sepconj[i][\vec k] C_6 = \conf{A_3\cup A_5}{\Gamma_4,\Gamma_6}{K_2\interleave[\vec k]K_4}{s_2, s_4}{\hat e_4}\\
    C_3 &= \conf{A_3}{\Gamma_3}{K_3}{s_3}{\hat e_3} \simil \conf{A_3}{\Gamma_4}{K_4}{s_4}{\hat e_4} = C_4
    &\dom{\Gamma_3}&=\dom{\Gamma_4} \\
    C_5 &= \conf{A_5}{\Gamma_5}{K_5}{s_5}{\hat e_5} \simil \conf{A_5}{\Gamma_6}{K_6}{s_6}{\hat e_6} = C_6
    &\dom{\Gamma_5}&=\dom{\Gamma_6}
  \end{align*}
  We proceed by cases on the transition $C_1\trans{\eta}C_1'$.

  {\flushleft
  \iref{propappc},
  \iref{propappf},
  \iref{propretc},
  \iref{propretf},
  \iref{tau}:}
  In all proponent transitions the proof is similar. We show the case \iref{propappf}.
  In this case $e_3=E\hole[\app{\alpha} v]$ and
  \begin{align*}
    C_1 = \conf{A_3\cup A_5}{\Gamma_3,\Gamma_5}{K_3\interleave[\vec k]K_5}{s_3, s_5}{E\hole[\app\alpha v]}
    &\trans{\lpropapp{\alpha}{i}}
    \conf{A_3\cup A_5}{\Gamma_3,\Gamma_5,\maps{i}{v}}{K_3\interleave[\vec k]K_5}{s_3, s_5}{\noe}
    =C_1'
    \\&\qquad\qquad\text{ and }
    i\not\in\dom{\Gamma_3,\Gamma_5}
    \\
    C_3 = \conf{A_3}{\Gamma_3}{K_3}{s_3}{E\hole[\app\alpha v]}
    &\trans{\lpropapp{\alpha}{i}}
    \conf{A_3}{\Gamma_3,\maps{i}{v}}{K_3}{s_3}{\cdot} = C_3'
  \end{align*}
  Moreover, $C_1'=C_3'\sepconj[1][\vec k]C_5$.
  By $C_3 \simil C_4$ there are two possibilities:
  \begin{itemize}
    \item $C_4 \wtrans{\lpropapp{\alpha}{i}} \botconf$ and $C_3' \simil \botconf$.
      The proof in this case is completed by $C_2' \trans{\lpropapp{\alpha}{i}} \botconf$ and
      $C_1 \utsepconj*{\simil} \botconf$ by either rule \iref{uptosepconjbotl} or \iref{uptosepconjbotr}.

    \item
    $C_4 = \conf{A_3}{\Gamma_4}{K_4}{s_4}{\hat e_4}
    \wtrans{\tau}\trans{\lpropapp{\alpha}{i}}
    \conf{A_3}{\Gamma_4,\maps{i}{v_4}}{K_4}{s_4'}{\noe} = C_4'$
      and
    $C_3' \simil C_4'$. By
    $\dom{\Gamma_3,\Gamma_5}=\dom{\Gamma_4,\Gamma_6}$ and \cref{lem:lts-taus} we derive:
    \[
      C_2 = \conf{A_3\cup A_5}{\Gamma_4,\Gamma_6}{K_4\interleave[\vec k]K_6}{s_4, s_6}{\hat e_4}
      \wtrans{\lpropapp{\alpha}{i}}
      \conf{A_3\cup A_5}{\Gamma_4,\Gamma_6,\maps{i}{v_4}}{K_4\interleave[\vec k]K_6}{s_4', s_6}{\noe}
      =C_2'
    \]
   Moreover,
   $C_2'=C_4'\sepconj[1][\vec k]C_6$ and
   $C_1' \utsepconj*{\simil} C_2'$.
  \end{itemize}

{\flushleft
\iref{opappc},\iref{opappf}:}
We show the case for \iref{opappf}. In this case $\eta=\lopapp{i}{\alpha}$, and  $\hat e_3=\noe$, and by \cref{lem:simil-properties}, $\hat e_4=\noe$.
Moreover $i\in\dom{\Gamma_3}$ or $i\in\dom{\Gamma_5}$.
If it is the former, then $C_3 \trans{\lopapp{i}{\alpha}} C_3'$ and $C_1'=C_3' \sepconj[1][\vec k] C_5$, otherwise
$C_5 \trans{\lopapp{i}{\alpha}} C_5'$ and $C_1'=C_3 \sepconj[2][\vec k] C_5'$. Moreover by the simulation we can show that
$C_4 \wtrans{\lopapp{i}{\alpha}} C_4'$ or
$C_6 \wtrans{\lopapp{i}{\alpha}} C_6'$, respectively.
In both cases $C_2 \wtrans{\lopapp{i}{\alpha}} C_2'$ and $C_2'=C_4'\sepconj[1][\vec k] C_6$ or
 $C_2'=C_4\sepconj[2][\vec k] C_6'$, respectively.
 Therefore $C_1' \utsepconj*{\simil} C_2'$.

{\flushleft
\iref{opretc},\iref{opretf}:}
We show the cases for \iref{opretf}. Here $\vec k$ is either $(1,\vec k')$ or $(2, \vec k')$. We consider the former case, the latter is symmetric. In this case $K_3 = E, K_3'$ and $K_3\interleave[\vec k]K_5= E,(K_3'\interleave[\vec k'] K_5)$. Therefore we have:
  \begin{align*}
    C_1 = \conf{A_3\cup A_5}{\Gamma_3,\Gamma_5}{E,(K_3'\interleave[\vec k'] K_5)}{s_3, s_5}{\noe}
    &\trans{\lopret{\alpha}}
    \conf{A_3\cup A_5\uplus \alpha}{\Gamma_3,\Gamma_5}{K_3'\interleave[\vec k']K_5}{s_3, s_5}{E\hole[\alpha]}
    =C_1'
    \\&\qquad\qquad\text{ and }
    \alpha\not\in\dom{A_3,A_5}
    \\
    C_3 = \conf{A_3}{\Gamma_3}{K_3}{s_3}{\noe}
    &\trans{\lopret{\alpha}}
    \conf{A_3\uplus\alpha}{\Gamma_3}{K_3}{s_3}{\cdot} = C_3'
  \end{align*}
  By $C_3 \simil C_4$ there are two possibilities:
  \begin{itemize}
    \item $C_4 \wtrans{\lopret{\alpha}} \botconf$ and $C_3' \simil \botconf$.
      The proof in this case is completed by $C_2' \trans{\lopret{\alpha}} \botconf$ and
      $C_1 \utsepconj*{\simil} \botconf$ by either rule \iref{uptosepconjbotl} or \iref{uptosepconjbotr}.

    \item
      $C_4 = \conf{A_3}{\Gamma_4}{K_4}{s_4}{\noe}
    \wtrans{\tau}\trans{\lopret{\alpha}}
      \conf{A_3\uplus\alpha}{\Gamma_4,\maps{i}{v_4}}{K_4'}{s_4'}{E_4\hole[\alpha]} = C_4'$
      and
    $C_3' \simil C_4'$
      and $K_4= E_4,K_4'$.
      Thus $K_4\interleave[\vec k]K_6=K_4\interleave[1,\vec k']K_6=E_4,(K_4'\interleave[\vec k']K_6)$.
      By \cref{lem:lts-taus} we derive:
    \[
      C_2 = \conf{A_3\cup A_5}{\Gamma_4,\Gamma_6}{E_4,(K_4'\interleave[\vec k']K_6)}{s_4, s_6}{\noe}
      \wtrans{\lpropapp{\alpha}{i}}
      \conf{A_3\cup A_5\uplus\alpha}{\Gamma_4,\Gamma_6}{K_4'\interleave[\vec k']K_6}{s_4', s_6}{E_4\hole[\alpha]}
      =C_2'
    \]
   Moreover,
   $C_2'=C_4'\sepconj[1][\vec k]C_6$ and
   $C_1' \utsepconj*{\simil} C_2'$.
  \end{itemize}

{\flushleft
\iref{tau}:}
In this case we have 
  \begin{align*}
    C_1 = C_3 \sepconj[1][\vec k] C_5
    &\trans{\tau}
    C_3' \sepconj[1][\vec k] C_5 = C_1'
    \\
    C_3 &\trans{\tau} C_3'
  \end{align*}
  By $C_3 \simil C_4$ we have
    $C_4 \wtrans{\tau} C_4'=\botconf$
    or
    $C_4 \wtrans{\tau} C_4'\not=\botconf$
    and
    $C_3' \simil C_4'$.
    In the former case the proof is completed by $C_2 \trans{\tau}\botconf$ and
    $C_1' \utsepconj*{\simil} \botconf$ by rule \iref{uptosepconjbotl}.
    In the latter case, by \cref{lem:lts-permute-weak},
  there exists $\pi_{l4}$ such that $C_4 \wtrans{\tau} C_4'\pi_{l4}$ and
  $(\fl{C_4'\pi_{l2}}\sdif\fl{C_4})\cap \fl{C_6} = \emptyset$.
  We have $C_3'\utperm*{\simil}C_4'\pi_{l4}$. Moreover we derive
  $C_2 = C_4 \sepconj[1][\vec k] C_6 \wtrans{\tau} C_4'\pi_{l4} \sepconj[1][\vec k] C_6 = C_2'$
  and $C_1' \utsepconj*{\utperm{\simil}} C_2'$.

{\flushleft \iref{term}:}
In this case we have $K_3=K_5=\emptyK$ and $\hat e_3=\hat e_5=\noe$. Therefore $C_3 \trans{\lterm} \botconf$ and
 $C_5 \trans{\lterm} \botconf$. Therefore by $C_3 \simil C_4$ and $C_5 \simil C_6$, we have
 $C_4 \trans{\lterm} \botconf$ and $C_6 \trans{\lterm} \botconf$. Therefore $K_4=K_6=\emptyK$ and
 $\hat e_4 =\hat e_6=\noe$, and thus $C_2 \trans{\lterm}\botconf$. The resulting \botconf configurations are related by $\utId{\simil}$.
\end{exclproof}
Many example equivalences have a finite transition system when using up to separation in conjunction with the simple techniques of the preceding section.

\begin{example}
  The following is a classic example equivalence from Meyer and Sieber \cite{MeyerS88}.
  The following expressions
  are equivalent at type $(\Unit\arrow\Unit)\arrow\Unit$.\smallskip
  
  \centerline{
  $M = \text{\lstinline{fun f -> ref x = 0 in f ()}}$
  \qquad\qquad
  $N =  \text{\lstinline{fun f -> f ()}}$
}\smallskip

\noindent
  After initial application of the function by the opponent, the proponent calls \lstinline{f}, growing the stack $K$ in the two configurations.
  At that point the opponent can apply the same functions again. The LTS of
  this example is thus infinite because $K$ can grow indefinitely. It is
  additionally infinite because the opponent can keep applying the initial function applications even after these return. However, if we apply the
  up-to separation technique immediately after the first opponent application,
  the $\Gamma$ environments become empty, and thus no second
  application of the same functions can happen. The LTS thus becomes trivially
  small. Note that no other up to technique is needed here.
\end{example}
\begin{example}\label{ex:bohrb}
  This example is due to Bohr and Birkedal~\cite{BohrB06} which includes a non-synchronised divergence.

    \begin{tabular}{ll}
      $M =\;$
      &\begin{minipage}[t]{0.7\textwidth}\vspace{-1.3em} \begin{lstlisting}[boxpos=t]
fun f -> 
  ref l1 = false in ref l2 = false in
  f (fun () -> if !l1 then _bot_ else l2 := true);
  if !l2 then _bot_ else l1 := true
\end{lstlisting}
    \end{minipage}
\\
      $N =\;$ 
      &\begin{minipage}[t]{0.7\textwidth}\vspace{-1.3em} \begin{lstlisting}[boxpos=t]
fun f -> f (fun () -> _bot_)
\end{lstlisting}
    \end{minipage}
    \end{tabular}
\\
Note that \lstinline{_bot_} is a diverging computation.
    This is a hard example to prove using environmental bisimulation even with up to techniques; requiring
    quantification over contexts within the proof. However, with up-to separation after the opponent
    applies the two functions, the $\Gamma$ environments are emptied, thus leaving only one application of $M$ and $N$ that needs to be explored by the bisimulation.
    Applications of the inner function provided as argument to \lstinline{f}
    only leads to a small number of reachable configurations. \Hobbit\ can indeed
    prove this equivalence. 
\end{example}

\subsubsection{Up to Proponent Function Re-entry}
\label{sec:reentry}

\begin{figure*}[t] 

  \[\begin{array}{@{}c@{}}
    \irule*[UpToReentry][uptoreentrance]{
      C_1 = \conf{A}{\Gamma_1}{K_1}{s_1}{\cdot}
      \bisim*{R}
      \conf{A}{\Gamma_2}{K_2}{s_2}{\cdot} = C_2
      \\\\
      \forall \vec \eta, C, A', \Gamma_1', \Gamma_2', s_1', s_2'.~\nbox{\big[
      (\lopapp{i}{\_} \not\in \{\vec \eta\} \text{ and }
      \\\conf{A}{\Gamma_1}{\emptyK}{s_1}{\cdot}
      \trans{\lopapp{i}{C}}\trans{\vec \eta}\asymp 
      \conf{A'}{\Gamma_1'}{\emptyK}{s_1'}{\cdot}
      \text{ and } 
      \\\conf{A}{\Gamma_2}{\emptyK}{s_2}{\cdot}
      \trans{\lopapp{i}{C}}\trans{\vec \eta}\asymp 
      \conf{A'}{\Gamma_2'}{\emptyK}{s_2'}{\cdot}
      \\\text{ implies } 
      \Gamma_1'=\Gamma_1
      \text{ and } 
      \Gamma_2'=\Gamma_2
      \text{ and } 
      s_1 = s_1'
      \text{ and } 
      s_2 = s_2'\big]}
      \\\\
      C_1 \trans{\lopapp{i}{C}}\trans{\vec \eta'}\trans{\lopapp{i}{C'}} \conf{A'}{\Gamma_1}{K_1',K_1}{s_1}{e_1'}
      \\\\
      C_2 \trans{\lopapp{i}{C}}\trans{\vec \eta'}\trans{\lopapp{i}{C'}} \conf{A'}{\Gamma_2}{K_2',K_2}{s_2}{e_2'}
    }{
      \conf{A'}{\Gamma_1}{K_1',K_1}{s_1}{e_1'}
      \utreentry*{\bisim{R}}
      \conf{A'}{\Gamma_2}{K_2',K_2}{s_2}{e_2'}
    }
    \quad
   \end{array}\]
  \hrule
  \caption{Up to Proponent Function Re-entry (omitting rules for $\bot$-configurations).}\label{fig:utreentry}
\end{figure*}

The higher-order nature of \lang and its LTS allows infinite nesting of opponent and proponent calls. 
Although up to separation avoids those in a number of examples, here we present a second novel up-to technique, which we call \emph{up to proponent function re-entry} (or simply, up to re-entry).
This technique has connections to the induction hypothesis in the definition of environmental bisimulations in \cite{KoutavasW06}.
However up to re-entry is specifically aimed at avoiding nested calls to proponent functions, and it is designed to work with our symbolic LTS.
In combination with other techniques this eliminates the need to consider configurations with unbounded stacks $K$ in many classical equivalences, including those in~\cite{MeyerS88}.

Up to re-entry is realised by function $\utreentry{}$ in \cref{fig:utreentry}.
The intuition of this up-to technique is that if the application of related functions at $i$ in the $\Gamma$ environments has no potential to change the local stores 
(up to garbage collection)
or increase the $\Gamma$ environments, then there are no additional observations to be made by nested calls to the $i$-functions.
Soundness follows similarly to up-to separation.
\begin{technical}
Soundness follows again by  showing
  $\utreentry{}\comp{\constf{\gfp\WP}}\progress(\utreentry{}\cup g)^\omega\comp{\constf{\gfp\WP}}$
  where $g=\utId{}\sqcup\utperm{}\sqcup\utgc{}$.
\end{technical}

\begin{technical}
\begin{lemma}
  Function $\utreentry{}$ is a sound up-to technique.
\end{lemma}
\end{technical}

In \Hobbit\ we require the user to flag the functions to be considered for the up to re-entry technique.
This annotation is later combined with state invariant annotations, as they are often used together.
Below is an example where the state invariant needed is trivial and up to separation together with up to re-entry are sufficient to finitise the LTS and thus prove the equivalence.

\begin{example}\label{ex:reentry}
  ~\\
  \centerline{
  $M = \text{\lstinline{ref x = 0 in fun f -> f (); !x}}$
  \qquad\qquad
  $N =  \text{\lstinline{fun f -> f (); 0}}$
}\smallskip

\noindent
This is like \cref{ex:sep} except the reference in $M$ is created outside of the function body. The LTS for this is as follows. Labels $\langle \bullet; !x_1 \rangle$ are continuations.\vspace{-2mm}
\begin{center}
\scalebox{0.7}{
\begin{tikzpicture}[node distance=2cm]
\node (s0) [pnode] {$M$};
\node (s1) [onode, right of=s0] {$x_1 : 0$};
\node (s2) [pnode, right=1.5cm of s1] {$x_1 : 0$};
\node (s3) [onode, right=1.5cm of s2,label=north:$\langle \bullet; !x_1 \rangle$] 
	{$x_1 : 0$};
\node[draw=none] (s3b) [below=0.8cm of s3] {$\cdots$};
\node (s4) [pnode, right=1.5cm of s3,label=north:$\langle \bullet; !x_1 \rangle$] 
	{$x_1 : 0$};
\node (s5) [onode, right=1.5cm of s4,label=north:$\langle \bullet; !x_1 \rangle ; \langle \bullet; !x_1 \rangle$] 
	{$x_1 : 0$};
\node[draw=none] (s5b) [below=0.8cm of s5] {$\cdots$};
\node[draw=none] (end) [right=1.5cm of s5] {$\cdots$};

\draw [->] 
(s0) edge node[anchor=south]{$\pret{g}$} (s1)
(s1) edge node[anchor=south]{$\oapp{g}{f_1}$} (s2)
(s2) edge node[anchor=south]{$\papp{f_1}{()}$} (s3)
(s3) edge node[anchor=south]{$\oapp{g}{f_2}$} (s4)
(s3) edge node[anchor=east]{$\oret{()}$} (s3b)
(s4) edge node[anchor=south]{$\papp{f_2}{()}$} (s5)
(s5) edge node[anchor=east]{$\oret{()}$} (s5b)
(s5) edge node[anchor=south]{$\oapp{g}{f_3}$} (end)
;
\end{tikzpicture}
}\vspace{-4mm}
\end{center}
Again, the opponent is allowed to reenter $g$ as before. With up-to reentry, however, the opponent skips nested calls to $g$ as these do not modify the state.\vspace{-2mm}
\begin{center}
\scalebox{0.7}{
\begin{tikzpicture}[node distance=2cm]
\node (s0) [pnode] {$M$};
\node (s1) [onode, right of=s0,label=south:$s_1$] {$x_1 :0$};
\node (s2) [pnode, right=1.5cm of s1] {$x_1: 0$};
\node (s3) [onode, right=1.5cm of s2,label=north:$\langle \bullet; !x_1 \rangle$] 
	{$x_1 = 0$};
\node[draw=none] (s3b) [below=0.8cm of s3] {};
\node (s4) [pnode, right=1.5cm of s3,label=north:$\langle \bullet; !x_1 \rangle$] 
	{$x_1 : 0$};
\node (s5) [onode, right=1.5cm of s4,label=south:$s_1$] 
	{$x_1 : 0$};
\node[draw=none] (s5b) [right=1.5cm of s5] {~~~};

\draw [->] 
(s0) edge node[anchor=south]{$\pret{g}$} (s1)
(s1) edge node[anchor=south]{$\oapp{g}{f_1}$} (s2)
(s2) edge node[anchor=south]{$\papp{f_1}{()}$} (s3)
(s3) edge node[anchor=south]{$\oret{()}$} (s4)
(s4) edge node[anchor=south]{$\pret{0}$} (s5)
;
\draw [->,strike through] (s3) -- node[anchor=east]{$\oapp{g}{f_2}$} (s3b);
\draw [->,strike through] (s5) -- (s5b);
\end{tikzpicture}
}\vspace{-2mm}
\end{center}
\end{example}

%% file: sym.tex
We extend \lang constants (\Con) with a countable set of symbolic constants ranged over by $\kappa$.
We define symbolic environments 
$
    \sig \mis \cdot \mor (\kappa \frown e),\sig
$,
where $\frown$ is either $=$ or $\not=$, and $e$ is an arithmetic expression over constants,
and interpret them as conjunctions of (in-)equalities, with the empty set interpreted as $\top$.

\begin{definition}[Satisfiability]
  Symbolic environment $\sig$ is \emph{satisfiable} if there exists an \emph{assignment} $\modelM$, mapping the symbolic constants of $\sig$ to actual constants, such that $\modelM\sigma$ is a tautology; we then write 
  $\modelM \vDash \sig$.
\end{definition}

We extend reduction configurations with a symbolic environment $\sig$, written as
$\symredconf{\sig}{s}{e}$. These constants are implicitly annotated with their type.
We modify the reduction semantics from \cref{fig:lang} to consider symbolic constants:
$$
\begin{array}{l@{\;\,}l@{\;\,}ll}
    \symredconf{\sig}{s}{\arithop{\vec c}}    & \redbase & \symredconf{\sig \land ({\kappa}={\mathop{op}(\vec c)})}{s}{\symval w}                     & \text{if } \kappa \text{ fresh} \\
    \symredconf{\sig}{s}{\cond {\symval c} {e_1} {e_2}}          & \redbase & \symredconf{\sig \land ({\symval c}={\true})}{s}{e_1}         & \text{if } {\sig \land  ({\symval c}={\true}}) \text{ is sat.} \\
    \symredconf{\sig}{s}{\cond {\symval c} {e_1} {e_2}}          & \redbase & \symredconf{\sig \land  ({\symval c}={\false})}{s}{e_2}         & \text{if } {\sig \land  ({\symval c}={\false}}) \text{ is sat.}\\
\end{array}
$$
All other reduction semantics rules carry the $\sig$.
The LTS from \cref{sec:lts} is modified to operate over configurations of the form $\sig \vdash C$ or $\cdot\vdash\botconf$.
We let $\symval C$ range over both forms of configurations.
All LTS rules for proponent transitions simply carry the $\sig$; rule \iref{tau} may increase $\sig$ due to the inner reduction.
Opponent transitions generate fresh symbolic constants, instead of actual constants:
labels \lopapp{i}{D\hole[\vec \alpha]} and \lopret{D\hole[\vec\alpha]} in rules \iref{opappf} and \iref{opretf} of \cref{fig:lts}, respectively, contain $D$ with symbolic, instead of concrete constants.
We adapt (bi-)simulation as follows.

\begin{definition}
  Binary relation \bisim{R} on symbolic configurations is a \emph{weak simulation} when
  for all $\symval C_1 \bisim*{R} \symval C_2$ and $\symval C_1 \trans{\eta_1} \symval C_1'$,
  there exists $\symval C_2'$
  such that $\symval C_2 \wtrans{\eta_2} \symval C_2'$ and 
  $\symval C_1' \bisim*{R} \symval C_2'$ and
  $(\symval C_1'.\sigma, \symval C_2'.\sigma)$ is sat.\ and
  for all $\modelM \models (\symval C_1'.\sigma, \symval C_2'.\sigma)$
  it is $\delta\eta_1 = \delta\eta_2$.
\end{definition}

\begin{lemma}
  $(\sigma_1 \vdash C_1) \simil (\sigma_2 \vdash C_2) \text{~iff
  for all~} \delta \models \sigma_1,\sigma_2 \text{~we have~} \delta C_1 \simil \delta C_2$.
\end{lemma}
\begin{technical}
\begin{definition}[Bisimilar Expressions]
  Two expressions $\vdash e_1: T$ and $\vdash e_2: T$ are bisimilar, and write $e_1 \bisimil e_2$, when 
  $\cdot \vdash \conf{\emptyA}{\emptyG}{\emptyK}{\emptyS}{e_1} \bisimil
  \cdot \vdash \conf{\emptyA}{\emptyG}{\emptyK}{\emptyS}{e_2}$.
  \defqed
\end{definition}


\begin{lemma}[Soundness of Symbolic Semantics]
  Suppose
  $\modelM \vDash \sig$
  and
  $\symltsconf{\sig}{\symval C}
  \trans{\symval \eta}
  \symltsconf{\sig'}{\symval C'}$;
  then there exists 
  $\modelM' \supseteq \modelM$ such that $\modelM' \vDash \sig'$ and 
  $\modelM{\symval C}
      \trans{\modelM'\symval \eta}
      \modelM'{\symval C'}$.
\end{lemma}

\begin{lemma}[Completeness of Symbolic Semantics]
  Suppose
  $\modelM \vDash \sig$
  and
  $\modelM \symval C \trans{\eta} C'$;
  then there exist $\symltsconf{\sig'}{\symval C'}$ and $\modelM' \supseteq \modelM$ and $\eta'$ such that
  $\modelM' \vDash \sigma'$ and $\eta = \modelM'\symval\eta$ and 
  $\symltsconf{\sig}{\symval C} 
  \trans{\symval \eta} 
  \symltsconf{\sig'}{\symval C'}$.
  
\end{lemma}
\end{technical}
\begin{corollary}[Soundness, Completeness]
  $(\cdot \vdash C_1) \simil (\cdot \vdash C_2) \text{~iff~}
  C_1 \simil C_2$.
\end{corollary}
The up-to techniques we have developed in previous sections apply unmodified to
the extended LTS as the techniques do not involve symbolic constants, with the exception of up to beta which requires adapting the definition of a beta move to consider all possible $\modelM$.
The
introduction of symbolic first-order transitions allows us to prove many interesting
first-order examples, such as the equivalence of bubble sort and insertion
sort, an example borrowed from \tool{Hector}~\cite{hector}
(omitted here, see the \Hobbit\ distribution).
Below is a simpler example showing the equivalence of two integer
swap functions which \Hobbit\ is able to prove.
\begin{example}\label{ex:swap}
~\\
    \begin{tabular}{llll}
      $M =\;$
      &\begin{minipage}[t]{0.3\textwidth}\vspace{-1.3em} \begin{lstlisting}[boxpos=t]
let swap xy =
  let (x,y) = xy
  in (y, x)
in swap
\end{lstlisting}
    \end{minipage}
      &
      $N =\;$
      &\begin{minipage}[t]{0.3\textwidth}\vspace{-1.3em} \begin{lstlisting}[boxpos=t]
fun xy -> let (x,y) = xy in
  ref x = x in ref y = y in
  x := !x - !y; y := !x + !y;
  x := !y - !x; (!x, !y)
\end{lstlisting}
    \end{minipage}
    \end{tabular}
\end{example}

%% file: generalisation.tex
The addition of symbolic constants into \lang and the LTS not only allows us to
consider all possible opponent-generated constants simultaneously in a symbolic
execution of proponent expressions, but also allows us to define an additional
powerful up-to technique: \emph{up to state invariants}.
We define this technique in two parts: \emph{up to abstraction} and
 \emph{up to tautology} realised by $\utabs{}$ and $\uttaut{}$.\footnote{%
  \Hobbit\ also implements an \emph{up to $\sigma$-normalisation and garbage collection} technique.}
\\[1ex]
  $\begin{array}{@{}c@{}}
    \irule*[UpTo$\utabs{}$][uptogen]{
      (\sig_1 \vdash C_1)
      \bisim*{R}
      (\sig_2 \vdash C_2)
    }{
      (\sig_1 \vdash C_1)\sub{\vec\kappa}{\vec c}
      \utabs*{\bisim{R}}
      (\sig_2 \vdash C_2)\sub{\vec\kappa}{\vec c}
    }
\quad
    \irule*[UpTo$\uttaut{}$][uptotaut]{
      (\sig_1,\sigma_1' \vdash C_1)
      \bisim*{R}
      (\sig_2,\sigma_2' \vdash C_2)
      \\\\
      \sig_1,\sig_2,\sigma_1',\sigma_2' \text{ is sat.}
      \\\\
      \sig_1,\sig_2 \land \neg(\sigma_1',\sigma_2') \text{ is not sat.}
    }{
      (\sig_1 \vdash C_1)
      \uttaut*{\bisim{R}}
      (\sig_2 \vdash C_2)
     }
  \end{array}$\\[1ex]
The first function $\utabs{}$ allows us to derive the equivalence of configurations by abstracting constants with fresh symbolic constants (of the same type) and instead prove equivalent the more abstract configurations.
The second function $\uttaut{}$ allows us to introduce tautologies into the symbolic environments. These are predicates which are valid; i.e., they hold for all instantiations of the abstract variables. Combining the two functions we can introduce a tautology $I(\vec c)$ into the symbolic environments, and then abstract constants $\vec c$ from the predicate but also from the configurations with symbolic ones, obtaining $I(\vec \kappa)$, which encodes an invariant that always holds.  

Currently in \Hobbit, up to abstraction and tautology are combined and applied in a principled way. Functions can be annotated with the following syntax:
\begin{lstlisting}
  $F$ = fun x {$\vec\kappa$ | $l_1$ as $C_1[\vec\kappa]$, ..., $l_n$ as $C_n[\vec\kappa]$ | $\phi$} -> e
\end{lstlisting}
The annotation instructs \Hobbit\ to use the two techniques when opponent applies related functions where at least one of them has such an annotation. If both functions contain annotations, then they are combined and the same $\vec{\kappa}$ are used in both annotations.
The techniques are used again when proponent returns from the functions,
and proponent calls opponent from within the functions.\footnote{Finer-grain control of application of these up-to techniques is left to future work.}
As discussed in \cref{sec:reentry}, the same annotation enables up to reentry in \Hobbit.

When \Hobbit\ uses the above two up-to techniques it
\begin{inparaenum}[\itshape 1)\upshape]
\item pattern-matches the values
currently in each location $l_i$ with the value context $C_i$ where
fresh symbolic constants $\vec\kappa$ are in its holes,
obtaining a substitution $\sub{\vec\kappa}{\vec c}$;
\item
the up to tautology technique is applied for the formula $\phi\sub{\vec\kappa}{\vec c}$;
and
\item
  the up to abstraction technique is applied by replacing
$\phi\sub{\vec\kappa}{\vec c}$ in the symbolic environment with $\phi$, and the
contents of locations $l_i$ with $C_i[\vec{\kappa}]$.
\end{inparaenum}
\begin{example}\label{ex:meyer-sieber-e6}
  Following is an example by Meyer and Sieber~\cite{MeyerS88} featuring location passing,
  adapted to \lang where locations are local. Full example in \cref{sec:meyer-sieber-e6-full}.
  \begin{lstlisting}
$M =$ let loc_eq loc1loc2 = $[\dots]$ in
     fun q -> ref x = 0 in
    	      let locx = (fun () -> !x) , (fun v -> x := v) in
    	      let almostadd_2 locz {w | x as w | w mod 2 == 0} =
    	        if loc_eq (locx,locz) then x := 1 else x := !x + 2
    	      in q almostadd_2; if !x mod 2 = 0 then _bot_ else ()

$N =$ fun q -> _bot_
  \end{lstlisting}
%
%
  In this example we simulate general references as a pair of read-write functions.
  Function \lstinline{loc_eq} implements a standard location equality test (see \cref{sec:meyer-sieber-e6-full}).
  The two higher-order expressions are equivalent because the opponent can only increase the contents of \lstinline{x} through the function \lstinline{almostadd_2}.
  As the number of times the opponent can call this function is unbounded, the LTS is infinite.
  However, the annotation of function \lstinline{almostadd_2} applies the up to state invariants technique when the function is called (and, less crucially, when it returns), replacing the concrete value of \lstinline{x} with a symbolic integer constant $w$ satisfying the invariant \lstinline{w mod 2 == 0}. This makes the LTS finite, up to permutations of symbolic constants.
  Moreover, up to separation removes the outer functions from the $\Gamma$ environments, thus preventing re-entrant calls to these functions.
  Note the up to techniques are applied even though one of the configurations is diverging (\lstinline{_bot_}).
  This would not be possible with the LTS and bisimulation of~\cite{BiernackiLP19}. 
\end{example}

\begin{technical}
\begin{example}\label{ex:cell-4}
  Using up to abstraction and tautology, \Hobbit\ is able to prove the example below, an adaptation from~\cite{KoutavasW06}.
  The invariant relates locations in the two expressions by assigning the same $\vec \kappa$ to
  \lstinline{wp}, \lstinline{w1} and \lstinline{w2} in both programs.
  \begin{lstlisting}
$M =$ ref y = 0 in
     let set z = {wp, wy1, wy2, wy | y as wy | 
        ((wp mod 2 = 0) && (wy = wy1))
        || ((wp mod 2 <> 0) && (wy == wy2))} -> y := z in
     let get () = !y in
     (set , get)

$N =$ ref y1 = 0 in ref y2 = 0 in ref p = 0 in
     let set1 z =
       {wp, wy1, wy2, wy | p as wp; y1 as wy1; y2 as wy2 | true}
       -> p := !p + 1; if !p mod 2 = 0 then y1 := z else y2 := z in
     let get1 () = if !p mod 2 = 0 then !y1 else !y2 in
     (set1 , get1)
  \end{lstlisting}
\end{example}
\end{technical}


%% file: imp.tex
We implemented the LTS and up-to techniques for \lang in a tool prototype called \Hobbit{}, which we ran on a test-suite of 105 equivalences and 68 inequivalences---3338 and 2263 lines of code for equivalences and inequivalences respectively.

\Hobbit{} is bounded in the total number of function calls it explores per path. We ran \Hobbit{} with a default bound of 6 calls except where a larger bound was found to prove or disprove equivalence---46 examples required a larger bound, and the largest bound used was 348.
To illustrate the impact of up-to techniques, we checked all files (pairs of expressions to be checked for equivalence) in five configurations: default (all up-to techniques on), up to separation off, annotations (up to state invariants and re-entry) off, up to re-entry off, and everything off. The tool stops at the first trace that disproves equivalence, after enumerating all traces up to the bound, or after timing out at 150 seconds. Time taken and exit status (equivalent, inequivalent, inconclusive) were recorded for each file; an overview of the experiment can be seen in 
the following table. All experiments ran on an Ubuntu 18.04 machine with 32GB RAM, Intel Core i7 1.90GHz CPU, with intermediate calls to Z3 4.8.10 to prune invalid internal symbolic branching and decide symbolic bisimulation conditions. All constraints passed to Z3 are of propositional satisfiability in conjunctive normal form (CNF). \linebreak
\vspace{-3mm}

%
\centerline{\small\renewcommand{\arraystretch}{1.1}
\begin{tabular}{| c |@{\;} c@{\;} |@{\;} c@{\;} |@{\;} c@{\;} |@{\;} c@{\;} |@{\;} c@{\;} |}
\hline
& \bf default & \bf  sep. off &\bf annot. off &\bf ree. off &\bf all off\\
\hline
\bf eq.~ & 72 $\mid$ 0 [5.6s] & 32 $\mid$ 0 [1622.9s] & 47 $\mid$ 0 [178.3s] & 57 $\mid$ 0 [177.6s] & 3 $\mid$ 0 [2098.5s]\\
\bf ineq.~ & 0 $\mid$ 68 [20.0s] & 0 $\mid$ 66 [312.8s] & 0 $\mid$ 68 [19.6s] & 0 $\mid$ 68 [20.1s] & 0 $\mid$ 65 [515.7s]\\
\hline
      \multicolumn{6}{|c|}{\small
        \thead{$a \mid b$ [$c$] for $a$ (out of 105) equivalences and \\
        $b$ (out of 68) inequivalences reported taking $c$ seconds in total.}
      }\\
\hline
\end{tabular}}\smallskip

We can observe that \Hobbit{} was sound and bounded-complete for our examples; no false reports and all inequivalences were identified. Up-to techniques also had a significant impact on proving equivalence. With all techniques on, it proved 68.6\% of our equivalences; a dramatic improvement over 2.9\% proven with none on. The most significant technique was up-to separation---necessary for 55.6\% of equivalences proven and reducing time taken by 99.99\%---which was useful when functions could be independently explored by the context. Following was annotations---necessary for 34.7\% of equivalences and decreasing time by 96.9\%---and up-to re-entry---20.8\% of files and decreased time by 96.8\%. Although the latter two required manual annotation, they enabled equivalences where our language was able to capture the proof conditions. Note that, since turning off invariant annotations also turns off re-entry, only 10 files needed up-to re-entry on top of invariant annotations.
In contrast, inequivalences did not benefit as much. This was expected as without up-to techniques \Hobbit{} is still based on bounded model checking, which is theoretically sound and complete for inequivalences. Nonetheless, three files timed out with techniques turned off, which suggests that the reduction in state space is still relevant when searching for counterexamples.


%% file: syteci.tex
There are two main classes of tools for contextual equivalence checking. The first one includes semantics-driven tools that tackle higher-order languages with state like ours. In this class belong game-based tools \tool{Hector}~\cite{hector} and \tool{Coneqct}~\cite{coneqct}, which can only address carefully crafted fragments of the language, delineated by type restrictions and bounded data types. The most advanced tool in this class is \tool{SyTeCi}~\cite{syteci}, which is based on logical relations and removes a good part of the language restrictions needed in the previous tools.
The second class concerns tools that focus on first-order languages, typically variants of C, with main tools including \tool{R\^eve}~\cite{reve}, \tool{SymDiff}~\cite{symdiff} and \tool{RVT}~\cite{rvt}. These are highly optimised for handling \emph{internal loops}, a problem orthogonal to handling the interactions between higher-order functions and their environment, addressed by \Hobbit\ and related tools. We believe the techniques used in these tools may be useful when adapted to \Hobbit, which we leave for future work.

In the higher-order contextual equivalence setting, the most relevant tool to compare with \Hobbit{} is \tool{SyTeCi}. This is because \tool{SyTeCi} supersedes previous tools by proving examples with fewer syntactical limitations. We ran the tools on examples from both \tool{SyTeCi}'s and our own benchmarks---7 and 15 equivalences, and 2 and 7 inequivalences from \tool{SyTeCi} and \Hobbit{} respectively---with a timeout of 150s and using Z3. Unfortunately, due to differences in parsing and \tool{SyTeCi}'s syntactical restrictions, the input languages were not entirely compatible and only few manually translated programs were chosen.

\centerline{\small\renewcommand{\arraystretch}{1.1}
\newcolumntype{Y}{>{\centering\arraybackslash}X}
\begin{tabularx}{0.95\textwidth}{| Y | Y | Y |}
\hline
& \bf SyTeCi & \bf Hobbit \\
\hline
\bf SyTeCi eq. examples & 3 $\mid$ 0 $\mid$ 4 (0.03s) & 1  $\mid$  0  $\mid$  6 (\textless 0.01s) \\
\bf Hobbit eq. examples & 8  $\mid$  0  $\mid$  7 (0.4s) & 15  $\mid$  0  $\mid$  0 (\textless 0.01s) \\
\hline
\hline
\bf SyTeCi ineq. examples & 0  $\mid$  2  $\mid$  0 (0.06s) & 0  $\mid$  2  $\mid$  0 (0.02s) \\
\bf Hobbit ineq. examples & 2  $\mid$  3  $\mid$  2 (0.52s) & 0  $\mid$  7  $\mid$  0 (0.45s) \\
\hline
      \multicolumn{3}{|c|}{\small
  $a \mid b \mid c$ ($d$) for $a$ eq's, $b$ ineq's and $c$ inconclusive's reported taking $d$ sec in total
      }\\
\hline
\end{tabularx}}\smallskip

We were unable to translate many of our examples because of restrictions in the input syntax 
supported by 
\tool{SyTeCi}.
Some of these restrictions were inessential (e.g.\ absence of tuples) while others were substantial: the tool does not support programs where references are allocated both inside and outside functions (e.g.~\cref{ex:bohrb}),
or with non-synchroniseable recursive calls.
Moreover, \tool{SyTeCi} relies on Constrained Horn Clause satisfiability which is undecidable. In our testing
\tool{SyTeCi} sometimes timed out on examples; in private correspondence with its creator this was attributed to Z3's ability to solve Constrained Horn Clauses. Finally, \tool{SyTeCi} was sound for equivalences, but not always for inequivalences as can be seen in the table above; the reason is unclear and may be due to bugs. On the other hand, \tool{SyTeCi} was able to solve equivalences we are not able to handle; e.g.~synchronisable recursive calls and examples like the well-bracketed state problem:\\[2mm]
\centerline{$\begin{aligned}
  M &= \text{\lstinline{ref x = 0 in fun f -> x:=0; f(); x:=1; f(); !x}}
  \\
  N &=  \text{\lstinline{fun f -> f(); f(); 1}}
\end{aligned}$}


%% file: rel.tex
Our experience with \Hobbit{} suggests that our technique provides a significant contribution to verification of contextual equivalence.
In the higher-order case, \Hobbit{} does not impose language restrictions as present in other tools.
Our tool is able to solve several examples that can not be solved by \tool{SyTeCi}, which is the most advanced tool in this family.
In the first-order case, the problem of contextual equivalence differs significantly as the interactions that a first-order expression can have with its context are limited; e.g.\ equivalence analyses do not need to consider callbacks or re-entrant calls. Moreover,
the distinction between global and local state is only meaningful in higher-order languages where a program phrase can invoke different calls of the same function, each with its own state.
Therefore, tools for first-order languages  focus on what in our setting are internal transitions and the complexities arising from e.g.\ unbounded datatypes and recursion, whereas we focus on external interactions with the context. 

As for limitations, \Hobbit{} {does not handle synchronised} internal recursion and well-bracketed state, which \tool{SyTeCi} can often solve. More generally, \Hobbit{} is not optimised for internal recursion as first-order tools are. 
In this work we have also disallowed reference types in \lang to simplify the technical development;
location exchange is encoded via function exchange (cf.~\cref{ex:meyer-sieber-e6}). 
We intend to address these limitations in future work.

\begin{technical}
\begin{NoHyper}
\newcommand\blfootnote[1]{%
  \begingroup
  \renewcommand\thefootnote{}\footnote{#1}%
  \addtocounter{footnote}{-1}%
  \endgroup

\blfootnote{%
The first two authors were supported, in part, by the Science Foundation Ireland grant 13/RC/2094 (LERO) and co-funded under the European Regional Development Fund. 
For the purpose of Open Access, the authors have a CC BY public copyright licence to any Author Accepted Manuscript version arising from this submission.
}
\end{NoHyper}
\end{technical}


%% file: typing.tex
\section{Typing rules of \lang}\label{sec:typing}

\newcommand\qweqwe{\quad}
\newcommand\preG[1][]{\Delta#1;\varSigma\vdash}
\newcommand\atype[1]{\mathsf{#1}}
\newcommand\vars{Vars}
\newcommand\refs{Refs}
\newcommand\meths{Meths}

\begin{gather*}
{\infer{c\text{ cons.\ of type }T}{ \preG c : T }}
\qweqwe
\infer{ (x:T)\in\Delta }{ \preG x : T }
\qweqwe
\infer{ \preG e_1: T_1\quad\dots\quad \preG e_n: T_n }{ \preG (e_1,\dots,e_n) : T_1*\dots*T_n }
\\[1mm]
\infer{ op:\vec T\to T\quad  \preG (\vec e):\vec T }{ \preG \arithop{\vec e} : T }
\qweqwe
\infer{ \preG e : \atype{bool}\qquad  \preG (e_1,e_2) : T*T }{ \preG \cond{e}{e_1}{e_2} : T }
\\[1mm]
\infer{\preG v : T \quad  \Delta;\varSigma,l:T\vdash e : T' }{ \preG \new l v e : T' }
\qweqwe 
\infer{ (l:T)\in\varSigma }{ \preG {!l} : T }
\qweqwe
\infer{ (l:T)\in\varSigma \quad \preG e : T }{\preG l := e : \atype{unit} }
\\[1mm]
\infer
{ \preG e: T \rightarrow T' \quad \preG e' : T }
{ \preG ee' : T' }
\qweqwe 
\infer{ \preG[,f:T\to T',x:T] e : T' }{ \preG \lam{x}{e} : T\to T'}
\\[1mm]
\infer{ \preG[,x_1:T_1,\dots,x_n:T_n] e : T \quad  \preG e' : \vec T }{ \preG \elet{(\vec x)}{e'}{e} : T }
\end{gather*}


%% file: SC.tex
\section{Proof of \cref{theorem:SC}}\label{sec:sound-complete}

We let a \emph{trace} be a sequence of app and ret moves (i.e.\ labels), as defined in  \cref{fig:lts}. A trace is \emph{complete} if it is starts with a fully bracketed segment, followed by a proponent return and is afterwards again fully bracketed, i.e.\ it adheres to the grammar:
\begin{align*}
  CT &::= \ Y\,\lpropret{D}\,X\\
  X  &::= \ \varepsilon\mor\lopapp{\alpha}{D}\ Y\,\lpropret{D}\,X\\
  Y  &::= \ \varepsilon\mor\lpropapp{\alpha}{D}\ X\,\lopret{D}\,Y
\end{align*}

\begin{theorem}[\cite{Laird07}]\label{theorem:FA}
Expressions $\vdash e_1: T$ and $\vdash e_2: T$ are contextually equivalent iff the configurations
  $\conf{\emptyA}{\emptyG}{\emptyK}{\emptyS}{e_1},
  \conf{\emptyA}{\emptyG}{\emptyK}{\emptyS}{e_2}$ produce the same complete traces.
\end{theorem}

\paragraph{Proof of \cref{theorem:SC}}
We first note that our LTS is deterministic modulo the selection of fresh locations in tau transitions. Moreover,
by \cref{theorem:FA}, it suffices to show that $e_1,e_2$ are bisimilar iff they have the same complete traces.

If the two expressions have the same complete traces then each of them can match any challenge posed by the other, so long as such a challenge can lead to a complete trace. If a challenge is doomed to not complete, then it can be matched by a transition to $\botconf$. 
Conversely, if $e_1,e_2$ are bisimilar then any transition sequence yielding a complete trace of $e_1$ can be simulated by $e_2$, and viceversa, so the two expressions have the same complete traces.\qed
 
\begin{remark}
  Soundness is also proved via the up to techniques in \cref{sec:soundness}.
  Completeness is also proved directly in \cref{sec:completeness}.
\end{remark}


%% file: enhancements.tex
\section{Theory of Enhancements}\label{sec:enhancements}
We develop our up-to techniques using the theory of bisimulation enhancements from \cite{PousS11,PousCompanion}.
Here we summarise main definitions, starting with the notions of progressions and compatible functions \cite{PousS11}.
The main result of this section is a set of proof obligations with which we can proof an up-to technique sound, shown in
\cref{lem:prfs}.
%
We start by defining basic operations on monotone functions.

\begin{definition}
  Consider monotone functions $f,g: \pow{X}\rightarrow \pow{X}$ on some set $X$.
  We write $f\comp g$ for the composition of $f$ and $g$, and
  $f \sqcup g$ for the function $\bisim{S} \mapsto f(\bisim{S}) \sqcup g(\bisim{S})$.
  For any set $F$ of functions, we write $\bigsqcup F$ for the function $\bisim{S} \mapsto \bigcup_{f\in F}f(\bisim{S})$.
  We also write $\constf{X}$ to be the constant function with range $\{X\}$.
  We let $f^0 \defeq \utId{}$ and $f^{n+1} \defeq f\comp f^{n}$. Moreover, we write
  $f^{\omega}$ to mean $\bigsqcup_{k<\omega} f^{k}$.
  We write $f\sqsubseteq g$ when, for all $\bisim{S}\in\pow{X}$, $f(\bisim{S}) \subseteq g(\bisim{S})$.
\end{definition}

The theory of enhancements we use here is based on the notion of \emph{weak progression}. 
Weak progression is first defined as a monotone function on configuration relations ($\WP(\bisim R))$,
and then used for a pre-fixpoint predicate on configuration relations (${\bisim{R}}\progress[\WP]{\bisim{S}}$)
and one on monotone functions over said relations ($f\progress[\WP]g$).
The latter functions are meant to encode up-to techniques.

\begin{definition}[Progressions ($\progress$)]
  ~
  \begin{itemize}
    \item
  $
    \WP(\bisim R) = \{(C_1,C_2) \where \forall C_1',\eta \wehave C_1\trans{\eta}C_1' \timplies \exists C_2'.
    C_2 \wtrans{\eta} C_2' \tand C_1' \bisim*{R} C_2'\}
  $.
    \item 
      $\bisim{R}$ \emph{weakly progresses} to $\bisim{S}$, and we write 
      ${\bisim{R}}\progress[\WP]{\bisim{S}}$ when ${\bisim{R}}\subseteq\WP(\bisim{S})$.

    \item For monotone functions 
      $f, g$ 
    we write $f \progress[\WP] g$ 
      when $f\comp\WP \sqsubseteq \WP\comp g$.
\defqed
  \end{itemize}
\end{definition}

\begin{lemma}
  \bisim{R} is a weak simulation when ${\bisim{R}}\progress[\WP]{\bisim{R}}$.
  Also, $({\simil}) = (\gfp\WP)$.
  \qed
\end{lemma}
The following gives the definition of an up-to technique, what it means to be sound, and the stronger notion of compatibility.
\begin{definition}~
  \begin{itemize}
    \item\emph{Bisimulation up-to:}
      \bisim{R} is a \emph{weak simulation up to $f$} when ${\bisim{R}} \progress[\WP] f(\bisim{R})$.

    \item\emph{Sound up-to technique:}
      Function $f$ is \emph{$\WP$-sound} when $\gfp{\WP\comp f} \subseteq \gfp\WP$.

    \item\emph{Compatibility:}
      Monotone function $f$ is \emph{$\WP$-compatible} when $f\progress[\WP]f$.
  \end{itemize}
\end{definition}

\begin{lemma}[\cite{PousS11}, Lem. 6.3.12] \label{lem:compat-alternative}
  $f\progress[\WP]f$ if and only if 
  for all ${\bisim{R}}\progress[\WP]{\bisim{S}}$ we have $f\comp\WP(\bisim{R}) \subseteq \WP\comp g(\bisim{S})$.
  \qed
\end{lemma}

\begin{lemma}[\cite{PousS11}, Thm. 6.3.9] \label{lem:compat-sound} 
  If $f$ is $\WP$-compatible then it is $\WP$-sound.
  \qed
\end{lemma}

\begin{lemma}[\cite{PousS11}, Prop. 6.3.11 and 6.3.12]\label{lem:pous-basic}
  The following functions are $\WP$-compatible:
  \begin{itemize}
    \item the reflexive \utRefl{} and identity \utId{} functions;
    \item $f \comp g$, for any $\WP$-compatible monotone functions $f$, $g$;
    \item $\bigsqcup F$, for any set $F$ of $\WP$-compatible monotone functions.\qed
    \end{itemize}
\end{lemma}

Pous \cite{PousCompanion} extends the theory of enhancements with the notion of companion of $\WP$, the largest 
 $\WP$-compatible function.
\begin{definition}[Companion]
  $\companion[\WP] \defeq \bigsqcup\{ f:\mathcal{P}(\Conf^2)\to\mathcal{P}(\Conf^2) \mid  {f\progress[\WP]f} \}$.
\end{definition}

\begin{lemma}[\cite{PousCompanion}]\label{lem:companion-props}~
  \begin{enumerate}
    \item\label{lem:companion-props-1}
      \protect{\companion[\WP]} is $\WP$-compatible: $\companion[\WP]\progress\companion[\WP]$;
    \item\label{lem:companion-props-2}
      $\WP$ is $\WP$-compatible: $\WP\sqsubseteq\companion[\WP]$;
    \item\label{lem:companion-props-3}
      \protect{\companion[\WP]} is idempotent: ${\utId{}}\sqsubseteq\companion[\WP]$ and 
      $\companion[\WP]\comp\companion[\WP]\sqsubseteq\companion[\WP]$;
    \item\label{lem:companion-props-5}
      \companion[\WP] is $\WP$-sound: $\gfp{\WP\comp\companion[\WP]}\subseteq \gfp\WP$.
      \qed
  \end{enumerate}
\end{lemma}

This gives rise a proof technique for proving up-to techniques sound.
\begin{lemma}\label{lem:subseteq-companion-sound}
  Let $f\sqsubseteq\companion[\WP]$. Then $f$ is $\WP$-sound.
\end{lemma}
\begin{proof}
  By showing that $f\cup\companion[\WP]\progress[\WP]f\cup\companion[\WP]$ and using \cref{lem:compat-sound}.
\end{proof}

\begin{lemma}[Function Composition Laws]\label{lem:fcl}
  Consider monotone functions $f,g,h: \pow{X}\rightarrow \pow{X}$ and set ${\bisim{S}}\in\pow{X}$. We have
  \begin{enumerate}
    \item $\constf{\bisim{S}} \comp f = \constf{\bisim{S}}$
    \item $(f\sqcup g)\comp h = (f\comp h) \sqcup (g \comp h)$
    \item $h\comp(f\sqcup g) = (h\comp f) \sqcup (h\comp g)$
    \item $(f\sqcup g) \sqsubseteq (f\sqcup h)$ and
          $(f\comp g) \sqsubseteq (f\comp h)$ and
          $(g\comp f) \sqsubseteq (h\comp f)$, when $g\sqsubseteq h$.
    \item
      $f \sqsubseteq f^\omega$ and $f \comp f^\omega = f^\omega \comp f \sqsubseteq f^\omega \comp f^\omega \sqsubseteq f^\omega$.
      \item $f^\omega\circ g = \bigsqcup_{i<\omega}(f^i\circ g)$.
      \qed
  \end{enumerate}
\end{lemma}

We distil this up-to technique to the following three proof obligations, each sufficient for proving the soundness of up-to techniques.
\begin{lemma}[POs for Up-To Soundness]\label{lem:prfs}
  Let $f$ be a monotone function and $\bisim{R}$ be a weak simulation; $f$ is $\WP$-sound when one of the following holds:
  \begin{enumerate}
    \item \label{lem:prfs-1}
      $f\progress[\WP]f$; or

    \item \label{lem:prfs-2}
      $f\progress[\WP](f\comp g)$, for some $g\sqsubseteq\companion[\WP]$; or



    \item \label{lem:prfs-3}
      $f = \bigsqcup_{f_i\in F}f_i\comp{\constf{\gfp\WP}}$,
      where $F$ is a set of monotone functions and, 
      for all $f_i\in F$, there exists $g_i\sqsubseteq\companion[\WP]$ such that
      $f_i\comp{\constf{\gfp\WP}}\progress[\WP](f\sqcup g_i)^{\omega}\comp{\constf{\gfp\WP}}$.

 \end{enumerate}
\end{lemma}
\begin{proof}~
  \begin{enumerate}
    \item By \cref{lem:compat-sound}.
    \item
      By \cref{lem:subseteq-companion-sound}, it suffices to show $f \sqsubseteq \companion[\WP]$.
      Because $f\sqsubseteq f \comp(\utId{}\sqcup\companion[\WP])\sqsubseteq f\comp\companion[\WP]$,
      it suffices to show $f\comp\companion[\WP]\progress[\WP]f\comp\companion[\WP]$ by 
      unfolding definitions and the premise:
      \begin{align*}
        f \comp \companion[\WP] \comp \WP \sqsubseteq
        f \comp \WP \comp \companion[\WP] \sqsubseteq
        \WP\comp f \comp g \comp \companion[\WP] \sqsubseteq
        \WP\comp f \comp \companion[\WP] \comp \companion[\WP] \sqsubseteq
        \WP\comp f \comp \companion[\WP].
      \end{align*}

    \item 
      Let $g=\bigsqcup_{f_i\in F}g_i$.
      By \cref{lem:subseteq-companion-sound}, it suffices to show $f \sqsubseteq \companion[\WP]$.
      Because
      \begin{align*}
        f_i\comp{\constf{\gfp\WP}}
        &=
        f_i\comp{\constf{\gfp\WP}}\comp{\constf{\gfp\WP}}
        \sqsubseteq
        f\comp{\constf{\gfp\WP}}
        \sqsubseteq
        (f\comp{\constf{\gfp\WP}}) \sqcup (\companion[\WP]\comp{\constf{\gfp\WP}})
        \\&=
        (f\sqcup\companion[\WP])\comp{\constf{\gfp\WP}}
        \sqsubseteq
        (f\sqcup\companion[\WP])^\omega\comp{\constf{\gfp\WP}}
      \end{align*}
      it suffices to show that 
      $(f\sqcup\companion[\WP])^\omega\comp{\constf{\gfp\WP}} \progress[\WP] (f\sqcup\companion[\WP])^\omega\comp{\constf{\gfp\WP}}$.
      This is proven by showing that for all $k$,
      \begin{equation}\tag{$P(k)$}
        (f\sqcup\companion[\WP])^k\comp{\constf{\gfp\WP}} \progress[\WP] (f\sqcup\companion[\WP])^\omega\comp{\constf{\gfp\WP}}.
        \end{equation}
      We proceed by induction on $k$. The base case is straightforward:
$$
\utId{}\circ\constf{\gfp\WP}\circ\WP =\constf{\gfp\WP} = \WP\circ \utId{}\circ\constf{\gfp\WP}\sqsubseteq
\WP\circ(f\sqcup\companion[\WP])^\omega\comp{\constf{\gfp\WP}}
$$
      In the inductive case we assume $P(k)$ and prove $P(k+1)$ as follows:
      \begin{align*}
        &(f\sqcup\companion[\WP])^{k+1}\comp{\constf{\gfp\WP}}\comp\WP
        \\&=
        (f\sqcup\companion[\WP])\comp(f\sqcup\companion[\WP])^{k}\comp{\constf{\gfp\WP}}\comp\WP
        \\&
        \sqsubseteq
        (f\sqcup\companion[\WP])\comp\WP\comp h
        &\text{($P(k),~h=(f\sqcup\companion[\WP])^\omega\comp{\constf{\gfp\WP}}$)}
        \\&
        \sqsubseteq
        (f\comp\WP\comp h) \sqcup (\companion[\WP]\comp\WP\comp h) 
        &\text{(\cref{lem:fcl})}
        \\&
        =
        \left(\bigsqcup_{f_i\in F}(f_i\comp{\constf{\gfp\WP}}\comp\WP\comp h)\right) \sqcup (\companion[\WP]\comp\WP\comp h)
        &\text{(definition of $f$ and \cref{lem:fcl})}
        \\&
        \sqsubseteq
        \left(\bigsqcup_{f_i\in F}(\WP\comp(f\sqcup g_i)^\omega\comp{\constf{\gfp\WP}}\comp h)\right)
        \sqcup (\companion[\WP]\comp\WP\comp{h})
        &\text{(premise)}
        \\&
        \sqsubseteq
        \left(\bigcup_{f_i\in F}(\WP\comp(f\sqcup \companion[\WP])^\omega\comp{\constf{\gfp\WP}})\right)
        \sqcup (\companion[\WP]\comp\WP\comp h)
        &\text{(\cref{lem:fcl} and premise on $g_i$)}
        \\&
        \sqsubseteq
        (\WP\comp(f\sqcup \companion[\WP])^\omega\comp{\constf{\gfp\WP}}) \sqcup (\WP\comp\companion[\WP]\comp h) 
        &\text{(\cref{lem:companion-props}~(\ref{lem:companion-props-1}))}
        \\&
        \sqsubseteq
        (\WP\comp\utId{}\comp h) \sqcup (\WP\comp\companion[\WP]\comp h) 
        &\text{(definition of $h$)}
        \\&
        =
        \WP\comp(\utId{}\sqcup\companion[\WP])\comp h
        &\text{(\cref{lem:fcl})}
        \\&
        = \WP\comp\companion[\WP]\comp h
        &\text{(\cref{lem:companion-props}~(\ref{lem:companion-props-3}))}
        \\&
        \sqsubseteq
        \WP\comp(f\sqcup\companion[\WP])\comp h
        &\text{(\cref{lem:fcl})}
        \\&
        \sqsubseteq
        \WP\comp(f\sqcup\companion[\WP])^\omega\comp {\constf{\gfp\WP}}
        &\text{(\cref{lem:fcl} and definition of h )}
         \\&&
        \qedhere
      \end{align*}
  \end{enumerate}
\end{proof}

As we are only interested in weak progression, in the following we drop the $\WP$ annotation from progressions, compatibility and companion.

\section{Simple Up-To Techniques}\label{app:simple}
We develop our up-to techniques using the theory of bisimulation enhancements from \cite{PousS11,PousCompanion} (see~Appendix~\ref{sec:enhancements}).
We start by presenting three straightforward up-to techniques which nevertheless are needed to reduce the configurations considered by bisimulation, achieving finite LTSs in many examples.
These techniques are \emph{up to permutations}, \emph{beta reductions}, \emph{garbage collection}, and \emph{weakening of knowledge environments}.
To present these techniques we first need the following definitions.

\begin{definition}[Permutations]
We consider permutations of store locations, $\pi_l$, abstract names, $\pi_\alpha$ and environment indices, $\pi_i$, respectively.
When applying a permutation $\pi_l$ to a store $s$, the former acts on both the domain and range of the latter.
When applying a permutation $\pi_i$ to an environment $\Gamma$, it only acts on its domain; other types of permutations only act on the codomain of $\Gamma$.
\end{definition}

\begin{definition}[$\beta$-move]\label{def:beta-move} 
  A $\tau$-transition $C \trans{\tau} C'$ is called a \emph{$\beta$-move}, and
  we write $C \betatrans C'$, when for all transitions $C \trans{\eta} C''$,
  one of the following holds:
  \begin{itemize}
    \item $\eta=\tau$ and $C'=C''$; or
    \item there exists $C'''$ such that $C' \trans{\eta} C'''$ and $C'' \betatrans C'''$ or $C''=C'''$.
  \end{itemize}
\end{definition}

\begin{definition}[Garbage Collection]\label{def:asymp}
  We let $(\asymp)$ be the largest equivalence relation between well-formed configurations with the axioms:
  \begin{itemize}
    \item $\conf{A}{\Gamma}{K}{s}{\hat e} \asymp \conf{A\uplus A'}{\Gamma}{K}{s}{\hat e}$
    \item $\conf{A}{\Gamma}{K}{s}{\hat e} \asymp \conf{A}{\Gamma}{K}{s,s_g}{\hat e}$
  \end{itemize}
  for any $A$, $A'$, $K$, $s$, $s_g$, $\hat e$ with $A'\cap \an{\Gamma,K,s,\hat e}=\emptyset$,
  $\dom{s_g}\cap\fl{\Gamma,K,s,\hat e}=\emptyset$.
  \defqed
\end{definition}

\begin{lemma}\label{lem:lts-perm}
  Let $\pi_l$, $\pi_\alpha$, and $\pi_i$ be
  permutations on locations,  abstract names, and indices, respectively,
 and $\pi=\pi_l\pi_\alpha\pi_i$.
  If $C \trans{\eta} C'$ then
  $C\pi \trans{\eta\pi_\alpha\pi_i} C'\pi$.
\end{lemma}
\begin{proof}
    By nominal sets reasoning (all transition rules are closed under permutations).
\end{proof}
\begin{lemma}\label{lem:lts-permute}
  Let $C=\conf{A}{\Gamma}{K}{s}{e} \trans{\eta} \conf{A'}{\Gamma'}{K'}{s'}{e'}=C'$; then
  for all finite $L_0,A_0,I_0$ there exist $\pi_l,\pi_\alpha,\pi_i$ such that
\[ C \trans{\eta\pi} C'\pi \text{ and }
  (A'\sdif A)\cap A_0 = (\dom{s'\pi}\sdif\dom{s})\cap L_0 = (\dom{\Gamma'\pi}\sdif\dom{\Gamma})\cap I_0=\emptyset
\]
where $\pi=\pi_l\pi_\alpha\pi_i$.
\end{lemma}
\begin{proof}
  By \cref{lem:lts-perm}, picking permutations $\pi$ that rename new  names in $C'$  to fresh ones, and therefore such that $C\pi=C$.
\end{proof}
\begin{corollary}\label{lem:lts-permute-weak}
  Let $C=\conf{A}{\Gamma}{K}{s}{e} \wtrans{\eta} \conf{A'}{\Gamma'}{K'}{s'}{e'}=C'$; then
  for all finite $L_0,A_0,I_0$ there exist $\pi_l,\pi_\alpha,\pi_i$ such that
\[ C \wtrans{\eta\pi} C'\pi \text{ and }
  (A'\sdif A)\cap A_0 = (\dom{s'\pi}\sdif\dom{s})\cap L_0 = (\dom{\Gamma'\pi}\sdif\dom{\Gamma})\cap I_0=\emptyset
\]
where $\pi=\pi_l\pi_\alpha\pi_i$.
\end{corollary}
\begin{proof}
  By induction on the length of the transition from $C_1$, using \cref{lem:lts-permute}.
\end{proof}

\begin{figure*}[t] 

  \[\begin{array}{@{}c@{}}
     \irule*[UpToBeta][uptobeta]{
      C_1'
      \bisim*{R}
      C_2'
      \\\\
       C_1 \betatranss C_1'
      \\\\
      C_2 \betatranss C_2'
    }{
      C_1
      \utbeta*{\bisim{R}}
      C_2
    }
    \quad
    \irule*[UpToPerm][uptoperm]{
      C_1
      \bisim*{R}
      C_2
    }{
      C_1\pi_{l_1}\pi_\alpha\pi_i
      \utperm*{\bisim{R}}
      C_2\pi_{l_2}\pi_\alpha\pi_i
    }
    \quad
    \irule*[UpToGC][uptogc]{
      C_1\asymp
      \bisim*{R}
      \asymp C_2
    }{
      C_1
      \utgc*{\bisim{R}}
      C_2
    }
    \qquad
    \irule*[UpToGC$\bot$][uptogcbot]{
      C_1 \asymp
      \bisim*{R}
      \botconf
    }{
      C_1
      \utgc*{\bisim{R}}
      \botconf
    }
     \\[1.5em]
    \irule*[UpToWeakening][uptoweakening]{
      \conf{A_1}{\Gamma_1,\maps{i}{v_1}}{K_1}{s_1}{\hat e_1}
      \bisim*{R}
      \conf{A_2}{\Gamma_2,\maps{i}{v_2}}{K_2}{s_2}{\hat e_2}
    }{
      \conf{A_1}{\Gamma_1}{K_1}{s_1}{\hat e_1}
      \utweak*{\bisim{R}}
      \conf{A_2}{\Gamma_2}{K_2}{s_2}{\hat e_2}
    }
    \quad
     \irule*[UpToWeakening$\bot$][uptoweakeningbot]{
       \conf{A_1}{\Gamma_1,\maps{i}{v_1}}{K_1}{s_1}{\hat e_1}
      \bisim*{R}
      \botconf
    }{
      \conf{A_1}{\Gamma_1}{K_1}{s_1}{\hat e_1}
      \utweak*{\bisim{R}}
      \botconf
    }
  \end{array}\]
  \vspace{-1mm}
  \hrule
  \vspace{-1mm}
  \caption{Simple Up-to techniques.}\label{fig:utsimple}
\end{figure*}

The monotone functions on relations $\utperm{}$, $\utbeta{}$,
$\utgc{}$, and $\utweak{}$, as shown on \cref{fig:utsimple}, define the sound enhancement techniques: up to permutations, up to beta reductions, up to garbage collection, and up to weakening, respectively.
\begin{lemma}
   Functions $\utperm{}$,
  $\utbeta{}$,
  $\utgc{}$, and
  $\utweak{}$ are sound up-to techniques.
\end{lemma}
Soundness follows by the bisimulation enhancement technique \cite{PousS11,PousCompanion}
  (\cref{lem:prfs}), showing  ${\utperm{}} \progress {\utperm{}}$,
  $\utbeta{} \progress \utbeta{}$,
  $\utgc{}\progress\utgc{}\comp\utperm{}$, and
  $\utweak{}\progress\utweak{}\comp\utperm{}$.

\begin{lemma}\label{lem:upto-perm}
   Function $\utperm{}$ is a sound up-to technique.
\end{lemma}
\begin{proof}
  \todo{update with $\pi_i$}
  From \cref{lem:prfs}~(\ref{lem:prfs-1}), it suffices to show that $\utperm{}$ is compatible; i.e.,
  ${\utperm{\WP(\bisim{R})}} \sqsubseteq \WP(\utperm{\bisim{R}})$, for any configuration relation \bisim{R}.

  Let $C_1 \mathrel{\WP(\bisim{R})} C_2$ and 
  $C_1\pi_1 \utperm*{\bisim{R}} C_2\pi_2$,
  where $\pi_1=\pi_\alpha\pi_{l1}$ and $\pi_1=\pi_\alpha\pi_{l1}$.
  Moreover, let $C_1\pi_1 \trans{\eta} C_1'$.
  Because of $\pi_1 \pi_1 = \mathsf{id}$ and \cref{lem:lts-perm} we get
  $C_1 \trans{\eta\pi_\alpha} C_1'\pi_1$. By definition of $\WP(\bisim{R})$, there exists $C_2'$ such that
  $C_2 \wtrans{\eta\pi_\alpha} C_2'$ and $C_1'\pi_1 \bisim*{R} C_2'$. By \cref{lem:lts-perm}
  $C_2\pi_2 \wtrans{\eta} C_2'\pi_2$, and by definition of $\utperm{\bisim{R}}$:
  $C_1' \utperm*{\bisim{R}} C_2'\pi$.
\end{proof}
\begin{lemma}\label{lem:upto-beta}
  Function $\utbeta{}$ is a sound up-to technique.
\end{lemma}
\begin{proof}
  From \cref{lem:prfs}~(\ref{lem:prfs-1}), it suffices to show that $\utbeta{}$ is compatible; i.e.,
  ${\utbeta{\WP(\bisim{R})}} \sqsubseteq \WP(\utbeta{\bisim{R}})$, for any configuration relation \bisim{R}.
  Let $C_1 \utbeta*{\WP(\bisim{R})} C_2$
  and $C_1' \mathrel{\WP(\bisim{R})} C_2'$
  and $C_1 \betatranss C_1'$
  and $C_2 \betatranss C_2'$.
  We need to show that for all $C_1''$ such that
  $C_1 \trans{\eta} C_1''$
  there exists $C_2''$ such that
  $C_2 \wtrans{\eta} C_2''$ and $C_1'' \utbeta*{\bisim{R}} C_2''$.

  Let $C_1\trans{\eta} C_1''$. By definition of a $\beta$-move (\cref{def:beta-move}), $C_1''=C_1'$ and $\eta=\tau$ or 
  there exists $C_3$ such that $C_1' \trans{\eta} C_3$ and $C_1''\betatranss C_3$.
  In the former case the proof is trivial.
  In the latter case, by definition of $\WP(\bisim{R})$,
  there exists $C_4$ such that
  $C_2' \wtrans{\eta} C_4$ and $C_3 \bisim*{R} C_4$.
  Moreover, $C_2 \wtrans{\tau} C_2' \wtrans{\eta} C_4$, and
  $C_1'' \betatranss C_3 \bisim*{R} C_4$ which implies $C_1'' \utbeta*{\bisim{R}} C_4$, concluding the proof.
\end{proof}
\begin{lemma}\label{lem:utgc-sound} 
  Function $\utgc{}$ is a sound up-to technique.
\end{lemma}
\begin{proof}
  From \cref{lem:prfs}~(\ref{lem:prfs-2}) and \cref{lem:upto-perm}, it suffices to show that $\utgc{}\progress\utgc{}\comp\utperm{}$.

  Let $C_1 \utgc*{\WP(\bisim{R})} C_2$.
  By case analysis on this derivation we have two cases:

  {\flushleft
  \iref{uptogc}:} $C_1 \asymp C_3 \mathrel{\WP(\bisim{R})} C_4 \asymp C_2$.
  Consider $C_1 \trans{\eta} C_1'$.
  
  By \cref{lem:lts-permute}, there exists $\pi=\pi_\alpha\pi_{l1}$ such that 
  $C_1 \trans{\eta\pi_\alpha} C_1'\pi$
  and $(\an{C_1'\pi}\sdif\an{C_1})\cap \an{C_3,C_2} = \emptyset$
  and $(\fl{C_1'\pi}\sdif\fl{C_1})\cap \fl{C_3} = \emptyset$.
  
  By \cref{lem:asymp-bisim}, there exists $C_3'$ such that 
  $C_3 \trans{\eta\pi_\alpha} C_3'$ and $C_1'\asymp C_3'\pi$.
  
  By definition of \WP, there exists $C_4'$ such that $C_4 \wtrans{\eta\pi_\alpha} C_4'$ and $C_3' \bisim*{R} C_4'$.
  
  By \cref{lem:lts-permute-weak}
  there exists $\pi_{l2}$ such that $C_4 \wtrans{\eta\pi_\alpha} C_4'\pi_{l2}$ and
  $(\fl{C_4'\pi_{l2}}\sdif\fl{C_4})\cap \fl{C_2} = \emptyset$.

  By \cref{lem:lts-basic},
  $(\an{C_4'\pi_{l2}}\sdif\an{C_4})\cap \an{C_2} = \emptyset$.

  By \cref{lem:asymp-bisim-weak}, there exists $C_2'$ such that $C_2 \wtrans{\eta\pi_\alpha} C_2'$ and $C_2' \asymp C_4'\pi_l$.

  By \cref{lem:lts-perm}
  $C_2 \wtrans{\eta} C_2'\pi_\alpha$, and moreover, by \cref{lem:asymp-perm}, $C_2'\pi_\alpha \asymp C_4'\pi_l\pi_\alpha$.

  Therefore, $C_1' \asymp C_3\pi_\alpha\pi_{l1} \bisim*{R} C_4'\pi_\alpha\pi_{l2} \asymp C_2'$, from which we derive
  $C_1' \utgc*{\utperm{\bisim{R}}} C_2'$ as required.

  {\flushleft
  \iref{uptogcbot}:}
  $C_3 \bisim*{R} \botconf$
  and $C_1 \asymp C_3$
  and $C_2 = \botconf$.
  We proceed with the same reasoning as in the above case, with the exception that $C_4'=\botconf$.
\end{proof}
\begin{lemma}\label{lem:utweak-sound}
  Function $\utweak{}$ is a sound up-to technique.
\end{lemma}
\begin{proof}
  Similar to the preceding proof, using \cref{lem:lts-weaken}.
\end{proof}

The up to $\utbeta{}$ technique is useful in reducing the configurations
considered in bisimulation, focusing only on the configurations before 
the observable transitions of \cref{fig:lts}; $\tau$-transitions are all
beta transitions and can be considered all at once.

Many simple example equivalences have infinite transition systems without these up-to techniques, even when combined with the more sophisticated up-to techniques in the following sections. A simple example is the following.
\begin{example}
  Consider the equivalent functions
  $$M=\text{\lstinline{(fun () -> ref l = 0 in 5)}} \qquad\qquad
  M=\text{\lstinline{(fun () -> 5)}}$$
  Due to the allocation of \lstinline{l} the LTS of this equivalence is finite only using the up to $\utgc{}$ technique.
\end{example}

\subsection{Proof of \cref{lem:utsepconj-sound}} \label{sec:uptosep-proof}
\begin{proof}
  By \cref{lem:prfs}~(\ref{lem:prfs-3}) it suffices to show
  $\utsepconj{}\comp{\constf{\gfp\WP}}\progress(\utsepconj{}\cup g)^\omega\comp{\constf{\gfp\WP}}$,
  where $g=\utId{}\sqcup\utperm{}$.
  We need to show
  $\utsepconj{}\comp{\constf{\gfp\WP}}\comp\WP(\bisim{R}) \subseteq 
  {\WP}\comp (\utsepconj{}\cup g)^\omega\comp{\constf{\gfp\WP}}(\bisim{R})$.
  Because $\constf{\gfp\WP}$ is the constant function mapping its argument to $(\simil)$, we need to show
  $\utsepconj{}(\simil) \subseteq 
  {\WP}\comp (\utsepconj{}\cup g)^\omega(\simil)$.
  We proceed by induction on the derivation of $C_1 \utsepconj*{\simil} C_2$.

  Let $C_1\utsepconj*{\simil}C_2$, 
  and $C_1\trans{\eta}C_1'$.
  If $C_1\trans{\eta}C_1'$ is produced by rule \iref{dummy} then the proof is trivial as $C_2$ can perform the same transition and $C_1'=\botconf \utId*{\simil} \botconf=C_2'$. We thus consider only non-\iref{dummy} transitions from $C_1$.
  We proceed by case analysis on the derivation
  $C_1\utsepconj*{\simil}C_2$. There are three cases: \iref{uptosepconj}, \iref{uptosepconjbotl} and \iref{uptosepconjbotr}.
  We only show the proof for the first case; the last two cases are handled in a similar but simpler fashion.

  W.l.o.g. we let $i=1$ and we have
  \begin{align*}
    C_1 &= C_3 \sepconj[i][\vec k] C_5 = \conf{A_3\cup A_5}{\Gamma_3,\Gamma_5}{K_1\interleave[\vec k]K_3}{s_1, s_3}{\hat e_3}\\
    C_2 &= C_4 \sepconj[i][\vec k] C_6 = \conf{A_3\cup A_5}{\Gamma_4,\Gamma_6}{K_2\interleave[\vec k]K_4}{s_2, s_4}{\hat e_4}\\
    C_3 &= \conf{A_3}{\Gamma_3}{K_3}{s_3}{\hat e_3} \simil \conf{A_3}{\Gamma_4}{K_4}{s_4}{\hat e_4} = C_4
    &\dom{\Gamma_3}&=\dom{\Gamma_4} \\
    C_5 &= \conf{A_5}{\Gamma_5}{K_5}{s_5}{\hat e_5} \simil \conf{A_5}{\Gamma_6}{K_6}{s_6}{\hat e_6} = C_6
    &\dom{\Gamma_5}&=\dom{\Gamma_6}
  \end{align*}
  We proceed by cases on the transition $C_1\trans{\eta}C_1'$.

  {\flushleft
  \iref{propappc},
  \iref{propappf},
  \iref{propretc},
  \iref{propretf},
  \iref{tau}:}
  In all proponent transitions the proof is similar. We show the case \iref{propappf}.
  In this case $e_3=E\hole[\app{\alpha} v]$ and
  \begin{align*}
    C_1 = \conf{A_3\cup A_5}{\Gamma_3,\Gamma_5}{K_3\interleave[\vec k]K_5}{s_3, s_5}{E\hole[\app\alpha v]}
    &\trans{\lpropapp{\alpha}{i}}
    \conf{A_3\cup A_5}{\Gamma_3,\Gamma_5,\maps{i}{v}}{K_3\interleave[\vec k]K_5}{s_3, s_5}{\noe}
    =C_1'
    \\&\qquad\qquad\text{ and }
    i\not\in\dom{\Gamma_3,\Gamma_5}
    \\
    C_3 = \conf{A_3}{\Gamma_3}{K_3}{s_3}{E\hole[\app\alpha v]}
    &\trans{\lpropapp{\alpha}{i}}
    \conf{A_3}{\Gamma_3,\maps{i}{v}}{K_3}{s_3}{\cdot} = C_3'
  \end{align*}
  Moreover, $C_1'=C_3'\sepconj[1][\vec k]C_5$.
  By $C_3 \simil C_4$ there are two possibilities:
  \begin{itemize}
    \item $C_4 \wtrans{\lpropapp{\alpha}{i}} \botconf$ and $C_3' \simil \botconf$.
      The proof in this case is completed by $C_2' \trans{\lpropapp{\alpha}{i}} \botconf$ and
      $C_1 \utsepconj*{\simil} \botconf$ by either rule \iref{uptosepconjbotl} or \iref{uptosepconjbotr}.

    \item
    $C_4 = \conf{A_3}{\Gamma_4}{K_4}{s_4}{\hat e_4}
    \wtrans{\tau}\trans{\lpropapp{\alpha}{i}}
    \conf{A_3}{\Gamma_4,\maps{i}{v_4}}{K_4}{s_4'}{\noe} = C_4'$
      and
    $C_3' \simil C_4'$. By
    $\dom{\Gamma_3,\Gamma_5}=\dom{\Gamma_4,\Gamma_6}$ and \cref{lem:lts-taus} we derive:
    \[
      C_2 = \conf{A_3\cup A_5}{\Gamma_4,\Gamma_6}{K_4\interleave[\vec k]K_6}{s_4, s_6}{\hat e_4}
      \wtrans{\lpropapp{\alpha}{i}}
      \conf{A_3\cup A_5}{\Gamma_4,\Gamma_6,\maps{i}{v_4}}{K_4\interleave[\vec k]K_6}{s_4', s_6}{\noe}
      =C_2'
    \]
   Moreover,
   $C_2'=C_4'\sepconj[1][\vec k]C_6$ and
   $C_1' \utsepconj*{\simil} C_2'$.
  \end{itemize}

{\flushleft
\iref{opappc},\iref{opappf}:}
We show the case for \iref{opappf}. In this case $\eta=\lopapp{i}{\alpha}$, and  $\hat e_3=\noe$, and by \cref{lem:simil-properties}, $\hat e_4=\noe$.
Moreover $i\in\dom{\Gamma_3}$ or $i\in\dom{\Gamma_5}$.
If it is the former, then $C_3 \trans{\lopapp{i}{\alpha}} C_3'$ and $C_1'=C_3' \sepconj[1][\vec k] C_5$, otherwise
$C_5 \trans{\lopapp{i}{\alpha}} C_5'$ and $C_1'=C_3 \sepconj[2][\vec k] C_5'$. Moreover by the simulation we can show that
$C_4 \wtrans{\lopapp{i}{\alpha}} C_4'$ or
$C_6 \wtrans{\lopapp{i}{\alpha}} C_6'$, respectively.
In both cases $C_2 \wtrans{\lopapp{i}{\alpha}} C_2'$ and $C_2'=C_4'\sepconj[1][\vec k] C_6$ or
 $C_2'=C_4\sepconj[2][\vec k] C_6'$, respectively.
 Therefore $C_1' \utsepconj*{\simil} C_2'$.

{\flushleft
\iref{opretc},\iref{opretf}:}
We show the cases for \iref{opretf}. Here $\vec k$ is either $(1,\vec k')$ or $(2, \vec k')$. We consider the former case, the latter is symmetric. In this case $K_3 = E, K_3'$ and $K_3\interleave[\vec k]K_5= E,(K_3'\interleave[\vec k'] K_5)$. Therefore we have:
  \begin{align*}
    C_1 = \conf{A_3\cup A_5}{\Gamma_3,\Gamma_5}{E,(K_3'\interleave[\vec k'] K_5)}{s_3, s_5}{\noe}
    &\trans{\lopret{\alpha}}
    \conf{A_3\cup A_5\uplus \alpha}{\Gamma_3,\Gamma_5}{K_3'\interleave[\vec k']K_5}{s_3, s_5}{E\hole[\alpha]}
    =C_1'
    \\&\qquad\qquad\text{ and }
    \alpha\not\in\dom{A_3,A_5}
    \\
    C_3 = \conf{A_3}{\Gamma_3}{K_3}{s_3}{\noe}
    &\trans{\lopret{\alpha}}
    \conf{A_3\uplus\alpha}{\Gamma_3}{K_3}{s_3}{\cdot} = C_3'
  \end{align*}
  By $C_3 \simil C_4$ there are two possibilities:
  \begin{itemize}
    \item $C_4 \wtrans{\lopret{\alpha}} \botconf$ and $C_3' \simil \botconf$.
      The proof in this case is completed by $C_2' \trans{\lopret{\alpha}} \botconf$ and
      $C_1 \utsepconj*{\simil} \botconf$ by either rule \iref{uptosepconjbotl} or \iref{uptosepconjbotr}.

    \item
      $C_4 = \conf{A_3}{\Gamma_4}{K_4}{s_4}{\noe}
    \wtrans{\tau}\trans{\lopret{\alpha}}
      \conf{A_3\uplus\alpha}{\Gamma_4,\maps{i}{v_4}}{K_4'}{s_4'}{E_4\hole[\alpha]} = C_4'$
      and
    $C_3' \simil C_4'$
      and $K_4= E_4,K_4'$.
      Thus $K_4\interleave[\vec k]K_6=K_4\interleave[1,\vec k']K_6=E_4,(K_4'\interleave[\vec k']K_6)$.
      By \cref{lem:lts-taus} we derive:
      \begin{align*}
        C_2 &= \conf{A_3\cup A_5}{\Gamma_4,\Gamma_6}{E_4,(K_4'\interleave[\vec k']K_6)}{s_4, s_6}{\noe}
        \\&\wtrans{\lpropapp{\alpha}{i}}
      \conf{A_3\cup A_5\uplus\alpha}{\Gamma_4,\Gamma_6}{K_4'\interleave[\vec k']K_6}{s_4', s_6}{E_4\hole[\alpha]}
      =C_2'
        \end{align*}
   Moreover,
   $C_2'=C_4'\sepconj[1][\vec k]C_6$ and
   $C_1' \utsepconj*{\simil} C_2'$.
  \end{itemize}

{\flushleft
\iref{tau}:}
In this case we have 
  \begin{align*}
    C_1 = C_3 \sepconj[1][\vec k] C_5
    &\trans{\tau}
    C_3' \sepconj[1][\vec k] C_5 = C_1'
    \\
    C_3 &\trans{\tau} C_3'
  \end{align*}
  By $C_3 \simil C_4$ we have
    $C_4 \wtrans{\tau} C_4'=\botconf$
    or
    $C_4 \wtrans{\tau} C_4'\not=\botconf$
    and
    $C_3' \simil C_4'$.
    In the former case the proof is completed by $C_2 \trans{\tau}\botconf$ and
    $C_1' \utsepconj*{\simil} \botconf$ by rule \iref{uptosepconjbotl}.
    In the latter case, by \cref{lem:lts-permute-weak},
  there exists $\pi_{l4}$ such that $C_4 \wtrans{\tau} C_4'\pi_{l4}$ and
  $(\fl{C_4'\pi_{l2}}\sdif\fl{C_4})\cap \fl{C_6} = \emptyset$.
  We have $C_3'\utperm*{\simil}C_4'\pi_{l4}$. Moreover we derive
  $C_2 = C_4 \sepconj[1][\vec k] C_6 \wtrans{\tau} C_4'\pi_{l4} \sepconj[1][\vec k] C_6 = C_2'$
  and $C_1' \utsepconj*{\utperm{\simil}} C_2'$.

{\flushleft \iref{term}:}
In this case we have $K_3=K_5=\emptyK$ and $\hat e_3=\hat e_5=\noe$. Therefore $C_3 \trans{\lterm} \botconf$ and
 $C_5 \trans{\lterm} \botconf$. Therefore by $C_3 \simil C_4$ and $C_5 \simil C_6$, we have
 $C_4 \trans{\lterm} \botconf$ and $C_6 \trans{\lterm} \botconf$. Therefore $K_4=K_6=\emptyK$ and
 $\hat e_4 =\hat e_6=\noe$, and thus $C_2 \trans{\lterm}\botconf$. The resulting \botconf configurations are related by $\utId{\simil}$.
\end{proof}

%% file: soundness.tex
\subsection{Language Lemmas}
The following lemmas hold for \lang extended with abstract names.

\begin{lemma}[Unique Decomposition]\label{lem:unique-decomp}
  Let $e=E\hole[e]$ and $\redconf{s}{e} \redbase $.
  Then for any $E'$ and $e'$ such that
  $e=E'\hole[e']$ and $\redconf{s}{e'} \redbase $, we have $E=E'$ and $e=e'$.
\end{lemma}
\begin{proof}
  By induction on $E$.
\end{proof}

\begin{lemma}\label{lem:red-perm}
  Let $\pi_l,\pi_\alpha,\pi_i$ be permutations on locations,  abstract names and indices respectively, and $\pi=\pi_l\pi_\alpha\pi_i$.
  If $\redconf{s}{e} \redbase \redconf{s'}{e'}$ then
  $\redconf{s\pi}{e\pi} \redbase \redconf{s'\pi}{e'\pi}$.
  Moreover,
  if $\redconf{s}{e} \red \redconf{s'}{e'}$ then
  $\redconf{s\pi}{e\pi} \red \redconf{s'\pi}{e'\pi}$.
\end{lemma}
\begin{proof}
  By nominal sets reasoning (all reduction rules are closed under permutations).
\end{proof}
\begin{lemma}\label{lem:redbase-subst} 
  \todo{unused}
  Let $\sigma=\concr{\vec\alpha}{\vec v}$, where $\vec v$ are closed $\lambda$-abstractions with $\fl{\vec v}\subseteq\dom{s}$
  and $\redconf{s\sigma}{e\sigma} \redbase \redconf{s'}{e'}$. Then one of the following holds:
  \begin{enumerate}
    \item $e = (\app c c')$ and $e'=w\sigma$ and $s'=s\sigma$ and $c^{\textsf{arith}}(c') = w$;
    \item $e =(\new l u e')$ and $s'=(s\stupd{l}{u})\sigma$ and $l\not\in\dom{s}$;
    \item $e =\deref l$ and $e'=u\sigma$ and $s'=s\sigma$ and $\sigma(l)=u$;
    \item $e =l \asgn u$ and $e'=\true\sigma$ and $s'=(s\stupd{l}{u})\sigma$;
    \item $e =(\app{(\lam x e'')} u)$ and $e'=(e''\sub{x}{u})\sigma$ and $s'=s$;
    \item $e =(\app \alpha u)$ and $e'=(e''\sub{x}{u})\sigma$ and $s'=s$ and $\sigma(\alpha)=(\lam x e'')$.
  \end{enumerate}
\end{lemma}
\begin{proof}
  By case analysis on the transition.\
\end{proof}
\begin{lemma}\label{lem:red-subst} 
  \todo{unused}
  Let $\sigma=\concr{\vec\alpha}{\vec v}$, where $\vec v$ are closed $\lambda$-abstractions with
  $\an{\vec v}\cap\{\vec \alpha\}=\emptyset$
  and $\fl{\vec v}\subseteq\dom{s}$
  and $\redconf{s\sigma}{e\sigma} \red \redconf{s'}{e'}$. Then there exists $s''$ such that 
  $s'=s''\sigma$ and
  one of the following holds:
  \begin{enumerate}
    \item there exists $e''$ such that  and $e'=e''\sigma$ and $\redconf{s}{e}\red\redconf{s'}{e'}$; or
    \item there exist $E$, $u$ and $\alpha\in\vec\alpha$ such that
      $e =E\hole[\app \alpha u]$ and $e'=(E\hole[e''\sub{x}{u}])\sigma$ and $\sigma(\alpha)=\lam x e''$.
  \end{enumerate}
\end{lemma}
\begin{proof}
  By definition of the transition, using \cref{lem:redbase-subst}.
\end{proof}

\subsection{LTS Lemmas}
  
\begin{lemma}\label{lem:lts-basic}
  Let $\conf{A}{\Gamma}{K}{s}{e} \trans{\eta} \conf{A'}{\Gamma'}{K'}{s'}{e'}$; then
  \begin{enumerate}
    \item if $\eta\not\in\{\lopapp{i}{\alpha},\lopret{\alpha}\where\text{any } i, \alpha\}$ then $A=A'$;
    \item if $\eta=\lopapp{i}{\alpha}$ or $\eta=\lopret{\alpha}$ then $A\uplus\alpha=A'$;
    \item if $\eta\not\in\{\lpropapp{\alpha}{i},\lpropret{i}\where\text{any } i, \alpha\}$ then $\dom{\Gamma}=\dom{\Gamma'}$;
    \item if $\eta=\lpropapp{\alpha}{i}$ or $\eta=\lpropret{i}$ then $\dom{\Gamma}\uplus i=\dom{\Gamma'}$;
    \item $\dom{s} \subseteq\dom{s'}$.
  \end{enumerate}
\end{lemma}
\begin{proof}
  By cases on the transition.
\end{proof}

\begin{lemma}\label{lem:asymp-perm}
  If $C \asymp C'$ then $C\pi \asymp C'\pi$.
  \qed
\end{lemma}
\begin{lemma}\label{lem:asymp-bisim}
  Let $C_1\asymp C_2$
  and $C_1 \trans{\eta} C_1'$; then
  $C_2 \trans{\eta} C_2'$ and $C_1' \asymp C_2'$, provided that
      $(\an{C_1'}\sdif\an{C_1})\cap \an{C_2} = \emptyset$ and
      $(\fl{C_1'}\sdif\fl{C_1})\cap \fl{C_2} = \emptyset$.
\end{lemma}
\begin{proof}
  By induction on the derivation of $C_1 \asymp C_2$ and case analysis on the transition from $C_1$.
\end{proof}
\begin{corollary}\label{lem:asymp-bisim-weak}
  Let $C_1\asymp C_2$
  and $C_1 \wtrans{\eta} C_1'$; then
  $C_2 \wtrans{\eta} C_2'$ and $C_1' \asymp C_2'$, provided that
      $(\an{C_1'}\sdif\an{C_1})\cap \an{C_2} = \emptyset$ and
      $(\fl{C_1'}\sdif\fl{C_1})\cap \fl{C_2} = \emptyset$.
\end{corollary}
\begin{proof}
  By induction on the length of the transition from $C_1$ and \cref{lem:lts-basic,lem:asymp-bisim}.
\end{proof}


\begin{lemma}
  Any transition $C\trans{\tau}C'$ is a $\beta$-move, provided $C'=\botconf$ implies $C=\botconf$.
\end{lemma}
\begin{proof}
  Case analysis on the transition relation gives us two cases:
  the transition is derived either by the \iref{dummy} or \iref{tau} rule.
  The former case is trivial because $C=C'$. 
  The latter is also trivial because unique decomposition
  (\cref{lem:unique-decomp}) implies that transitions derived by the \iref{tau}
  rule can only perform that transition and the transition derived by
  \iref{dummy}. The \iref{tau}-transition satisfies the first condition of \cref{def:beta-move} and the \iref{dummy}-transition satisfies the second condition of the same definition.
\end{proof}

\begin{lemma}\label{lem:lts-weaken}
  Let $C_1 = \conf{A}{\Gamma}{K}{s}{\hat e}$
  and $C_2 = \conf{A}{\Gamma,\maps i v}{K}{s}{\hat e}$ be well formed configurations. Then the following hold:
  \begin{enumerate}
    \item
      If $C_1 \trans{\eta} C_1' = \conf{A'}{\Gamma'}{K'}{s'}{\hat e'}$,
      where $\eta\not\in\{\lpropapp{\alpha}{i},\lpropret{i}\where \text{any } \alpha\}$,
      then
      $$C_2 \trans{\eta} \conf{A'}{\Gamma',\maps i v}{K'}{s'}{\hat e'}.$$
    \item
      If $C_1 \trans{\eta} C_1'=\conf{A'}{\Gamma,\maps{i}{u}}{K'}{s'}{\hat e'}$
      where $\eta=\lpropapp{\alpha}{i}$ or $\eta=\lpropret{i}$,
      then $$C_2 \trans{\eta'} \conf{A'}{\Gamma,\maps{j}{u}}{K'}{s'}{\hat e'}$$
      where $\eta'=\lpropapp{\alpha}{j}$ or $\eta'=\lpropret{j}$, respectively, and
      $j\not=i$.
    \item
      If $C_2 \trans{\eta} \conf{A'}{\Gamma',\maps i v}{K'}{s'}{\hat e'}$,
      where $\eta\not\in\{\lopapp{i}{\alpha}\where \text{any } \alpha\}$,
      then $$C_1 \trans{\eta} \conf{A'}{\Gamma,\maps{i}{u}}{K'}{s'}{\hat e'}$$

  \end{enumerate}
\end{lemma}
\begin{proof}
By case analysis on the transitions.
\end{proof}

\begin{lemma}\label{lem:lts-concr}
  Let $C_1 = \conf{A}{\Gamma}{K}{s}{\hat e}$
  and $C_1\concr{\alpha}{v}$ be well-formed configurations;
  then the following hold:
  \begin{enumerate}
    \item
      If $C_1 \trans{\tau} C_1'$
      then
      $C_1\concr{\alpha}{v} \trans{\tau} C_1'\concr{\alpha}{v}$.
    \item
      If $C_1\concr{\alpha}{v} \trans{\tau} C_2'$
     and  $\hat e\not= E\hole[\app\alpha u]$ (for any $E,u$)
      then there exists $C_1'$ such that $C_2'=C_1'\concr{\alpha}{v}$
      and $C_1 \trans{\tau} C_1'$.

  \end{enumerate}
\end{lemma}
\begin{proof}
  By case analysis on the transitions.
\end{proof}

\begin{lemma}\label{lem:lts-taus}
  Let $C_1 = \conf{A_1}{\Gamma_1}{K_1}{s}{e} \wtrans{\tau} \conf{A_1'}{\Gamma_2'}{K_1'}{s'}{e'}$.
  For any $A_2$, $\Gamma_2$, $K_2$, and $E$ such that the configurations
  $C_2 = \conf{A_2}{\Gamma_2}{K_2}{s}{e}$ and
  $C_2' = \conf{A_2}{\Gamma_2}{K_2}{s}{E[e]}$ are well-formed:
  \[
    C_2 \wtrans{\tau} \conf{A_2}{\Gamma_2}{K_2}{s'}{e'}
    \qquad\text{and}\qquad
    C_2' \wtrans{\tau} \conf{A_2}{\Gamma_2}{K_2}{s'}{E[e']}
  \]
\end{lemma}
\begin{proof}
  By case analysis on the transitions.
\end{proof}
\subsection{Simple Simulation Results}

\begin{lemma}[Equivalent Knowledge Environments]\label{lem:simil-properties}
  Consider $C_1 = \conf{A_1}{\Gamma_1}{K_1}{s_1}{\hat e_1}$
  and $C_2 = \conf{A_2}{\Gamma_2}{K_2}{s_2}{\hat e_2}$ with $C_1 \simil C_2$.
  \begin{enumerate}
    \item $\hat e_1 = \noe$ if and only if $\hat e_2=\noe$.
    \item If there exists a trace $C_1 \trans{t\lterm}$ then $\sizeof{K_1}=\sizeof{K_2}$.
    \item $\dom{\Gamma_1} \subseteq \dom{\Gamma_2}$
    \item If $C_1 \trans{\eta} C_1'$ and $\eta \in \{\lpropapp{\alpha}{i}, \lpropret{i}\where \text{ some } i ,\alpha\}$
      then $\dom{\Gamma_1}=\dom{\Gamma_2}$.
    \item If $\Gamma_1 \not=\emptyG$ or $K_1\not=\emptyK$ then $A_1 \supseteq A_2$.
      \qed
  \end{enumerate}
\end{lemma}

\subsection{Up to Fold}

\begin{definition}[Concretisation]
  A \emph{concretisation} $\concr{\alpha}{v}$
  is defined when $\alpha\not\in\an{v}$, and we
  write $e\concr{\alpha}{v}$ for the expression obtained after substituting $\alpha$ for $v$ in $e$.
  We let $\dom{\concr{\alpha}{v}}=\{\alpha\}$ and $\rng{\concr{\alpha}{v}}=\{v\}$,
  and let $\kappa$ range over concretisations.
  We lift concretisation to contexts, environments and stores point-wise; we also lift it to configurations, writing
  $C\kappa$
  to mean $\conf{A\sdif\dom{\kappa}}{\Gamma\kappa}{K\kappa}{s\kappa}{e\kappa}$, when $C=\conf{A}{\Gamma}{K}{s}{e}$.
  \defqed
\end{definition}

\begin{figure*}[t] 
  \[\begin{array}{@{}c@{}}
    \irule[UTFe1][utfe1]{
      \conf{A_1}{\Gamma_1}{E_1,K_1}{s_1}{e_1}
      \bisim*{R}
      \conf{A_2}{\Gamma_2}{E_2,K_2}{s_2}{e_2}
    }{
      \conf{A_1}{\Gamma_1}{K_1}{s_1}{E_1\hole[e_1]}
      \utfolde*{\bisim{R}}
      \conf{A_2}{\Gamma_2}{K_2}{s_2}{E_2\hole[e_2]}
    }
    \\\\
    \irule[UTFe2][utfe2]{
      \conf{A_1}{\Gamma_1}{K_{11},E_1,E'_1,K_{12}}{s_1}{\hat e_1}
      \bisim*{R}
      \conf{A_2}{\Gamma_2}{K_{21},E_2,E'_2,K_{22}}{s_2}{\hat e_2}
      \\
      \sizeof{K_{11}}=\sizeof{K_{21}}
    }{
      \conf{A_1}{\Gamma_1}{K_{11},E'_1\hole[E_1],K_{12}}{s_1}{\hat e_1}
      \utfolde*{\bisim{R}}
      \conf{A_2}{\Gamma_2}{K_{21},E'_2\hole[E_2],K_{22}}{s_2}{\hat e_2}
    }
    \\\\
    \irule[UTFd1][utfd1]{
      \conf{A_1}{\Gamma_1}{E_1,K_1}{s_1}{e_1}
      \bisim*{R}
      \botconf
    }{
      \conf{A_1}{\Gamma_1}{K_1}{s_1}{E_1\hole[e_1]}
      \utfolde*{\bisim{R}}
      \botconf
    }
    \quad
    \irule[UTFd2][utfd2]{
      \conf{A_1}{\Gamma_1}{K_{11},E_1,E'_1,K_{12}}{s_1}{\hat e_1}
      \bisim*{R}
      \botconf
    }{
      \conf{A_1}{\Gamma_1}{K_{11},E'_1\hole[E_1],K_{12}}{s_1}{\hat e_1}
      \utfolde*{\bisim{R}}
      \botconf
    }
  \end{array}\]
    \Hline
  \[\begin{array}{@{}c@{}}
     \irule[UTFv][utfv]{
       \conf{A_1\uplus\alpha}{\Gamma_1,\maps{i}{v_1}}{K_1}{s_1}{\hat e_1}
      \bisim*{R}
      \conf{A_2\uplus\alpha}{\Gamma_2,\maps{i}{v_2}}{K_2}{s_2}{\hat e_2}
      \\\\
      \kappa_1=\concr{\alpha}{v_1}
      \\
      \kappa_2=\concr{\alpha}{v_2}
      \\
      \alpha\not\in\an{v_1,v_2}
    }{
      \conf{A_1}{\Gamma_1\kappa_1}{K_1\kappa_1}{s_1\kappa_1}{\hat e_1\kappa_1}
      \utfoldv*{\bisim{R}}
      \conf{A_2}{\Gamma_2\kappa_2}{K_2\kappa_2}{s_2\kappa_2}{\hat e_2\kappa_2}
    }
   \\\\
    \irule[UTFd3][utfd3]{
      \conf{A_1\uplus\alpha}{\Gamma_1,\maps{i}{v_1}}{K_1}{s_1}{\hat e_1}
      \bisim*{R}
      \botconf
      \\
      \kappa_1=\concr{\alpha}{v_1}
      \\
      \alpha\not\in\an{v_1}
    }{
      \conf{A_1}{\Gamma_1\kappa_1}{K_1\kappa_1}{s_1\kappa_1}{\hat e_1\kappa_1}
      \utfoldv*{\bisim{R}}
      \botconf
    }
  \end{array}\]
  \hrule
  \caption{Up-to fold.}\label{fig:utfold}
\end{figure*}

\begin{lemma}\label{lem:utfolde-sound} 
  Let $f= \utfolde{}\sqcup\utfoldv{}$; then
  $\utfolde{}\comp{\constf{\gfp\WP}}\progress f^\omega\comp{\constf{\gfp\WP}}$.
\end{lemma}
\begin{proof}
  We need to show
  $\utfolde{}\comp{\constf{\gfp\WP}}\comp\WP(\bisim{R}) \sqsubseteq 
  {\WP}\comp f^\omega\comp{\constf{\gfp\WP}}(\bisim{R})$.
  Because $\constf{\gfp\WP}$ is the constant function mapping its argument to $(\simil)$, we need to show
  $\utfolde{}(\simil) \sqsubseteq 
  {\WP}\comp f^\omega(\simil)$.

  Let $C_1\utfolde*{\simil}C_2$, 
  and $C_1\trans{\eta}C_1'$.
  If $C_1\trans{\eta}C_1'$ is produced by rule \iref{dummy} then the proof is trivial as $C_2$ can perform the same transition and $C_1'=\botconf \mathrel{f^0(\simil)} \botconf=C_2'$. We thus consider only non-\iref{dummy} transitions from $C_1$.
  We proceed by case analysis on the derivation
  $C_1\utfolde*{\simil}C_2$.

  {\flushleft
  \iref{utfe1}:}
  $C_1=\conf{A_1}{\Gamma_1}{K_1}{s_1}{E_1\hole[e_1]}$ and
  $C_2=\conf{A_2}{\Gamma_2}{K_2}{s_2}{E_2\hole[e_2]}$ and
  $C_3=\conf{A_1}{\Gamma_1}{E_1,K_1}{s_1}{e_1}$ and
  $C_4=\conf{A_2}{\Gamma_2}{E_2,K_2}{s_2}{e_2}$ and
  $C_3 \mathrel{\simil} C_4$.
  In this case $C_1$ is a proponent configuration thus the transition can only be produced by rules
  \iref{propappc}, \iref{propappf}, \iref{propretc}, \iref{propretf}, and \iref{tau}.

  \begin{itemize}
    \item
      \iref{propappc} and
      \iref{propappf}: Here we show only the case for the latter rule.
      We have $\eta=\lpropapp{\alpha}{i}$ and
      $E_1\hole[e_1] = E\hole[\app \alpha v]$
      and
      $C'_1=\conf{A_1}{\Gamma_1,\maps{i}{v_1}}{E_1';K_1}{s_1}{\noe}$.
      We have the following three cases:
      \begin{itemize}
        \item $e_1=E_1'\hole[\app\alpha v_1]$, $E=E_1\hole[E_1'\hole]]$:
          In this case configuration $C_3$ can perform the same transition:
          $C_3\trans{\eta}C_3'=\conf{A_1}{\Gamma_1,\maps{i}{v_1}}{E_1',E_1,K}{s_1}{\noe}$.
          Because $C_3\simil C_4$,
          $C_4\wtrans{\eta}C_4'=\conf{A_2}{\Gamma_2'}{E_2',E_2,K}{s_2'}{\noe}$ or
          $C_4\wtrans{\eta}C_4'=\botconf$, and $C_3'\simil C_4'$.
          In the case where $C_4$ goes to \botconf, $C_2\trans{\eta}\botconf$ and
          $C_1' \utfolde*{\simil} \botconf$ by rule \iref{utfd2}.
          In the other case we derive $C_2\wtrans{\eta}C_2'=\conf{A_2}{\Gamma_2'}{E_2\hole[E_2']}{s_2'}{\noe}$
          and $C_1' \utfolde*{\bisim{S}} C_2'$ by rule \iref{utfe2}.

        \item $e_1=\alpha$, $E_1=E_1'\hole[\app\hole v_1]$
      and $C'_1=\conf{A_1}{\Gamma_1,\maps{i}{v_1}}{E_1';K_1}{s_1}{\noe}$:

          In this case configuration $C_3$ can perform the transitions:
          \begin{align*}
            C_3
            &\trans{\lpropret{j}}
            \conf{A_1}{\Gamma_1,\maps{j}{\alpha}}{E_1'\hole[\app\hole v_1],K_1}{s_1}{\noe}
            \\&\trans{\lopret{\alpha'}}
            \conf{A_1\uplus\alpha'}{\Gamma_1,\maps{j}{\alpha}}{K_1}{s_1}{E_1'\hole[\app\alpha' v_1]}
            \\&\trans{\lpropapp{\alpha'}{j'}}
            \conf{A_1\uplus\alpha'}{\Gamma_1,\maps{j}{\alpha},\maps{j'}{v_1}}{E_1',K_1}{s_1}{\noe}
            &(\alpha'\not\in\an{\Gamma_1,K_1,s_1,E_1',v_1})
            \\&\trans{\lopapp{j}{\alpha''}}
            \conf{A_1\uplus\alpha',\alpha''}{\Gamma_1,\maps{j}{\alpha},\maps{j'}{v_1}}{E_1',K_1}{s_1}{\app\alpha{\alpha''}}
            &(\alpha''\not\in\an{\Gamma_1,K_1,s_1,E_1',v_1})
            \\&\trans{\lpropapp{\alpha}{i}}
            \conf{A_1\uplus\alpha',\alpha''}{\Gamma_1,\maps{j}{\alpha},\maps{j'}{v_1},\maps{i}{\alpha''}}{
              \hole,E_1',K_1}{s_1}{\noe}
            =C_3'
          \end{align*}
          Because $C_3\simil C_4$, there exists $C_4'$ such that
          \[
            C_4
            \wtrans{\lpropret{j}}
            \wtrans{\lopret{\alpha'}}
            \wtrans{\lpropapp{\alpha'}{j'}}
            \wtrans{\lopapp{j}{\alpha''}}
            \wtrans{\lpropapp{\alpha}{i}}
            C_4'
          \]
          and $C_3' \simil C_4'$. By analysis of the transitions we have two possibilities:
          \begin{itemize}
            \item $C_4'=\botconf$ and in this case we have
              $C_2 \trans{\eta}\botconf$.
             Moreover:
             \begin{align*}
               C_3'&=
               \conf{A_1\uplus\alpha',\alpha''}{\Gamma_1,\maps{j}{\alpha},\maps{j'}{v_1},\maps{i}{\alpha''}}{
               \hole,E_1',K_1}{s_1}{\noe}
               \simil \botconf
              \\
               &\mathrel{\phantom{=}}
               \conf{A_1\uplus\alpha',\alpha''}{\Gamma_1,\maps{i}{\alpha''}}{
               \hole,E_1',K_1}{s_1}{\noe}
               \concr{\alpha''}{v_1}
               \concr{\alpha'}{\alpha}
               \utfoldv*{\utfoldv{\simil}}
               \botconf
              &\text{(\iref{utfd3})}
              \\
               &\mathrel{\phantom{=}}
               \conf{A_1}{\Gamma_1,\maps{i}{v_1}}{
               \hole,E_1',K_1}{s_1}{\noe}
               \utfoldv*{\utfoldv{\simil}}
               \botconf
               &\hspace{18ex}
               \llap{$(\alpha',\alpha''\not\in\an{\Gamma_1,K_1,s_1,E_1',v_1})$}
              \\
               C_1'&=\conf{A_1}{\Gamma_1,\maps{i}{v_1}}{
               E_1'\hole,K_1}{s_1}{\noe}
              \utfolde*{\utfoldv{\utfoldv{\simil}}}
               \botconf
              &\text{(\iref{utfe2})}
              \end{align*}

            \item
              Otherwise we have:
             \begin{align*}
               C_4 & =\conf{A_2}{\Gamma_2}{E_2,K_2}{s_2}{e_2}
               \wtrans{\tau}
              \conf{A_2}{\Gamma_2}{E_2,K_2}{s_2'}{w_2}
              \\& \trans{\lpropret{j}}
               \conf{A_2}{\Gamma_2,\maps{j}{w_2}}{E_2,K_2}{s_2'}{\noe}
               \\&\trans{\lopret{\alpha'}}
              \conf{A_2\uplus\alpha'}{\Gamma_2,\maps{j}{w_2}}{K_2}{s_2'}{E_2\hole[\alpha']}
               &(\alpha'\not\in\an{\Gamma_2,K_2,w_2})
              \\& \wtrans{\tau}
              \conf{A_2\uplus\alpha'}{\Gamma_2,\maps{j}{w_2}}{K_2}{s_2''}{E_2'\hole[\app\alpha' w_2']}
              \\&\trans{\lpropapp{\alpha'}{j'}}
              \conf{A_1\uplus\alpha'}{\Gamma_2,\maps{j}{w_2},\maps{j'}{w_2'}}{E'_2,K_2}{s_2''}{\noe}
              \\&\trans{\lopapp{j}{\alpha''}}
               \conf{A_1\uplus\alpha',\alpha''}{\Gamma_2,\maps{j}{w_2},\maps{j'}{w_2'}}{E'_2,K_2}{s_2''}{\app{w_2}{\alpha''}}
               &(\alpha''\not\in\an{w_2'})
              \\& \wtrans{\tau}
               \conf{A_1\uplus\alpha',\alpha''}{\Gamma_2,\maps{j}{w_2},\maps{j'}{w_2'}}{E'_2,K_2}{s_2''}{E''_2\hole[\app{\alpha}{v_2}]}
              \\&\trans{\lpropapp{\alpha}{i}}
               \conf{A_1\uplus\alpha',\alpha''}{\Gamma_2,\maps{j}{w_2},\maps{j'}{w_2'},\maps{i}{v_2'}}{
                 E''_2,E'_2,K_2}{s_2''}{\noe}
               =C_4'
             \end{align*}
              Using \cref{lem:lts-concr} and LTS rule $\iref{tau}$, we can derive the following transitions from $C_2$:
             \begin{align*}
               C_2 & =\conf{A_2}{\Gamma_2}{K_2}{s_2}{E_2\hole[e_2]}
               \wtrans{\tau}
               \conf{A_2}{\Gamma_2}{K_2}{s_2'}{E_2\hole[w_2]}
              \\& =
               \conf{A_2\uplus\alpha'}{\Gamma_2}{K_2}{s_2'}{E_2\hole[\alpha']}\concr{\alpha'}{w_2}
               &(\alpha'\not\in\an{w_2})
              \\& \wtrans{\tau}
               \conf{A_2\uplus\alpha'}{\Gamma_2}{K_2}{s_2''}{E_2'\hole[\app\alpha' w_2']}\concr{\alpha'}{w_2}
              \\&=
               \conf{A_1\uplus\alpha',\alpha''}{\Gamma_2}{K_2}{s_2''}{E'_2\hole[\app{w_2}{\alpha''}]}
               \concr{\alpha''}{w_2'}
               \concr{\alpha'}{w_2}
               &(\alpha''\not\in\an{w_2'})
              \\& \wtrans{\tau}
               \conf{A_1\uplus\alpha',\alpha''}{\Gamma_2}{,K_2}{s_2''}{
                 E'_2\hole[E''_2\hole[\app{\alpha}{v_2}]]}
               \concr{\alpha''}{w_2'}
               \concr{\alpha'}{w_2}
              \\&\trans{\lpropapp{\alpha}{i}}
               \conf{A_1\uplus\alpha',\alpha''}{\Gamma_2,\maps{i}{v_2}}{
                 E'_2\hole[E''_2],K_2}{s_2''}{\noe}
               \concr{\alpha''}{w_2'}
               \concr{\alpha'}{w_2}
               =C_2'
             \end{align*}
             We also have
             \[\abox{
               C_3'&=
               \conf{A_1\uplus\alpha',\alpha''}{\Gamma_1,\maps{j}{\alpha},\maps{j'}{v_1},\maps{i}{\alpha''}}{
               \hole,E_1',K_1}{s_1}{\noe}
             \\&\simil
               \conf{A_1\uplus\alpha',\alpha''}{\Gamma_2,\maps{j}{w_2},\maps{j'}{w_2'},\maps{i}{v_2'}}{
               E''_2,E'_2,K_2}{s_2''}{\noe}
               =C_4'
             }\]
             By rule \iref{utfv}
             \[\nbox{
               \conf{A_1\uplus\alpha',\alpha''}{\Gamma_1,\maps{i}{\alpha''}}{
               \hole,E_1',K_1}{s_1}{\noe}
               \concr{\alpha''}{v_1}
               \concr{\alpha'}{\alpha}
               \\\quad\utfoldv*{\utfoldv{\simil}}
               \conf{A_1\uplus\alpha',\alpha''}{\Gamma_2,\maps{i}{v_2'}}{
               E''_2,E'_2,K_2}{s_2''}{\noe}
               \concr{\alpha''}{w_2'}
               \concr{\alpha'}{w_2}
             }\]
            Because $\alpha',\alpha''\not\in\an{\Gamma_1,K_1,s_1,E_1',v_1}$
             \[\nbox{
               \conf{A_1}{\Gamma_1,\maps{i}{v_1}}{
               \hole,E_1',K_1}{s_1}{\noe}
               \\\quad\utfoldv*{\utfoldv{\simil}}
               \conf{A_1\uplus\alpha',\alpha''}{\Gamma_2,\maps{i}{v_2'}}{
               E''_2,E'_2,K_2}{s_2''}{\noe}
               \concr{\alpha''}{w_2'}
               \concr{\alpha'}{w_2}
             }\]
             By rule \iref{utfe2}
             \[\nbox{
               \conf{A_1}{\Gamma_1,\maps{i}{v_1}}{
               E_1'\hole,K_1}{s_1}{\noe}
             \\\quad\utfolde*{\utfoldv{\utfoldv{\simil}}}
               \conf{A_1\uplus\alpha',\alpha''}{\Gamma_2,\maps{i}{v_2'}}{
                 E'_2\hole[E''_2],K_2}{s_2''}{\noe}
               \concr{\alpha''}{w_2'}
               \concr{\alpha'}{w_2}
             }\]
             Therefore $C_1' \utfolde*{\utfoldv{\utfoldv{\simil}}} C_2'$ as required.
           \end{itemize}

        \item $e_1=v_1$, $E_1=E_1'\hole[\app\alpha\hole]$
            and $C'_1=\conf{A_1}{\Gamma_1,\maps{i}{v_1}}{E_1';K_1}{s_1}{\noe}$:
          In this case configuration $C_3$ can perform the transitions:
          \begin{align*}
            C_3
            &\trans{\lpropret{j}}
            \conf{A_1}{\Gamma_1,\maps{j}{v}}{E_1'\hole[\app\alpha\hole],K_1}{s_1}{\noe}
            \\&\trans{\lopret{\alpha'}}
            \conf{A_1\uplus\alpha'}{\Gamma_1,\maps{j}{v}}{K_1}{s_1}{E_1'\hole[\app\alpha\alpha']}
            &(\alpha'\not\in\an{\Gamma_1,v_1,K_1,s_1,E_1'})
            \\&\trans{\lpropapp{\alpha}{i}}
            \conf{A_1\uplus\alpha'}{\Gamma_1,\maps{j}{v},\maps{i}{\alpha'}}{E_1',K_1}{s_1}{\noe}
            =C_3'
          \end{align*}
          Because $C_3\simil C_4$, there exists $C_4'$ such that
          $C_4
            \wtrans{\lpropret{j}}
            \wtrans{\lopret{\alpha'}}
            \wtrans{\lpropapp{\alpha}{i}}
          C_4'$
          and $C_3' \simil C_4'$. By analysis of the transitions we have
          \begin{align*}
            C_4 & =\conf{A_2}{\Gamma_2}{E_2,K_2}{s_2}{e_2}
            \wtrans{\tau}
            \conf{A_2}{\Gamma_2}{E_2,K_2}{s_2'}{v_2}
            \\& \trans{\lpropret{j}}
            \conf{A_2}{\Gamma_2,\maps{j}{v_2}}{E_2,K_2}{s_2'}{\noe}
            \\& \trans{\lopret{ \alpha'}}
            \conf{A_2\uplus\alpha'}{\Gamma_2,\maps{j}{v_2}}{K_2}{s_2'}{E_2\hole[\alpha']}
            &(\alpha'\not\in\an{\Gamma_2,v_2,K_2,s_2',E_2'})
            \\& \wtrans{\tau}
            \conf{A_2\uplus\alpha'}{\Gamma_2,\maps{j}{v_2}}{K_2}{s_2''}{E_2'\hole[\app \alpha w_2]}
            &
            \\& \trans{\lpropapp{\alpha}{i}}
            \conf{A_2\uplus\alpha'}{\Gamma_2,\maps{j}{v_2},\maps{i}{w_2}}{E_2',K_2}{s_2''}{\noe}
            = C_4'
          \end{align*}
          Using \cref{lem:lts-concr} and LTS rule $\iref{tau}$, we can derive the following transitions from $C_2$:
          \begin{align*}
            C_2 & =\conf{A_2}{\Gamma_2}{K_2}{s_2}{E_2\hole[e_2]}
            \wtrans{\tau}
            \conf{A_2}{\Gamma_2}{K_2}{s_2'}{E_2[v_2]}
            \\&=
            \conf{A_2\uplus\alpha'}{\Gamma_2}{K_2}{s_2'}{E_2\hole[\alpha']}
            \concr{\alpha'}{v_2}
            \wtrans{\tau}
            \conf{A_2\uplus\alpha'}{\Gamma_2}{K_2}{s_2''}{E_2'\hole[\app \alpha w_2]}
            \concr{\alpha'}{v_2}
            &
            \\& \trans{\lpropapp{\alpha}{i}}
            \conf{A_2\uplus\alpha'}{\Gamma_2,\maps{i}{w_2}}{E_2',K_2}{s_2''}{\noe}
            \concr{\alpha'}{v_2}
            = C_2'
          \end{align*}
          We also have
          \[
            C_3' =
            \conf{A_1\uplus\alpha'}{\Gamma_1,\maps{j}{v},\maps{i}{\alpha'}}{E_1',K_1}{s_1}{\noe}
            \simil
            \conf{A_2\uplus\alpha'}{\Gamma_2,\maps{j}{v_2},\maps{i}{w_2}}{E_2',K_2}{s_2''}{\noe}
            = C_4'
          \]
          And by rule \iref{utfv}:
          \[\abox{
            C_1' &=
            \conf{A_1}{\Gamma_1,\maps{i}{v_1}}{E_1',K_1}{s_1}{\noe}
            \\&=
            \conf{A_1\uplus\alpha'}{\Gamma_1,\maps{i}{\alpha'}}{E_1',K_1}{s_1}{\noe}
            \concr{\alpha'}{v_1}
            \utfoldv*{\simil}
            \conf{A_2\uplus\alpha'}{\Gamma_2,\maps{i}{w_2}}{E_2',K_2}{s_2''}{\noe}
            \concr{\alpha'}{v_2}
            = C_2'
          }\]
 
      \end{itemize}
    \item
      \iref{propretc} and \iref{propretf}: here we show only the case for the latter rule.
      We have $\eta=\lpropret{i}$ and $E_1=\hole$
      and $C'_1=\conf{A_1}{\Gamma_1,\maps{i}{v_1}}{K_1}{s_1}{\noe}$.
      Configuration $C_3$ can perform the transitions:
      \begin{align*}
        C_3
        &\trans{\lpropret{j}}
        \conf{A_1}{\Gamma_1,\maps{j}{v_1}}{\hole,K_1}{s_1}{\noe}
        \\&\trans{\lopret{\alpha'}}
        \conf{A_1\uplus\alpha'}{\Gamma_1,\maps{j}{v_1}}{K_1}{s_1}{\hole[\alpha']}
        &(\alpha'\not\in\an{\Gamma_1,v_1,K_1,s_1})
        \\&\trans{\lpropret{i}}
        \conf{A_1\uplus\alpha'}{\Gamma_1,\maps{j}{v_1},\maps{i}{\alpha'}}{K_1}{s_1}{\noe}
        =C_3'
      \end{align*}
      Because $C_3\simil C_4$, there exists $C_4'$ such that
      $C_4
        \wtrans{\lpropret{j}}
        \wtrans{\lopret{\alpha'}}
        \wtrans{\lpropret{i}}
      C_4'$
      and $C_3' \simil C_4'$. By analysis of the transitions we have
      \begin{align*}
        C_4 &=
        \conf{A_2}{\Gamma_2}{E_2,K_1}{s_1}{e_2}
        \wtrans{\tau}
        \conf{A_2}{\Gamma_2}{E_2,K_1}{s_1'}{v_2}
        \\&\trans{\lpropret{j}}
        \conf{A_2}{\Gamma_2,\maps{j}{v_2}}{E_2,K_2}{s_2'}{\noe}
        \\&\trans{\lopret{\alpha'}}
        \conf{A_2\uplus\alpha'}{\Gamma_2,\maps{j}{v_2}}{K_2}{s_2'}{E_2\hole[\alpha']}
        &(\alpha'\not\in\an{\Gamma_2,v_2,K_2,s_2',E_2})
        \\&\wtrans{\tau}
        \conf{A_2\uplus\alpha'}{\Gamma_2,\maps{j}{v_2}}{K_2}{s_2''}{w_2}
        \\&\trans{\lpropret{i}}
        \conf{A_2\uplus\alpha'}{\Gamma_2,\maps{j}{v_2},\maps{i}{w_2}}{K_2}{s_2''}{\noe}
        =C_4'
      \end{align*}
      Using \cref{lem:lts-concr} and LTS rule $\iref{tau}$, we can derive the following transitions from $C_2$:
      \begin{align*}
        C_2 &=
        \conf{A_2}{\Gamma_2}{K_1}{s_1}{E_2\hole[e_2]}
        \wtrans{\tau}
        \conf{A_2}{\Gamma_2}{K_1}{s_1'}{E_2\hole[v_2]}
        \\&=
        \conf{A_2\uplus\alpha'}{\Gamma_2}{K_2}{s_2'}{E_2\hole[\alpha']}
        \concr{\alpha'}{v_2}
        \wtrans{\tau}
        \conf{A_2\uplus\alpha'}{\Gamma_2}{K_2}{s_2''}{w_2}
        \concr{\alpha'}{v_2}
        \\&\trans{\lpropret{i}}
        \conf{A_2\uplus\alpha'}{\Gamma_2,\maps{i}{w_2}}{K_2}{s_2''}{\noe}
        \concr{\alpha'}{v_2}
        =C_2'
      \end{align*}
      Moreover,
      \[\abox{
        C_1' &= \conf{A_1}{\Gamma_1,\maps{i}{v_1}}{K_1}{s_1}{\noe}
        \\ &= \conf{A_1\uplus\alpha'}{\Gamma_1,\maps{i}{\alpha'}}{K_1}{s_1}{\noe} \concr{\alpha'}{v_1}
        \utfoldv*{\simil}
        \conf{A_2\uplus\alpha'}{\Gamma_2,\maps{i}{w_2}}{K_2}{s_2''}{\noe}\concr{\alpha'}{v_2}
        =C_2'
      }\]

    \item
      \iref{tau}:
      We have $\eta=\tau$ and 
      and $C'_1=\conf{A_1}{\Gamma_1}{K_1}{s_1'}{e_1'}$
      and $\redconf{s_1}{E_1\hole[e_1]} \red \redconf{s_1'}{e_1'}$.
      By the reduction rule we get $E_1\hole[e_1]=E\hole[e]$ and 
      $\redconf{s_1}{e} \redbase \redconf{s_1'}{e'}$ and $e_1'=E\hole[e']$, for some $E,e,e'$.

      We proceed by cases on $E_1\hole[e_1]=E\hole[e]$:
      \begin{itemize}
        \item $e_1=E_{e_1}\hole[e]$ and $E=E_1\hole[E_{e_1}\hole]$ and $e_1'=E_1\hole[{E_{e_1}\hole[e']}]$.
          In this case $C_3$ can perform the same transition:
          \[
            C_3 = \conf{A_1}{\Gamma_1}{E_1,K_1}{s_1}{E_{e_1}\hole[e]}
            \trans{\tau} \conf{A_1}{\Gamma_1}{E_1,K_1}{s_1'}{E_{e_1}\hole[e']} = C_3'
          \]
          Because $C_3 \simil C_4$,
          \[
            C_4 = \conf{A_2}{\Gamma_2}{E_2,K_2}{s_2}{e_2}
            \wtrans{\tau} \conf{A_2}{\Gamma_2}{E_2,K_2}{s_2'}{e_2'} = C_4'
          \]
          and $C_3' \simil C_4'$. Therefore
          \[
            C_2 = \conf{A_2}{\Gamma_2}{K_2}{s_2}{E_2\hole[e_2]}
            \wtrans{\tau} \conf{A_2}{\Gamma_2}{K_2}{s_2'}{E_2\hole[e_2']} = C_2'
          \]
          and from rule \iref{utfe1}, $C_1' \utfolde*{\simil} C_2'$.

        \item $e_1=v_1$ and $E_1=E\hole[F]$
          and
          $\redconf{s_1}{F\hole[v_1]} \redbase \redconf{s_1'}{e'}$,
          where $F$ is one of the following contexts:
          $(\arithop{\vec c_{11},\hole,\vec c_{12}})$ or
          $(\app \hole v)$ or
          $(\app v \hole)$ or
          $(l \asgn \hole)$ or
          $(\cond \hole{e_{11}}{e_{12}})$.
          We proceed by cases on $F$:
          \begin{itemize}
            \item $F=\arithop{\vec c_{11},\hole,\vec c_{12}}$:
              Here it must be $v_1=c$ and $e'=c'=\arithop{\vec c_{11},c,\vec c_{12}}$. Thus $C_3$ can perform the transitions
              \begin{align*}
                C_3 &= \conf{A_1}{\Gamma_1}{E\hole[F],K_1}{s_1}{c}
                \trans{\lpropret{c}}
                \conf{A_2}{\Gamma_1}{E\hole[F],K_1}{s_1}{\noe}
                \\&\trans{\lopret{c}}
                \conf{A_2}{\Gamma_1}{K_1}{s_1}{E\hole[F{\hole[c]}]}
                \trans{\tau}
                \conf{A_2}{\Gamma_1}{K_1}{s_1}{E\hole[c']}
                =C_3' = C_1'
              \end{align*}
              Because $C_3 \simil C_4$,
              \begin{align*}
                C_4 &= \conf{A_2}{\Gamma_2}{E_2,K_2}{s_2}{e_2}
                \wtrans{\tau}
                \conf{A_2}{\Gamma_2}{E_2,K_2}{s_2'}{c}
                \trans{\lpropret{c}}
                \conf{A_2}{\Gamma_2}{E_2,K_2}{s_2'}{\noe}
                \\&\trans{\lopret{c}}
                \conf{A_2}{\Gamma_2}{K_2}{s_2'}{E_2\hole[c]}
                \wtrans{\tau}
                C_4'
              \end{align*}
              and $C_3'\simil C_4'$. Moreover we derive the transitions:
              \begin{align*}
                C_2 &= \conf{A_2}{\Gamma_2}{K_2}{s_2}{E_2\hole[e_2]}
                \wtrans{\tau}
                \conf{A_2}{\Gamma_2}{K_2}{s_2'}{E_2\hole[c]}
                \wtrans{\tau}
                C_4' = C_2'
              \end{align*}
              We also derive $C_1' \mathrel{f^0(\simil)} C_2'$ as needed.
             
            \item $F=\app \hole v$: Here it must be $v_1 = \lam{x} e_{11}$ and $e'=e_{11}\sub{x}{v}$.
              By the LTS and $C_3\simil C_3'$ we have:
             \begin{align*}
               C_3 &= \conf{A_1}{\Gamma_1}{E\hole[F],K_1}{s_1}{\lam{x} e_{11}}
                \\&\trans{\lpropret{i}}
                \conf{A_2}{\Gamma_1,\maps{i}{\lam{x} e_{11}}}{E\hole[F],K_1}{s_1}{\noe}
                \\&\trans{\lopret{\alpha}}
                \conf{A_1\uplus\alpha}{\Gamma_1,\maps{i}{\lam{x} e_{11}}}{K_1}{s_1}{E\hole[\app\alpha v]}
                &(\alpha\not\in\an{A_1,\Gamma_1,e_{11},K_1,s_1,E,v})
                \\&\trans{\lpropapp{\alpha}{j}}
                \conf{A_1\uplus\alpha}{\Gamma_1,\maps{i}{\lam{x} e_{11},\maps{j}{v}}}{E,K_1}{s_1}{\noe}
                \\&\trans{\lopapp{i}{\alpha'}}
                \conf{A_1\uplus\alpha}{\Gamma_1,\maps{i}{\lam{x} e_{11},\maps{j}{v}}}{E,K_1}{s_1}{e_{11}\sub{x}{\alpha'}}
                =C_3'
                &(\alpha'\not\in\an{A_1,\Gamma_1,e_{11},K_1,s_1,E,v})
             \end{align*}
             \begin{align*}
               C_4 &= \conf{A_2}{\Gamma_2}{E_2,K_2}{s_2}{e_2}
                \wtrans{\tau}
                \conf{A_2}{\Gamma_2}{E_2,K_2}{s_2'}{v_2}
                \\&\trans{\lpropret{i}}
                \conf{A_2}{\Gamma_2,\maps{i}{v_2}}{E_2,K_2}{s_2'}{\noe}
                \\&\trans{\lopret{\alpha}}
                \conf{A_2\uplus\alpha}{\Gamma_2,\maps{i}{v_2}}{K_2}{s_2'}{E_2\hole[\alpha]}
                &(\alpha\not\in\an{A_2,\Gamma_2,v_2,K_2,s_2',E_2})
                \\&\wtrans{\tau}
                \conf{A_2\uplus\alpha}{\Gamma_2,\maps{i}{v_2}}{K_2}{s_2''}{E_2'\hole[\app{\alpha}{v_2'}]}
                \\&\trans{\lpropapp{\alpha}{j}}
                \conf{A_2\uplus\alpha,\alpha'}{\Gamma_2,\maps{i}{v_2},\maps{j}{v_2'}}{E_2',K_2}{s_2''}{\noe}
                \\&\trans{\lopapp{i}{\alpha'}}
                \conf{A_2\uplus\alpha,\alpha'}{\Gamma_2,\maps{i}{v_2},\maps{j}{v_2'}}{E_2',K_2}{s_2''}{e_2'}
                &(\alpha'\not\in\an{A_2,\Gamma_2,v_2,v_2',E_2',K_2,s_2''})
               \\&=
                \begin{cases}
                  \conf{A_2\uplus\alpha,\alpha'}{\Gamma_2,\maps{i}{v_2},\maps{j}{v_2'}}{
                          E_2',K_2}{s_2''}{\app{\alpha''}{\alpha'}}
                                          &\rlap{\text{if } $v_2=\alpha''$}\\
                  \conf{A_2\uplus\alpha,\alpha'}{\Gamma_2,\maps{i}{v_2},\maps{j}{v_2'}}{
                          E_2',K_2}{s_2''}{\app{e_2'}\sub{x}{\alpha'}}
                                          &\rlap{\text{if } $v_2=\lam{x} e_2'$}
                \end{cases}
               \\&\wtrans{\tau}
               \conf{A_2\uplus\alpha,\alpha'}{\Gamma_2,\maps{i}{v_2},\maps{j}{v_2'}}{
                 E_2',K_2}{s_2'''}{e_2''} = C_4'
                & C_3' \simil C_4'
             \end{align*}
             \begin{align*}
               C_2 &= \conf{A_2}{\Gamma_2}{K_2}{s_2}{E_2[e_2]}
                \wtrans{\tau}
                \conf{A_2}{\Gamma_2}{E_2,K_2}{s_2'}{E_2[v_2]}
                \\&=
                \conf{A_2\uplus\alpha}{\Gamma_2}{K_2}{s_2'}{E_2\hole[\alpha]}
                \concr{\alpha}{v_2}
                \wtrans{\tau}
                \conf{A_2\uplus\alpha}{\Gamma_2}{K_2}{s_2''}{E_2'\hole[\app\alpha v_2']}
                \concr{\alpha}{v_2}
                \\&=
                \conf{A_2\uplus\alpha,\alpha'}{\Gamma_2}{K_2}{s_2''}{E_2'\hole[\app\alpha \alpha']}
                \concr{\alpha'}{v_2'}
                \concr{\alpha}{v_2}
                \\&
                \begin{cases}
                  =
                  \conf{A_2\uplus\alpha,\alpha'}{\Gamma_2}{K_2}{s_2''}{E_2'\hole[\app{\alpha''} \alpha']}
                  \concr{\alpha'}{v_2'}
                  \concr{\alpha}{v_2}
                                          &\rlap{\text{if } $v_2=\alpha''$}\\
                  \trans{\tau}
                  \conf{A_2\uplus\alpha,\alpha'}{\Gamma_2}{K_2}{s_2''}{E_2'\hole[e_2'\sub{x}{\alpha'}]}
                                          &\rlap{\text{if } $v_2=\lam{x} e_2'$}
                \end{cases}
               \\&\wtrans{\tau}
               \conf{A_2\uplus\alpha,\alpha'}{\Gamma_2}{K_2}{s_2'''}{E_2'\hole[e_2'']}
                  \concr{\alpha'}{v_2'}
                  \concr{\alpha}{v_2}
               = C_2'
             \end{align*}
             Moreover: $C'_1 \utfolde*{\utfoldv{\utfoldv{\simil}}} C'_2$, as required.

            \item $F=\app v \hole$: Here it must be $v = \lam{x} e_{11}$ and $e'=e_{11}\sub{x}{v_1}$.
              By the LTS and $C_3\simil C_3'$ we have:
             \begin{align*}
               C_3 &= \conf{A_1}{\Gamma_1}{E\hole[F],K_1}{s_1}{v_1}
                \\&\trans{\lpropret{i}}
                \conf{A_2}{\Gamma_1,\maps{i}{v_1}}{E\hole[F],K_1}{s_1}{\noe}
                \\&\trans{\lopret{\alpha}}
                \conf{A_1\uplus\alpha}{\Gamma_1,\maps{i}{v_1}}{K_1}{s_1}{E\hole[\app {(\lam{x}{e_{11}})} \alpha]}
                \\&\trans{\tau}
                \conf{A_1\uplus\alpha}{\Gamma_1,\maps{i}{v_1}}{K_1}{s_1}{E\hole[e_{11}\sub{x}{\alpha}]}
                =C_3'
                &(\alpha\not\in\an{A_1,\Gamma_1,v_{1},K_1,s_1,E,v_1})
             \end{align*}
             \begin{align*}
               C_4 &= \conf{A_2}{\Gamma_2}{E_2,K_2}{s_2}{e_2}
                \wtrans{\tau}
                \conf{A_2}{\Gamma_2}{E_2,K_2}{s_2'}{v_2}
                \\&\trans{\lpropret{i}}
                \conf{A_2}{\Gamma_2,\maps{i}{v_2}}{E_2,K_2}{s_2'}{\noe}
                \\&\trans{\lopret{\alpha}}
                \conf{A_2\uplus\alpha}{\Gamma_2,\maps{i}{v_2}}{K_2}{s_2'}{E_2\hole[\alpha]}
                &(\alpha\not\in\an{A_2,\Gamma_2,v_2,K_2,s_2',E_2})
                \\&\wtrans{\tau}
                \conf{A_2\uplus\alpha}{\Gamma_2,\maps{i}{v_2}}{K_2}{s_2''}{e_2'}
                =C_4'
                & C_3' \simil C_4'
             \end{align*}
             \begin{align*}
               C_2 &= \conf{A_2}{\Gamma_2}{K_2}{s_2}{E_2[e_2]}
                \wtrans{\tau}
                \conf{A_2}{\Gamma_2}{E_2,K_2}{s_2'}{E_2[v_2]}
                \\&=
                \conf{A_2\uplus\alpha}{\Gamma_2}{K_2}{s_2'}{E_2\hole[\alpha]}
                \concr{\alpha}{v_2}
                \wtrans{\tau}
                \conf{A_2\uplus\alpha}{\Gamma_2}{K_2}{s_2''}{e_2'}
                \concr{\alpha}{v_2}
                =C_2'
             \end{align*}
             Moreover, $C_1' \utfoldv*{\simil} C_2'$.
  
           \item $F=l \asgn \hole$: Here $e'=\true$ and $s_1'=s_1\stupd{l}{v_1}$.
              By the LTS and $C_3\simil C_3'$ we have:
             \begin{align*}
               C_3 &= \conf{A_1}{\Gamma_1}{E\hole[F],K_1}{s_1}{v_1}
                \\&\trans{\lpropret{i}}
                \conf{A_2}{\Gamma_1,\maps{i}{v_1}}{E\hole[F],K_1}{s_1}{\noe}
                \\&\trans{\lopret{\alpha}}
                \conf{A_1\uplus\alpha}{\Gamma_1,\maps{i}{v_1}}{K_1}{s_1}{E\hole[l \asgn \alpha]}
                \\&\trans{\tau}
                \conf{A_1\uplus\alpha}{\Gamma_1,\maps{i}{v_1}}{K_1}{s_1\stupd{l}{\alpha}}{E\hole[\true]}
                =C_3'
                &(\alpha\not\in\an{A_1,\Gamma_1,v_{1},K_1,s_1,E,v_1})
             \end{align*}
             \begin{align*}
               C_4 &= \conf{A_2}{\Gamma_2}{E_2,K_2}{s_2}{e_2}
                \wtrans{\tau}
                \conf{A_2}{\Gamma_2}{E_2,K_2}{s_2'}{v_2}
                \\&\trans{\lpropret{i}}
                \conf{A_2}{\Gamma_2,\maps{i}{v_2}}{E_2,K_2}{s_2'}{\noe}
                \\&\trans{\lopret{\alpha}}
                \conf{A_2\uplus\alpha}{\Gamma_2,\maps{i}{v_2}}{K_2}{s_2'}{E_2\hole[\alpha]}
                &(\alpha\not\in\an{A_2,\Gamma_2,v_2,K_2,s_2',E_2})
                \\&\wtrans{\tau}
                \conf{A_2\uplus\alpha}{\Gamma_2,\maps{i}{v_2}}{K_2}{s_2''}{e_2'}
                =C_4'
                & C_3' \simil C_4'
             \end{align*}
             \begin{align*}
               C_2 &= \conf{A_2}{\Gamma_2}{K_2}{s_2}{E_2[e_2]}
                \wtrans{\tau}
                \conf{A_2}{\Gamma_2}{E_2,K_2}{s_2'}{E_2[v_2]}
                \\&=
                \conf{A_2\uplus\alpha}{\Gamma_2}{K_2}{s_2'}{E_2\hole[\alpha]}
                \concr{\alpha}{v_2}
                \wtrans{\tau}
                \conf{A_2\uplus\alpha}{\Gamma_2}{K_2}{s_2''}{e_2'}
                \concr{\alpha}{v_2}
                =C_2'
             \end{align*}
             Moreover, $C_1' \utfoldv*{\simil} C_2'$.

            \item $F=\cond{\hole}{e_{11}}{e_{12}}$:
              Here $v_1=\true$ or $v_1=\false$, and $e'=e_{11}$ or $e'=e_{12}$, respectively.
              In both sub-cases, the proof proceeds as in the case where $F=\arithop{\vec c_{11},\hole,\vec c_{12}}$.

          \end{itemize}

      \end{itemize}

  \end{itemize}

{\flushleft
  \iref{utfe2}: In this case we have:}
      $C_1 = \conf{A_1}{\Gamma_1}{K_{11},E'_1\hole[E_1],K_{12}}{s_1}{\hat e_1}$ and
      $C_2 = \conf{A_2}{\Gamma_2}{K_{21},E'_2\hole[E_2],K_{22}}{s_2}{\hat e_2}$ and
      $C_3 = \conf{A_1}{\Gamma_1}{K_{11},E_1,E'_1,K_{12}}{s_1}{\hat e_1}$ and
      $C_4 = \conf{A_2}{\Gamma_2}{K_{21},E_2,E'_2,K_{22}}{s_2}{\hat e_2}$ and
      $\sizeof{K_{11}}=\sizeof{K_{21}}$ and $C_3 \simil C_4$.
      When $\sizeof{K_{11}}>0$, because $C_1$ ($C_3$) and $C_2$ (resp.~ $C_4$) have the same expressions, 
      the proof requires a simple simulation diagram chasing, with the resulting configurations related in
      $\utfolde{\simil}$ via rule \iref{utfe2}.
      Similarly when $\eta \not\in\{\lopret{c},\lopret{\alpha}\where \text{any } c, \alpha$.
      When $\eta \in\{\lopret{c},\lopret{\alpha}\where \text{any } c, \alpha$, the simulation diagram is
      similar, but completed by relating the resulting configurations in
      $\utfolde{\simil}$ via rule \iref{utfe1}.
      Note that a $\lterm$-transition is not possible from $C_1$.

{\flushleft
  \iref{utfd1},
  \iref{utfd2}:}
  The proof in these cases proceeds as the corresponding cases above, with the simplification that the right-hand side configurations ($C_2$ and $C_4$) perform the required transitions via the LTS rule \iref{dummy}.
\end{proof}

\begin{lemma}\label{lem:utfoldv-sound} 
  Let $f= \utfolde{}\sqcup\utfoldv{}$
  and $g=\utperm{}$; then
  $\utfoldv{}\comp{\constf{\gfp\WP}}\progress(f\sqcup g)^\omega\comp{\constf{\gfp\WP}}$.
\end{lemma}
\begin{proof}
  We need to show
  $\utfoldv{}\comp{\constf{\gfp\WP}}\comp\WP(\bisim{R}) \sqsubseteq 
  {\WP}\comp (f\sqcup g)^\omega\comp{\constf{\gfp\WP}}(\bisim{R})$.
  Because $\constf{\gfp\WP}$ is the constant function mapping its argument to $(\simil)$, we need to show
  $\utfoldv{}(\simil) \sqsubseteq 
  {\WP}\comp (f\sqcup g)^\omega(\simil)$.

  Let $C_1\utfoldv*{\simil}C_2$, 
  and $C_1\trans{\eta}C_1'$.
  If $C_1\trans{\eta}C_1'$ is produced by rule \iref{dummy} then the proof is trivial as $C_2$ can perform the same transition and $C_1'=\botconf \utId*{\simil} \botconf=C_2'$. We thus consider only non-\iref{dummy} transitions from $C_1$.
  By case analysis, the derivation
  $C_1\utfoldv*{\simil}C_2$ can be produced by the \iref{utfv} or the \iref{utfd3} rules. We show the former; the proof of the latter is similar but simpler.

  {\flushleft
  \iref{utfv}:}
  Here we have
      $C_1 = C_3\kappa_1$ and
      $C_2 = C_4\kappa_2$ and
      $C_3=\conf{A_3\uplus\alpha}{\Gamma_3,\maps{i}{v_3}}{K_3}{s_3}{\hat e_3}$ and
      $C_4=\conf{A_4\uplus\alpha}{\Gamma_4,\maps{i}{v_4}}{K_4}{s_4}{\hat e_4}$ and
      $C_3 \simil C_4$ and
      $\alpha\not\in\an{v_3,v_4}$ and
      $\kappa_1=\concr{\alpha}{v_3}$ and
      $\kappa_2=\concr{\alpha}{v_4}$.
   We proceed by cases on the transition from $C_1$. Cases \iref{propretc}, \iref{propretf}, \iref{opretc}, \iref{opretf}, and \iref{term} are straightforward and are proved using \cref{lem:lts-concr}.

   \begin{itemize}
     \item \iref{propappc}, \iref{propappf}: Both cases are similar; here we show the latter.
       We have $\eta=\lpropapp{\alpha'}{j}$ and $\alpha\in A_3$ ($\alpha\not=\alpha'$) and
       $\hat e_3\kappa_1=E_1\hole[\app{\alpha'} u_1]$ and
       \[
          C_1=\conf{A_3}{\Gamma_3\kappa_1}{K_3\kappa_1}{s_3\kappa_1}{E_1\hole[\app{\alpha'} u_1]}
          \trans{\lpropapp{\alpha'}{j}}
          \conf{A_3}{\Gamma_3\kappa_1,\maps{j}{u_1}}{E_1,K_3\kappa_1}{s_3\kappa_1}{\noe}
          =C_1'
        \]
        We choose $j'\not\in\dom{\Gamma_3}\uplus i$
        and create the permutation $\pi=(j\bij j')$. We derive
       \[
          C_1=\conf{A_3}{\Gamma_3\kappa_1}{K_3\kappa_1}{s_3\kappa_1}{E_1\hole[\app{\alpha'} u_1]}
          \trans{\lpropapp{\alpha'}{j'}}
          \conf{A_3}{\Gamma_3\kappa_1,\maps{j'}{u_1}}{E_1,K_3\kappa_1}{s_3\kappa_1}{\noe}
          =C_1'\pi
        \]
        We consider cases on 
       $\hat e_3\kappa_1=E_1\hole[\app{\alpha'} u_1]$:
       \begin{itemize}
         \item
           $\hat e_3=E_3\hole[\app{\alpha'} u_3]$ and
           $E_1=E_3\kappa_1$ and $u_1=u_3\kappa_1$.
           In this case we have
           \begin{align*}
             C_3=\conf{A_3\uplus\alpha}{\Gamma_3,\maps{i}{v_3}}{K_3}{s_3}{E_3\hole[\app{\alpha'} u_3]}
             &\trans{\lpropapp{\alpha'}{j'}}
             \conf{A_3\uplus\alpha}{\Gamma_3,\maps{i}{v_3},\maps{j'}{u_3}}{E_3,K_3}{s_3}{\noe}
             =C_3'
           \end{align*}
           By $C_3 \simil C_4$ we have 
           $C_4 \wtrans{\lpropapp{\alpha'}{j'}} C_4$ 
           and $C_3' \simil C_4'$.

           We have two cases. The first is when $C_4'=\botconf$. This case is straightforward using rule \iref{utfd3}
           and relating the resulting configurations in $\utperm{\utfoldv{\simil}}$.
           The other case is as follows.
           \begin{align*}
             C_4=\conf{A_4\uplus\alpha}{\Gamma_4,\maps{i}{v_4}}{K_4}{s_4}{\hat e_4}
             &\wtrans{\tau}
             \conf{A_4\uplus\alpha}{\Gamma_4,\maps{i}{v_4}}{K_4}{s_4'}{E_4\hole[\app{\alpha'} u_4]}
             \\
             &\trans{\lpropapp{\alpha'}{j'}}
             \conf{A_4\uplus\alpha}{\Gamma_4,\maps{i}{v_4},\maps{j'}{u_4}}{E_4,K_4}{s_4'}{\noe}
             =C_4'
             &\text{and } C_3' \simil C_4'
             \\
             C_2=\conf{A_4\uplus\alpha}{\Gamma_4,\maps{i}{v_4}}{K_4}{s_4}{\hat e_4}\kappa_2
             &\wtrans{\tau}
             \conf{A_4\uplus\alpha}{\Gamma_4,\maps{i}{v_4}}{K_4}{s_4'}{E_4\hole[\app{\alpha'} u_4]}\kappa_2
             &(\text{\ref{lem:lts-concr}})
             \\
             &\trans{\lpropapp{\alpha'}{j'}}
             \conf{A_4\uplus\alpha}{\Gamma_4,\maps{i}{v_4},\maps{j'}{u_4}}{E_4,K_4}{s_4'}{\noe}\kappa_2
             \\
              C_2
             &\wtrans{\lpropapp{\alpha'}{j}}
             C_4'\kappa_2\pi = C_2'
             &(\ref{lem:lts-perm})
            \end{align*}
            Moreover, $C_1' \utperm*{\utfoldv{\simil}} C_2'$, as required.
 
         \item $\hat e_3=E_3\hole[\app{\alpha} u_3]$ and
           $E_1=E_3\kappa_1$ and $u_1=u_3\kappa_1$ and $\kappa_1=\concr{\alpha}{\alpha'}$ and $v_3=\alpha$.
           In this case we have
           \begin{align*}
             C_3=\conf{A_3\uplus\alpha}{\Gamma_3,\maps{i}{v_3}}{K_3}{s_3}{E_3\hole[\app{\alpha} u_3]}
             &\trans{\lpropapp{\alpha}{j'}}
             \conf{A_3\uplus\alpha}{\Gamma_3,\maps{i}{v_3},\maps{j'}{u_3}}{E_3,K_3}{s_3}{\noe}
             =C_3'
           \end{align*}
           By $C_3 \simil C_4$ we have 
           $C_4 \wtrans{\lpropapp{\alpha}{j'}} C_4$ 
           and $C_3' \simil C_4'$.

           We have two cases. The first is when $C_4'=\botconf$. This case is straightforward using rule \iref{utfd3}
           and relating the resulting configurations in $\utperm{\utfoldv{\simil}}$.
           The other case is as follows.
           \begin{align*}
             C_4=\conf{A_4\uplus\alpha}{\Gamma_4,\maps{i}{v_4}}{K_4}{s_4}{\hat e_4}
             &\wtrans{\tau}
             \conf{A_4\uplus\alpha}{\Gamma_4,\maps{i}{v_4}}{K_4}{s_4'}{E_4\hole[\app\alpha u_4]}
             \\
             &\trans{\lpropapp{\alpha}{j'}}
             \conf{A_4\uplus\alpha}{\Gamma_4,\maps{i}{v_4},\maps{j'}{u_4}}{E_4,K_4}{s_4'}{\noe}
             =C_4'
             &\text{and } C_3' \simil C_4'
             \\
             C_2=\conf{A_4\uplus\alpha}{\Gamma_4,\maps{i}{v_4}}{K_4}{s_4}{\hat e_4}\kappa_2
             &\wtrans{\tau}
             \conf{A_4\uplus\alpha}{\Gamma_4,\maps{i}{v_4}}{K_4}{s_4'}{E_4\hole[\app\alpha u_4]}\kappa_2
             \\
             &\trans{\lpropapp{\alpha'}{j'}}
             \conf{A_4\uplus\alpha}{\Gamma_4,\maps{i}{v_4},\maps{j'}{u_4}}{E_4,K_4}{s_4'}{\noe}\kappa_2
             &(\text{\ref{lem:lts-concr}})
             \\
              C_2
             &\wtrans{\lpropapp{\alpha'}{j}}
             C_4'\kappa_2\pi = C_2'
             &(\ref{lem:lts-perm})
            \end{align*}
            Moreover, $C_1' \utperm*{\utfoldv{\simil}} C_2'$, as required.
        \end{itemize}
     \item \iref{opappc}, \iref{opappf}: Both cases are similar; here we show the latter.
       We have $\eta=\lopapp{i'}{\alpha'}$ and $\alpha' \not\in A_3$ and $\hat e_3=\noe$ and
       $\Gamma_3(i')=u_3$ and $\app{u_3\kappa_1} \alpha' \funred e_1'$ and
       \[
         C_1= \conf{A_3\uplus\alpha}{\Gamma_3,\maps{i}{v_3}}{K_3}{s_3}{\noe}\kappa_1
          \trans{\lopapp{i'}{\alpha'}}
          \conf{A_3\uplus\alpha'}{\Gamma_3\kappa_1}{K_3\kappa_1}{s_3\kappa_1}{e_1'}
          =C_1'
        \]
        We choose $\alpha''\not\in A\uplus \alpha$
        and create the permutation $\pi=(\alpha'\bij \alpha'')$.
        We derive:
       \[
         C_1= \conf{A_3\uplus\alpha}{\Gamma_3,\maps{i}{v_3}}{K_3}{s_3}{\noe}\kappa_1
          \trans{\lopapp{i'}{\alpha''}}
          \conf{A_3\uplus\alpha''}{\Gamma_3\kappa_1}{K_3\kappa_1}{s_3\kappa_1}{e_1''}
          =C_1'\pi
        \]
         with $\app{u_3\kappa_1} \alpha'' \funred e_1''$.
         We consider cases on $u_3$ and $v_3$:
         \begin{itemize}
           \item $u_3 = \alpha$ and
             $v_3 = \lam x e_3''$: we have $e_1''=e_3''\sub{x}{\alpha''} = e_3''\sub{x}{\alpha''}\kappa_1$ 
                 ($\alpha\not\in\an{v_3}\cup\alpha''$).
             \begin{align*}
               C_1= \conf{A_3\uplus\alpha}{\Gamma_3,\maps{i}{v_3}}{K_3}{s_3}{\noe}\kappa_1
               &\trans{\lopapp{i'}{\alpha''}}
                \conf{A_3\uplus\alpha,\alpha''}{\Gamma_3,\maps{i}{v_3}}{K_3}{s_3}{e_3''\sub{x}{\alpha''}}\kappa_1
               =C_1'\pi
               \\
               C_3=\conf{A_3\uplus\alpha}{\Gamma_3,\maps{i}{v_3}}{K_3}{s_3}{\noe}
               &\trans{\lopapp{i'}{\alpha''}}
               \conf{A_3\uplus\alpha,\alpha''}{\Gamma_3,\maps{i}{v_3}}{K_3}{s_3}{\app{\alpha} \alpha''}
               \\
               &\trans{\lpropapp{\alpha}{j}}
               \conf{A_3\uplus\alpha,\alpha''}{\Gamma_3,\maps{i}{v_3},\maps{j}{\alpha''}}{\hole,K_3}{s_3}{\noe}
               \\
               &\trans{\lopapp{i}{\alpha'''}}
               \conf{A_3\uplus\alpha,\alpha'',\alpha'''}{\Gamma_3,\maps{i}{v_3},\maps{j}{\alpha''}}{\hole,K_3}{s_3}
               {e_3''\sub{x}{\alpha'''}}
               =C_3'
            \end{align*}
                 By $C_3 \simil C_4$ we have two cases. The first is when
                 $C_4$ weakly matches these transitions and becomes
                 $C_4'=\botconf$, due to an application of the \iref{dummy} rule.
                 This case is proved using rule \iref{utfd3} and relating the resulting
                 configurations in $\utperm{\utfolde{\utfoldv{\utfoldv{\simil}}}}$.
                 The other case is as follows.

                 By \cref{lem:simil-properties}, $\hat e_4=\noe$ and
                 $\Gamma_4(i')=u_4$ and 
             we proceed by cases on $u_4$ and $v_4$:
           \begin{itemize}
             \item $u_4 = \alpha_4\not=\alpha$:
               This is not possible because then $C_4$ would not be able to match the second transition from $C_3$.

             \item $u_4 = \alpha$ and
             $v_4 = \lam x e_4''$: we have:
            \begin{align*}
               C_4=\conf{A_4\uplus\alpha}{\Gamma_4,\maps{i}{v_4}}{K_4}{s_4}{\noe}
               &\trans{\lopapp{i'}{\alpha''}}
               \conf{A_4\uplus\alpha}{\Gamma_4,\maps{i}{v_4}}{K_4}{s_4}{\app{\alpha}{\alpha''}}
               \\
               &\trans{\lpropapp{\alpha}{j}}
               \conf{A_4\uplus\alpha,\alpha''}{\Gamma_4,\maps{i}{v_4},\maps{j}{\alpha''}}{\hole,K_4}{s_4}{\noe}
               \\
               &\trans{\lopapp{i}{\alpha'''}}
               \conf{A_4\uplus\alpha,\alpha'',\alpha'''}{\Gamma_4,\maps{i}{v_4},\maps{j}{\alpha''}}
               {\hole,K_4}{s_4}{e_4''\sub{x}{\alpha'''}}
               \\
               &\wtrans{\tau}
               \conf{A_4\uplus\alpha,\alpha'',\alpha'''}{\Gamma_4,\maps{i}{v_4},\maps{j}{\alpha''}}{\hole,K_4}{s_4'}{e_4''}
               =C_4'
            \end{align*}
           and $C_3' \simil C_4'$.
            \begin{align*}
               C_2=\conf{A_4\uplus\alpha}{\Gamma_4,\maps{i}{v_4}}{K_4}{s_4}{\noe}\kappa_2
               &\trans{\lopapp{i'}{\alpha''}}
               \conf{A_4\uplus\alpha,\alpha'',\alpha'''}{\Gamma_4,\maps{i}{v_4},\maps{j}{\alpha''}}
               {K_4}{s_4}{e_4''\sub{x}{\alpha'''}}\kappa_2\kappa_4
               \\&\qquad\qquad\qquad\qquad\qquad\qquad\qquad\qquad\qquad\qquad(\kappa_4=\concr{\alpha'''}{\alpha''})
               \\
               &\wtrans{\tau}
               \conf{A_4\uplus\alpha,\alpha'',\alpha'''}{\Gamma_4,\maps{i}{v_4},\maps{j}{\alpha''}}{K_4}{s_4'}{e_4''}
               \kappa_2\kappa_4
               =C_2'\pi
               \\
               C_2&\trans{\lopapp{i'}{\alpha'}} C_2'
            \end{align*}
               Moreover we derive $C_1' \utperm*{\utfolde{\utfoldv{\utfoldv{\simil}}}} C_2'$.
  
             \item $u_4 = \alpha$ and
             $v_4 = \alpha_4 \in A_4$: The proof is as before with transitions:
            \begin{align*}
               C_4=\conf{A_4\uplus\alpha}{\Gamma_4,\maps{i}{v_4}}{K_4}{s_4}{\noe}
               &\trans{\lopapp{i'}{\alpha''}}
               \trans{\lpropapp{\alpha}{j}}
               \\
               &\trans{\lopapp{i}{\alpha'''}}
               \conf{A_4\uplus\alpha,\alpha'',\alpha'''}{\Gamma_4,\maps{i}{v_4},\maps{j}{\alpha''}}
               {\hole,K_4}{s_4}{\app {\alpha_4} \alpha'''}
               =C_4'
               \\
               C_2=\conf{A_4\uplus\alpha}{\Gamma_4,\maps{i}{v_4}}{K_4}{s_4}{\noe}\kappa_2
               &\trans{\lopapp{i'}{\alpha''}}
               \conf{A_4\uplus\alpha,\alpha'',\alpha'''}{\Gamma_4,\maps{i}{v_4},\maps{j}{\alpha''}}
               {K_4}{s_4}{\app{\alpha_4}\alpha'''}\kappa_2\kappa_4
               =C_2'\pi
               \\&\qquad\qquad\qquad\qquad\qquad\qquad\qquad\qquad\qquad\qquad(\kappa_4=\concr{\alpha'''}{\alpha''})
               \\
               C_2&\trans{\lopapp{i'}{\alpha'}} C_2'
            \end{align*}
           and $C_3' \simil C_4'$.
               Moreover we derive $C_1' \utperm*{\utfolde{\utfoldv{\utfoldv{\simil}}}} C_2'$.
 
             \item $u_4 = \lam x e_4$:
             In this case we derive:
            \begin{align*}
               C_4=\conf{A_4\uplus\alpha}{\Gamma_4,\maps{i}{v_4}}{K_4}{s_4}{\noe}
               &\trans{\lopapp{i'}{\alpha''}}
               \conf{A_4\uplus\alpha}{\Gamma_4,\maps{i}{v_4}}{K_4}{s_4}{e_4\sub{x}{\alpha''}}
               \\
               &\wtrans{\tau}
               \conf{A_4\uplus\alpha}{\Gamma_4,\maps{i}{v_4}}{K_4}{s_4'}{E_4\hole[\app{\alpha}{w_4}]}
               \\
               &\trans{\lpropapp{\alpha}{j}}
               \conf{A_4\uplus\alpha,\alpha''}{\Gamma_4,\maps{i}{v_4},\maps{j}{w_4}}{E_4,K_4}{s_4}{\noe}
               \\
               &\trans{\lopapp{i}{\alpha'''}}
               \conf{A_4\uplus\alpha,\alpha'',\alpha'''}{\Gamma_4,\maps{i}{v_4},\maps{j}{w_4}}
               {E_4,K_4}{s_4}{e_4'}
               &(\app{v_4}{\alpha'''}\funred e_4')
               \\
               &\wtrans{\tau}
               \conf{A_4\uplus\alpha,\alpha'',\alpha'''}{\Gamma_4,\maps{i}{v_4},\maps{j}{w_4}}{E_4,K_4}{s_4'}{e_4''}
               =C_4'
            \end{align*}
           and $C_3' \simil C_4'$.
            \begin{align*}
               C_2=\conf{A_4\uplus\alpha}{\Gamma_4,\maps{i}{v_4}}{K_4}{s_4}{\noe}\kappa_2
               &\trans{\lopapp{i'}{\alpha''}}
               \conf{A_4\uplus\alpha}{\Gamma_4,\maps{i}{v_4}}{K_4}{s_4}{e_4\sub{x}{\alpha''}}\kappa_2
               \\
               &\wtrans{\tau}
               \conf{A_4\uplus\alpha}{\Gamma_4,\maps{i}{v_4}}{K_4}{s_4'}{E_4\hole[\app{\alpha}{w_4}]}
               \kappa_2
               &(\ref{lem:lts-concr})
               \\
               &\mathrel{({\trans{\tau}}\cup {=})}
               \conf{A_4\uplus\alpha,\alpha'',\alpha'''}{\Gamma_4,\maps{i}{v_4},\maps{j}{w_4}}
               {K_4}{s_4}{E_4\hole[e_4']}
               \kappa_2\kappa_4
               \\&\qquad\qquad\qquad(\kappa_4=\concr{\alpha'''}{w_4},
               (\app{v_4}{\alpha'''})\kappa_2\kappa_4\funred e_4'\kappa_2\kappa_4)
               \\
               &\wtrans{\tau}
               \conf{A_4\uplus\alpha,\alpha'',\alpha'''}{\Gamma_4,\maps{i}{v_4},\maps{j}{w_4}}{K_4}{s_4'}
               {E_4\hole[e_4'']}
               \kappa_2\kappa_4
               =C_2'\pi
               \\
               C_2
               &\wtrans{\lopapp{i'}{\alpha''}}
               C_2'
            \end{align*}
               Moreover we derive $C_1' \utperm*{\utfolde{\utfoldv{\utfoldv{\simil}}}} C_2'$.
           \end{itemize}

           \item $u_3 = \alpha$ and
             $v_3 = \alpha_3\in A_3$: 
             In this case we have $e_1''=\app{\alpha_3}\alpha''=(\app{\alpha_3}\alpha'')\kappa_1$
             ($\alpha_3\not=\alpha\not=\alpha'')$, and derive the transitions:
             \begin{align*}
               C_1= \conf{A_3\uplus\alpha}{\Gamma_3,\maps{i}{v_3}}{K_3}{s_3}{\noe}\kappa_1
               &\trans{\lopapp{i'}{\alpha''}}
                \conf{A_3\uplus\alpha,\alpha''}{\Gamma_3,\maps{i}{v_3}}{K_3}{s_3}{\app\alpha\alpha''}\kappa_1
               =C_1'\pi
               \\
               C_3= \conf{A_3\uplus\alpha}{\Gamma_3,\maps{i}{v_3}}{K_3}{s_3}{\noe}
               &\trans{\lopapp{i'}{\alpha''}}
                \conf{A_3\uplus\alpha,\alpha''}{\Gamma_3,\maps{i}{v_3}}{K_3}{s_3}{\app\alpha\alpha''}
               =C_3'
              \end{align*}
                 By $C_3 \simil C_4$ we have two cases. The first is when
                 $C_4$ weakly matches these transitions and becomes
                 $C_4'=\botconf$. This case is proved using rule \iref{utfd3} and relating the resulting
                 configurations in $\utperm{\utfoldv{\simil}}$.
                 The other case is as follows:
                 By \cref{lem:simil-properties}, $\hat e_4=\noe$ and
                 $\Gamma_4(i')=u_4$ and 
             \begin{align*}
               C_4= \conf{A_4\uplus\alpha}{\Gamma_4,\maps{i}{v_4}}{K_4}{s_4}{\noe}
               &\trans{\lopapp{i'}{\alpha''}}
                \conf{A_4\uplus\alpha,\alpha''}{\Gamma_4,\maps{i}{v_4}}{K_4}{s_4}{e_4}
                \wtrans{\tau}C_4'
               &((\app{u_4}{\alpha''})\funred e_4)
               \\
               C_2=C_4\kappa_2
               &\wtrans{\lopapp{i'}{\alpha''}}
               C_4'\kappa_2=C_2\pi
               &(\text{as above})
                \\
               C_2=C_4\kappa_2
               &\wtrans{\lopapp{i'}{\alpha'}}
               C_4'\kappa_2\pi=C_2
              \end{align*}
              Moreover $C_3' \simil C_4'$ and thus
              $C_1' \utperm*{\utfoldv{\simil}} C_2'$.
 
           \item $v_1 = \alpha_3\in A_3$: Here $\alpha_3\not=\alpha$ and the proof proceeds as in the previous case.
             The resulting configurations are again related in $\utperm{\utfoldv{\simil}}$.
           \item $u_3 = \lam x e_3$: similarly.
\end{itemize}

     \item \iref{tau}:
       Here we have $\eta=\tau$ and
       \[
         C_1 = \conf{A_3\uplus\alpha}{\Gamma_3,\maps{i}{v_3}}{K_3}{s_3}{\hat e_3}\kappa_1
         \trans{\tau}
         \conf{A_3}{\Gamma_3\kappa_1}{K_3\kappa_1}{s_3'\kappa_1}{e_1'}
         = C_1'
       \]
       We distinguish the case when $\hat e_3\kappa_1 \not= E_1\hole[\app{\lam x e_1'} u_1]$, for any $E_1$, $\lam{x}e_1'$, $u_1'$. This case is straightforward and follows from \cref{lem:lts-concr}.
       In the remaining case 
       we proceed by case analysis of the equality
       $\hat e_3\kappa_1 = E_1\hole[\app{\lam x e_1'} u_1]$:
       \begin{itemize}
         \item $\hat e_3 = E_3\hole[\app{\lam x e_3'} u_3]$ and $E_1=E_3\kappa_1$ and
           $\lam x e_1' = \lam x e_3'\kappa_1$ and $u_1 = u_3\kappa_1$ and $s_3=s_3'$.
           This case follows again from \cref{lem:lts-concr}.
         \item $\hat e_3 = E_3\hole[\app\alpha u_3]$ and $E_1=E_3\kappa_1$ and
           $v_3 = \lam x e_1'$ and $u_1 = u_3\kappa_1$ and $s_3=s_3'$.
           Here we have
           \begin{align*}
             C_3=
             \conf{A_3\uplus\alpha}{\Gamma_3,\maps{i}{v_3}}{K_3}{s_3}{E_3\hole[\app\alpha u_3]}
             &\trans{\lpropapp{\alpha}{j}}
             \conf{A_3\uplus\alpha}{\Gamma_3,\maps{i}{v_3},\maps{j}{u_3}}{E_3,K_3}{s_3}{\noe}
             \\
             &\trans{\lopapp{i}{\alpha'}}
             \conf{A_3\uplus\alpha,\alpha'}{\Gamma_3,\maps{i}{v_3},\maps{j}{u_3}}{E_3,K_3}{s_3}{e_1'\sub{x}{\alpha'}}
             =C_3'
           \end{align*}
                 By $C_3 \simil C_4$ we have two cases. The first is when
                 $C_4$ weakly matches these transitions and becomes
                 $C_4'=\botconf$. This case is proved using rule \iref{utfd3} and relating the resulting
                 configurations in $\utfolde{\utfoldv{\utfoldv{\simil}}}$.
                 The other case is as follows:
           \begin{align*}
             C_4=
             \conf{A_4\uplus\alpha}{\Gamma_4,\maps{i}{v_4}}{K_4}{s_4'}{\hat e_4}
             &\wtrans{\tau}
             \conf{A_4\uplus\alpha}{\Gamma_4,\maps{i}{v_4}}{K_4}{s_4'}{E_4\hole[\app\alpha u_4]}
             \\
             &\trans{\lpropapp{\alpha}{j}}
             \conf{A_4\uplus\alpha}{\Gamma_4,\maps{i}{v_4},\maps{j}{u_4}}{E_4,K_4}{s_4'}{\noe}
             \\
             &\trans{\lopapp{j}{\alpha'}}
             \conf{A_4\uplus\alpha,\alpha'}{\Gamma_4,\maps{i}{v_4},\maps{j}{u_4}}{E_4,K_4}{s_4'}{e_4'}
             &(\app{v_4}{\alpha'} \funred e_4')
             \\
             &\wtrans{\tau}
             C_4'
           \end{align*}
           and $C_3' \simil C_4'$. As above we can derive:
           \begin{align*}
             C_2=
             \conf{A_4\uplus\alpha}{\Gamma_4,\maps{i}{v_4}}{K_4}{s_4'}{\hat e_4}\kappa_2
             &\wtrans{\tau}
             \conf{A_4\uplus\alpha}{\Gamma_4,\maps{i}{v_4}}{K_4}{s_4'}{E_4\hole[\app\alpha u_4]}\kappa_2
             \\
             &\mathrel{({\trans{\tau}}\cup{=})}
             \conf{A_4\uplus\alpha,\alpha'}{\Gamma_4,\maps{i}{v_4},\maps{j}{u_4}}{K_4}{s_4'}
             {E_4\hole[e_4']}\kappa_2\kappa_4
             \\
             &\qquad\qquad\qquad\qquad\qquad
             (\kappa_4=\concr{\alpha'}{u_4}, (\app{v_4}{\alpha'})\kappa_2\kappa_4 \funred e_4'\kappa_2\kappa_4)
             \\
             &\wtrans{\tau}
             C_4'\kappa_2\kappa_4
           \end{align*}
           Moreover we derive $C_1' \utfolde*{\utfoldv{\utfoldv{\simil}}} C_2'$.
           \qedhere
          \end{itemize}

   \end{itemize}
 \end{proof}

\begin{proposition}\label{prop:utfold-sound} 
  Functions \utfolde{} and \utfoldv{} are sound up-to techniques.
\end{proposition}
\begin{proof}
  Consider $f=\utfolde{}\sqcup\utfoldv{}$. It suffices to show that $f$ is \WP-sound.
  By \cref{lem:prfs}~(\ref{lem:prfs-3}), it is sufficient to show
  \begin{align*}
    \utfolde{} \comp{\constf{\gfp\WP}}&\progress(f\sqcup g_1)^{\omega}\comp{\constf{\gfp\WP}} &\text{and}\\
    \utfoldv{} \comp{\constf{\gfp\WP}}&\progress(f\sqcup g_2)^{\omega}\comp{\constf{\gfp\WP}}
  \end{align*}
  where $g_1=\emptyset$ and $g_2=g=\utperm{}$.
  We have $g_i\sqsubseteq\companion$ by \cref{lem:upto-perm}.
  We finally establish the needed progressions by \cref{lem:utfolde-sound,lem:utfoldv-sound}.
\end{proof}

\subsection{Up to Context}
\todo{This section now requires no proof. Some content will move into the main paper, the rest to be removed.}

\begin{figure*}[t] 

  \[\begin{array}{@{}c@{}}
    \irule[UpToCxt${}_s$][uptocxts]{
      \conf{A_1}{\Gamma_1}{K_1}{s_1}{\hat e_1}
      \bisim*{R}
      \conf{A_2}{\Gamma_2}{K_2}{s_2}{\hat e_2}
      \\
      \vec l \not\in\dom{s_1}\sqcup\dom{s_2}
      \\
      \fl{\vec v} \sqsubseteq \{\vec l\}
      \\
      \an{\vec v} \sqsubseteq A_1\cap A_2
    }{
      \conf{A_1}{\Gamma_1}{K_1}{s_1\stupd{\vec l}{\vec v}}{\hat e_1}
      \utcxt*[][\vec l]{\bisim{R}}
      \conf{A_2}{\Gamma_2}{K_2}{s_2\stupd{\vec l}{\vec v}}{\hat e_2}
    }
    \\\\
    \irule[UpToCxt${}_\Gamma$][uptocxtgamma]{
      \conf{A_1}{\Gamma_1}{K_1}{s_1}{\hat e_1}
      \utcxt*[][\vec l]{\bisim{R}}
      \conf{A_2}{\Gamma_2}{K_2}{s_2}{\hat e_2}
      \\
      i\not\in\dom{\Gamma_1}\sqcup\dom{\Gamma_2}
      \\
      \fl{v} \sqsubseteq \{\vec l\}
      \\
      \an{v} \sqsubseteq A_1\cap A_2
    }{
      \conf{A_1}{\Gamma_1,\maps{i}{v}}{K_{1}}{s_1}{\hat e_1}
      \utcxt*[][\vec l]{\bisim{R}}
      \conf{A_2}{\Gamma_2,\maps{i}{v}}{K_{2}}{s_2}{\hat e_2}
    }
    \\\\
    \irule[UpToCxt${}_K$][uptocxtk]{
      \conf{A_1}{\Gamma_1}{K_{11},K_{12}}{s_1}{\hat e_1}
      \utcxt*[][\vec l]{\bisim{R}}
      \conf{A_2}{\Gamma_2}{K_{21},K_{22}}{s_2}{\hat e_2}
      \\
      \sizeof{K_{11}}=\sizeof{K_{21}}
      \\
      \fl{E} \sqsubseteq \{\vec l\}
      \\
      \an{E} \sqsubseteq A_1\cap A_2
    }{
      \conf{A_1}{\Gamma_1}{K_{11},E,K_{12}}{s_1}{\hat e_1}
      \utcxt*[][\vec l]{\bisim{R}}
      \conf{A_2}{\Gamma_2}{K_{21},E,K_{21}}{s_2}{\hat e_2}
    }
    \\\\
    \irule[UpToCxt${}_e$][uptocxte]{
      \conf{A_1}{\Gamma_1}{K_1}{s_1}{\noe}
      \utcxt*[][\vec l]{\bisim{R}}
      \conf{A_2}{\Gamma_2}{K_2}{s_2}{\noe}
      \\
      \fl{e} \sqsubseteq \{\vec l\}
      \\
      \an{e} \sqsubseteq A_1\cap A_2
    }{
      \conf{A_1}{\Gamma_1}{K_1}{s_1}{e}
      \utcxt*[][\vec l]{\bisim{R}}
      \conf{A_2}{\Gamma_2}{K_2}{s_2}{e}
    }
    \\\\
    \utcxt{} = \displaystyle\bigcup_{\vec l} \utcxt[][\vec l]{}
  \end{array}\]
  \hrule
  \caption{Up-to context.}\label{fig:utcontext}
\end{figure*}

\begin{lemma}
  Let 
  $g=
  \utId{}           
  \sqcup\utperm{}     
  \sqcup\utweak{}     
  \sqcup\utgc{}       
  \sqcup\utfolde{}    
  \sqcup\utfoldv{}    
  $; then
  $\utcxt[][\vec l]{}\comp{\constf{\gfp\WP}}\progress(\utcxt{}\sqcup g)^\omega\comp{\constf{\gfp\WP}}$.
\end{lemma}
\begin{proof}
  We need to show
  $\utcxt[][\vec l]{}\comp{\constf{\gfp\WP}}\comp\WP(\bisim{R}) \sqsubseteq 
  {\WP}\comp (\utcxt{}\sqcup g)^\omega\comp{\constf{\gfp\WP}}(\bisim{R})$.
  Because $\constf{\gfp\WP}$ is the constant function mapping its argument to $(\simil)$, we need to show
  $\utcxt[][\vec l]{}(\simil) \subseteq 
  {\WP}\comp (\utcxt{}\sqcup g)^\omega({\simil})$.
  
  We proceed by induction on the derivation of $C_1 \utcxt*[][\vec l]{\simil} C_2$.

  {\flushleft
  \iref{uptocxts}:}
We have
$C_1 = \conf{A_1}{\Gamma_1}{K_1}{s_1\stupd{\vec l}{\vec v}}{\hat e_1}$ and
$C_2 = \conf{A_2}{\Gamma_2}{K_2}{s_2\stupd{\vec l}{\vec v}}{\hat e_2}$ and
$C_3 = \conf{A_1}{\Gamma_1}{K_1}{s_1}{\hat e_1}$ and
$C_4 = \conf{A_2}{\Gamma_2}{K_2}{s_2}{\hat e_2}$ and
$C_3 \simil C_4$ and
$\vec l \not\in\dom{s_1}\sqcup\dom{s_2}$ and
$\fl{\vec v} \sqsubseteq \{\vec l\}$ and
$\an{\vec v} \sqsubseteq A_1\cap A_2$.
  By \cref{def:asymp}, $C_1 \asymp C_3$ and $C_2 \asymp C_4$.
  Therefore $C_1 \utgc*{\simil} C_2$. By \cref{lem:utgc-sound},
  $C_1 \simil C_2$, thus
  $C_1 \mathrel{\WP\comp\utId{\simil}} C_2$, and finally
  $C_1 \mathrel{\WP\comp(\utcxt{}\cup g)^\omega(\simil)} C_2$.

  {\flushleft
  \iref{uptocxtgamma}:}
We have
$C_1 = \conf{A_1}{\Gamma_1,\maps{i}{v}}{K_{1}}{s_1}{\hat e_1}$ and
$C_2 = \conf{A_2}{\Gamma_2,\maps{i}{v}}{K_{2}}{s_2}{\hat e_2}$ and
$C_3 = \conf{A_1}{\Gamma_1}{K_1}{s_1}{\hat e_1}$ and
$C_4 = \conf{A_2}{\Gamma_2}{K_2}{s_2}{\hat e_2}$ and
$C_3 \utcxt*[][\vec l]{\bisim{R}} C_4$ and
$i\not\in\dom{\Gamma_1}\cup\dom{\Gamma_2}$ and
$\fl{v} \subseteq \{\vec l\}$ and
$\an{v} \subseteq A_1\cap A_2$.
  \todo[inline]{todo - Maybe up-to separation covers this lemma entirely!}

  Let $C_1 \utcxt*[][\vec l]{\simil} C_2$ and $C_1 \trans{\eta} C_1'$. We need to show that there exists $C_2'$ such that 
  $C_2 \wtrans{\eta} C_2'$ and $C_1' \bisim*{S} C_2'$, where $\bisim{S} = (\utcxt{}\cup g)^\omega(\simil)$.

  The proof is similar in the remaining two cases.
\qed
\end{proof}

\begin{lemma}
  $e_1 \cxteq e_2$ iff
  $\lam x e_1 \cxteq \lam x e_2$
\end{lemma}
\begin{proof}
  See Theorem 3.2 in \cite{KoutavasW06}.
\end{proof}
\begin{theorem}[Soundness of $\bisimil$]
  $e_1 \bisimil e_2$ implies $e_1 \cxteq e_2$.
\end{theorem}
\begin{proof}
  It suffices to show that $(bisimil)$ is a congruence.
  By the above lemma, it suffices to show that $(bisim)$ is a congruence for values.
  Let $v_1 \bisimil v_2$  and context $D$.
  We have
  \begin{align*}
    v_1 &\bisimil v_2
    \\
    \conf{\cdot}{\maps{i}{v_1}}{\cdot}{\cdot}{\cdot}
    &\bisimil
    \conf{\cdot}{\maps{i}{v_1}}{\cdot}{\cdot}{\cdot}
    &\text{ (by bisimulation def. and \iref{propretf} transition)}
    \\
    \conf{\alpha}{\maps{i}{v_1}}{\cdot}{\cdot}{\cdot}
    &\bisimil
    \conf{\alpha}{\maps{i}{v_2}}{\cdot}{\cdot}{\cdot}
    &\text{ (by soundness of $\utgc{\bisimil}$)}
    \\
    \conf{\alpha}{\maps{i}{v_1}}{\cdot}{\cdot}{D\hole[\alpha]}
    &\bisimil
    \conf{\alpha}{\maps{i}{v_2}}{\cdot}{\cdot}{D\hole[\alpha]}
    &\text{ (by soundness of $\utcxt{\bisimil}$)}
    \\
    \conf{\cdot}{\cdot}{\cdot}{\cdot}{D\hole[v_1]}
    &\bisimil
    \conf{\cdot}{\cdot}{\cdot}{\cdot}{D\hole[v_2]}
    &\text{ (by soundness of $\utfoldv{\bisimil}$)}
    \\
    D\hole[v_1] &\bisimil D\hole[v_2]
    &\text{ (by definition)}
  \end{align*}
  \qed

\end{proof}


%% file: completeness.tex
\newcommand\traces[2][]{\mathsf{Traces}_{#1}(#2)}
\newcommand\cotraces[2][]{\underline{\mathsf{Traces}}_{#1}(#2)}
\newcommand\absNames{\mathcal{A}}
\newcommand\indNames{\mathcal{I}}
\newcommand\rcomp{\triangleright}
\newcommand\RR{\bisim{R}}
\newcommand\dualz[1]{\mathsf{Dualz}(#1)}
\newcommand\dual[2]{\overline{#1}^{#2}}
\newcommand\newz[1]{\mathsf{new}(#1)}
\newcommand\ext[1]{\mathsf{Extn}(#1)}

\section{Completeness of (\bisimil)}\label{sec:completeness}

Let us denote by $\absNames$ and $\indNames$ the countably infinite sets of abstract and index names respectively.\todo{are the indices names or just numbers?}
In this section we present the proof of the following result.

\begin{theorem}[Completeness]\label{thm:complete}
For any two doubly closed expressions $e_1$ and $e_2$, if $e_1\cxteq e_2$ then $e_1\bisimil e_2$.
\end{theorem}

We start off with a few auxiliary results.

\begin{lemma}\label{lem:l1}
Let $C$ be a configuration with state $s$, and $r$ a location such that $s(r)=\false$ and:
  \begin{itemize}
  \item all the assignments of $r$ occurring in $C$ are of the form $r\asgn\false$
  \item $C$ contains a subterm (in one of its $\Gamma,K,s,\hat e$) of the form $\cond{\deref r}{e_\true}{e_\false}$.
  \end{itemize}
  Then, ${C}\bisimil{C'}$, where $C'$ is obtained from $C$ by replacing the subterm above with $e_\false$.
\end{lemma}
\begin{proof}
Let us define the relation $C\funred_r C'$ to hold for each pair of $C,C'$ as above. The statement then follows by  showing that $\bisim{R}=\{ (C,C') \mid C\funred^*_r C'\}$ is a weak bisimulation up to $\utbeta{}$.
\end{proof}

\begin{lemma}\label{lem:l2}
  Let $C=\conf{A}{\Gamma}{K}{s}{e}$ be a configuration with $\alpha\in A$, $s(l)=\lambda z.\alpha z$ for some location $l$ that has no assignments in $C$, and:
  \begin{itemize}
  \item $C'$ be obtained from $C$ by replacing an occurrence of $\alpha$ with $\lambda z.\alpha z$; or
  \item $C'$ be obtained from $C$ by replacing an occurrence of $\alpha$ with $\deref l$ (other than in $s(l)=\lambda z.\alpha z$).
  \end{itemize}
  Then, ${C}\bisimil{C'}$.
\end{lemma}
\begin{proof}
Let us define the relation $C\funred_\alpha C'$ to hold for each pair of $C,C'$ as above. The statement then follows by  showing that $\bisim{R}=\{ (C,C') \mid C\funred^*_\alpha C'\}$ is a weak bisimulation up to $\utbeta{}$.
\end{proof}

\begin{lemma}\label{lem:l3}\todo{unused}
  Let $C=\conf{A}{\Gamma}{K}{s}{e}$ be a configuration with $\alpha\in A$, $s(l)=\lambda z.c z$ for some constant $c$ and location $l$ that has no assignments in $C$, and:
  \begin{itemize}
  \item $C'$ be obtained from $C$ by replacing an occurrence of $\alpha$ with $\lambda z.c z$; or
  \item $C'$ be obtained from $C$ by replacing an occurrence of $\alpha$ with $\deref l$ (other than in $s(l)=\lambda z.c z$).
  \end{itemize}
  Then, ${C'}\simil{C}$.
\end{lemma}
\begin{proof}
Let us define the relation $C\funred_c C'$ to hold for each pair of $C,C'$ as above. The statement then follows by  showing that $\bisim{R}=\{ (C',C) \mid C\funred^*_c C'\}$ is a weak simulation up to $\utbeta{}$.
\end{proof}

\begin{definition}[Traces]
For any $C$, we let $\traces{C}\defeq \{\ t \mid \tau,{\lterm}\notin t\land C \wtrans{t}{\!\!}^*{}\trans{\downarrow}\botconf\}$.
\end{definition}

\begin{lemma}\label{lem:similtraces}\todo{only implicitly used, but useful to have}
  For any two configurations $C_1$ and $C_2$, $C_1\simil C_2$ iff  $\traces{C_1}\subseteq\traces{C_2}$.
\end{lemma}
\begin{proof}
For the right-to-left direction , let $\bisim{R}$ be the relation:
  \[
\bisim{R}\defeq\{(C_1,C_2)\mid \traces{C_1}\subseteq\traces{C_2}\}.
\]
We claim that $\bisim{R}$ is a simulation. 
Take $C_1\bisim{R}C_2$. 
Let $C_1\trans{\eta}C_1'$ and suppose that $C_1'\neq\botconf$.
If $\eta=\tau$ then, since the operational semantics is deterministic and transitions to $\botconf$ do not contribute traces, $\traces{C_1}=\traces{C_1'}$ and therefore $C_1'\bisim{R}C_2$.
If $\eta$ is one of $\lpropapp{\alpha}{v},\lpropret{v}$, for some $\alpha,v$,
then there are two cases:
\begin{itemize}
\item If there exists some $\eta t\in\traces{C_1}$ then equality of traces implies that $C_2\wtrans{\eta}C_2'$, for some $C_2'\neq\botconf$. By determinacy of the LTS outside $\botconf$'s, we have that $\traces{C_1'}=\traces{C_2'}$.
  \item Otherwise, $C_2\trans{\eta}\botconf$ and $\traces{C_1'}=\traces{\botconf}=\emptyset$.
\end{itemize}
Therefore, in both cases $C_1'\bisim{R}C_2'$.
Similarly if $\eta$ is one of $\lopapp{i}{v},\lopret v$, for some $i,v$.
Finally,
let $C_1\trans{\eta}\botconf$. If $\eta={\lterm}$ then  $\epsilon\in\traces{C_i}$ and therefore $C_2\wtrans{\downarrow}\botconf$;
if $\eta\not={\lterm}$ then  $C_2\trans{\eta}\botconf$.
In both cases, we conclude by noting that $\botconf\bisim{R}\botconf$.

Conversely, suppose $C_1\simil C_2$ and let $t=\eta_1\cdots\eta_n\in\traces{C_1}$. Then, there are $C_1=C_1^0,C_1^1,\dots,C_1^n$ such that $C_1^{i-1}\wtrans{\eta_i}C_1^i$ , for each $1\leq i\leq n$,
 and $C_1^n\wtrans{\lterm}\botconf$. Since $C_1\simil C_2$, $C_2$ can simulate these transitions and produce the same trace $t$.
\end{proof}

\begin{corollary}
For any two doubly closed expressions $e_1$ and $e_2$, $e_1\bisimil e_2$ iff $\traces{C_{e_1}}=\traces{C_{e_2}}$, where $C_{e_i}=\conf{\emptyA}{\emptyG}{\emptyK}{\emptyS}{e_i}$.
\end{corollary}
\begin{proof}
Directly from the previous lemma.
\end{proof}

We call a finite partial bijection $\phi:\absNames\cup\indNames\overset{\cong}{\rightharpoonup}\absNames\cup\indNames$ a \emph{dualiser} if, for all $x\in\dom{\phi}$, $x\in\absNames\iff \phi(x)\notin\absNames$.
%
%
%
Given stacks $K_1,K_2$, we define their (left) composition $K_1\rcomp K_2$ recursively by:
\[
K_1\rcomp K_2 \defeq\begin{cases}
\hole & \text{if }K_1=K_2=\emptyK\\
(K_1\rcomp K_2')[E] &\text{if }K_2=E,K_2'\text{ and } |K_1|=|K_2|\\
(K_1'\rcomp K_2)[E] &\text{if }K_1=E,K_1'\text{ and } |K_1|=|K_2|+1\\
\text{undefined} &\text{otherwise}
\end{cases}
\]
Thus, if well defined, $K_1\rcomp K_2$ is an evaluation context.

\begin{definition}
  For any two configurations $C_1,C_2$, with $C_i=\conf{A_i}{\Gamma_i}{K_i}{s_i}{\hat e_i}$ and dualiser $\phi$, we say that \emph{$(C_1,\phi)$ is a context for $C_2$} if:
  \begin{itemize}
  \item $\phi:A_2\cup\dom{\Gamma_2}\overset{\cong}{\to}A_1\cup\dom{\Gamma_1}$
  \item $\dom{s_1}\cap\dom{s_2}=\emptyset$
  \item exactly one of $\hat e_1,\hat e_2$ is $\noe$
  \item $K_1\rcomp K_2$ is well defined
      \item taking $\psi$ to be the following map from $A_1\cup A_2$ to terms: \
  $
\psi = (\Gamma_2\circ (\phi^{-1}\upharpoonright A_1))\cup(\Gamma_1\circ (\phi\upharpoonright A_2))
   $, there is some $k$ such that $\psi^k=\psi^{k+1}$ \ (note: $\psi^1\defeq\psi$ and $\psi^{n+1}\defeq\alpha\mapsto (\psi^n(a))\psi$).
  \end{itemize}
  In such a case, we set $C_1[C_2]_\phi\defeq\redconf{s_1\cup s_2}{(K_1\rcomp K_2)[e]}\psi^*$, where $e\in\{\hat e_1,\hat e_2\}\setminus\{\noe\}$ and $\psi^*\defeq\psi^k$.
\end{definition}

We write $C_1\diamond C_2$ when $C_1$ and $C_2$ have the same contexts.
We call a configuration $C$ \emph{lost} if either $C=\botconf$ or there is no context $(C',\phi)$ for $C$ such that $C'[C]_\phi\Downarrow$. Finally, for moves $\eta,\eta'$, we let $\eta\funred_P\eta'$ hold if:
\begin{itemize}
  \item $\eta$ is in one of the forms $\lpropapp{\alpha}{i}$ or $\lpropret{i}$, and
  \item $\eta'$ is in one of the forms $\lpropapp{\alpha}{c}$ or $\lpropret{c}$, respectively, for some constant $c$.
\end{itemize}

\begin{theorem}[One-step definability]\label{thm:definty}
  Given $C_0,C_0',C',\phi'$ and $\eta\neq\tau,\lterm$ such that $C_0\trans{\eta}C_0'$ and $C'[C_0']_{\phi'}\Downarrow$, there is a context $(C,\phi)$ for $C_0$ such that:
  \begin{itemize}
  \item $C[C_0]_\phi\Downarrow$ and
  \item for all $C_0''\diamond C_0$, $C[C_0'']_\phi\Downarrow$ implies that $C_0''\wtrans{\eta'}$ for some $\eta\funred_P\eta'$.
\end{itemize}
  \end{theorem}
  \begin{proof}
    Let us assume
$C_0=\conf{A_0}{\Gamma_0}{K_0}{s_0}{\hat e_0}$,
    $C'=\conf{A'}{\Gamma'}{K'}{s'}{\hat e'}$ and build a configuration $C$ and a dualiser $\phi'\subseteq\phi$. We do case analysis on $\eta$.

    If $\eta= \lpropapp{\alpha}{i}$,
    let $i_\alpha=\phi'(\alpha)\in\dom{\Gamma'}$ and $\alpha_i=\phi'(i)\in A'$.
    We have $\hat e'=e'$ and
    take $\phi=\phi'\setminus\{(i,\alpha_i)\}$ and
    $C=\conf{A'\setminus\{\alpha_i\}}{\Gamma}{K}{s}{e}{\noe}[\deref l/\alpha_i]$ with $s=s'\uplus[r\mapsto\true,l\mapsto\true]$.
We also let $\dom{\Gamma}=\dom{\Gamma'}$ and for each $j\in\dom{\Gamma'}$, assuming $\Gamma'(j)=\lambda x.e_j$:
 $$
 \Gamma(j)\defeq \begin{cases}
     \lambda x. \cond{\deref r}{\bot}{e_j}& \text{if }j\neq i_\alpha\\
     \lambda x. \cond{\deref r}{(r\asgn\false;\,l\asgn \lambda z.xz;\,{e'})}{e_j} & \text{if }j=i_\alpha
     \end{cases}$$
     Moreover, if $K'$ is empty then $K\defeq K'$; and if $K'=E',K_1$ then \
$K \defeq E'[(\lambda x.\cond{\deref r}{\bot}{x})\hole],K_1$.
Setting $\Gamma'',K'',s'',e''$ to be ${\Gamma[\deref l/\alpha_i]},{K[\deref l/\alpha_i]},{s'[\deref l/\alpha_i]},e'[\deref l/\alpha_i]$ respectively, the non-blocking transitions from $C$ are
$C\wtrans{\lopapp{i_\alpha}{\alpha_i}}C''=\conf{A'}{\Gamma''}{ K''}{ s''\uplus[r\mapsto\false,l\mapsto\lambda z.\alpha_i z]}{ e''}$ (modulo renamings of $\alpha_i$)
and
$C\wtrans{\lopapp{i}{c}}C'_{c}=\conf{A}{\Gamma''}{ K''}{ s''\uplus[r\mapsto\false,l\mapsto\lambda z.cz]}{e''}$, for all constants $c$.
Let us now consider $C[C_0]_\phi=\redconf{s_0\uplus s}{(K\rcomp K_0)[E[\alpha v]]}\psi^*$, with $\psi$ defined as above.
By construction,  $C[C_0]_\phi\trans{}^*C''[C_0']_{\phi'}$ and thus, by \cref{lem:l1,lem:l2}, $C[C_0]_\phi\Downarrow$.
Take now some $C_0''\diamond C_0$. By construction, in order for $C[C_0'']_\phi$ to take a transition, it must be an internal call to $i_\alpha$ with an arbitrary argument, i.e.\ we need $C_0''\wtrans{\eta'}$ for some $\eta\funred_P\eta'$.

If $\eta= \lpropret{i}$, 
    let $\alpha_i=\phi'(i)\in A'$.
    We have $\hat e'=e'$ and take $C=\conf{A'\setminus\{\alpha_i\}}{\Gamma}{E,K'}{s}{e}{\noe}[\deref l/\alpha_i]$ with $s=s'\uplus[r\mapsto\true,l\mapsto\true]$, and $\phi=\phi'\setminus\{(i,\alpha_i)\}$.
    We also let $\dom{\Gamma}=\dom{\Gamma'}$ and set
$\Gamma(j)\defeq
     \lambda x. \cond{\deref r}{\bot}{e_j}$ (where $\Gamma(j)=\lambda x.e_j$) and
$
  E \defeq
    (\lambda x.r\asgn\false;\,l\asgn\lambda z.xz;\, e)\hole.
$
The argument then follows that of the previous case.

If $\eta= \lpropapp{\alpha}{c}$,
    let $i_\alpha=\phi'(\alpha)\in\dom{\Gamma'}$.
    We have $\hat e'=e'$ and take $C=\conf{A'}{\Gamma}{K}{s}{\noe}$ with $s=s'\uplus[r\mapsto\true]$, and $\phi=\phi'$.
We also let $\dom{\Gamma}=\dom{\Gamma'}$ and for each $j\in\dom{\Gamma'}$, assuming $\Gamma'(j)=\lambda x.e_j$:
 $$
 \Gamma(j)\defeq \begin{cases}
     \lambda x. \cond{\deref r}{\bot}{e_j}& \text{if }j\neq i_\alpha\\
     \lambda x. \cond{\deref r}{(r\asgn\false;\,\cond{x==c}{e}{\bot})}{e_j} & \text{if }j=i_\alpha
     \end{cases}$$
     Moreover, if $K'$ is empty then $K\defeq K'$; and if $K'=E',K_1$ then \
$K \defeq E'[(\lambda x.\cond{\deref r}{\bot}{x})\hole],K_1$.
The only non-blocking transition from $C$ is
$C\wtrans{\lopapp{i_\alpha}{c}}C''=\conf{A'}{\Gamma}{ K}{ s'\uplus[r\mapsto\false]}{ e}$.
Let us now consider $C[C_0]_\phi=\redconf{s_0\uplus s}{(K\rcomp K_0)[E[\alpha c]]}\psi^*$.
By construction,  $C[C_0]_\phi\trans{}^*C''[C_0']_{\phi'}$ and thus, by \cref{lem:l1}, $C[C_0]_\phi\Downarrow$.
Take now some $C_0''\diamond C_0$. By construction, in order for $C[C_0'']_\phi$ to take a transition, it must be an internal call to $i_\alpha$ with  argument $c$, i.e.\ we need $C_0''\wtrans{\eta}$.

If $\eta= \lpropret{c}$,
    we have $\hat e'=e'$ and take $C=\conf{A'}{\Gamma}{E,K'}{s}{\noe}$ with $s=s'\uplus[r\mapsto\true]$, and $\phi=\phi'$.
We also let $\dom{\Gamma}=\dom{\Gamma'}$ and for each $j\in\dom{\Gamma'}$, $\Gamma(j)\defeq
     \lambda x. \cond{\deref r}{\bot}{e_j}$ (where $\Gamma(j)=\lambda x.e_j$) and \
$
  E \defeq
    (\lambda x.r\asgn\false;\cond{x==c}{e}{\bot})\hole.
$
We conclude as in the previous case.

If $\eta= \lopapp{\alpha}{v}$,
we must have $K'=E,K$.
We take
$C=\conf{A'}{\Gamma'\setminus\{v\}}{K}{s'}{E[\phi(\alpha) v']}$, with $ v'=\Gamma'(v)$ if $v\in\indNames$, and $ v'=v$ otherwise.
Moreover, we set $\phi=\phi'\setminus\{(v,\phi'(v))\}$ if $v\in\indNames$, and $\phi=\phi'$ otherwise.
We can see that $C[C_0]_\phi\trans{}C'[C_0']_{\phi'}$ and, for any $C_0''\diamond C_0$, if $C[C_0'']_\phi$ is well formed then  $C_0''\trans{\eta}$. 
Similarly, if  $\eta= \lopret{v}$,
we construct
$C=\conf{A'}{\Gamma'\setminus\{v\}}{K'}{s'}{ v'}$ and $\phi$, with $v'$ and $\phi$ defined as above.
\end{proof}

\begin{proof}[Proof of \cref{thm:complete}]
  It suffices to show that the following relation is a bisimulation.
  \begin{align*}
\bisim{R} \defeq \{ (C_1,C_2)\mid C_1,C_2\text{ lost} \lor(C_1\diamond C_2\land  \forall (C,\phi)\in\mathsf{Cxt}(C_i).\ C[C_1]_\phi\Downarrow\iff C[C_2]_\phi\Downarrow)\}
  \end{align*}
  Let $C_1\bisim{R}C_2$ and suppose $C_1\trans{\eta}C_1'$. If $C_1'$ is lost then $C_2\trans{\eta}\botconf$, as required. Otherwise:
  \begin{itemize}
  \item If $\eta=\tau$ then, for all $C,\phi$, $C[C_1]_\phi\Downarrow$ iff $C[C_1']_\phi\Downarrow$, so $C_1'\RR C_2$.
  \item If $\eta=\lpropapp{\alpha}{v}$ then let $(C',\phi')$ be a context such that $C'[C_1']_{\phi'}\Downarrow$. By \cref{thm:definty}, there is a context $(C,\phi)$ such that
$C[C_1]_\phi\Downarrow$ and,
for all $C_1''\diamond C_1$, if $C[C_1'']_\phi\Downarrow$ then $C_1''\wtrans{\eta'}$, for some $\eta\funred_P\eta'$. By hypothesis, $C[C_2]_\phi\Downarrow$, thus $C_2\wtrans{\eta'}C_2'$ and $\eta\funred_P\eta'$.
If $\eta'\neq\eta$ then we repeat the same argument (swapping the roles of $C_1$ and $C_2$) to conclude that $C_1\wtrans{\eta''}C_1''$ for some $\eta''\funred_P\eta'$, which contradicts determinacy of our LTS. Thus, $\eta=\eta'$ and,
again by hypothesis, $C_1'\bisim{R}C_2'$.
The case for $\eta=\lpropret{v}$ is treated similarly.
  \item If $\eta=\lopapp{i}{v}$ then let $(C',\phi')$ be a context such that $C'[C_1']_{\phi'}\Downarrow$. By \cref{thm:definty}, there is a context $(C,\phi)$ such that
$C[C_1]_\phi\Downarrow$ and,
for all $C_1''\diamond C_1$, if $C[C_1'']_\phi\Downarrow$ then $C_1''\wtrans{\eta}$. By hypothesis, $C[C_2]_\phi\Downarrow$, thus $C_2\wtrans{\eta}C_2'$. Now pick any context $(C_0',\phi_0')$ such that $C_0'[C_1']_{\phi_0'}\Downarrow$, say $C_0'=\conf{A_0}{\Gamma_0}{E,K_0}{s_0}{\noe}$. Taking $C_0'=\conf{A_0}{\Gamma_0\setminus\{\phi_0'(v)\}}{K_0}{s_0}{E[\phi_0'(i)v']}$, where $v'=v$ if the latter is a constant and $v'=\Gamma_0(\phi_0'(v))$ otherwise,
and $\phi_0=\phi$ if $v$ is a constant and $\phi_0=\phi_0'\setminus\{(v,\phi_0'(v))\}$ otherwise,
we have $C_0'[C_n]_{\phi_0}\trans{}C_0[C_n']_{\phi_0'}$, for $n=1,2$. Thus, $C_0'[C_n]_{\phi_0}\Downarrow$ and hence $C_0[C_2']_{\phi_0}\Downarrow$.
The case of $\eta=\lopret{v}$ is addressed similarly.
  \end{itemize}
\end{proof}

%% file: examples.tex
\section{Examples}\label{app:examples}

\subsection{Simple Invariants}

\textit{Up to invariants} states that values stored in references can be abstracted if they validate a predicate. Consider \cref{ex:inv}.
\begin{example}\label{ex:inv}
  ~\\\vspace{-1.3em}
  \begin{lstlisting}
    $M =$ ref x = 0 in fun () -> x++; !x > 0
    $N =$ fun () -> true
  \end{lstlisting}
\end{example}
This example is also like the previous two. This time, the reference $x$ is incremented in $M$ and the function checks whether $x$ holds a positive integer.
\begin{center}
\scalebox{0.7}{
\begin{tikzpicture}[node distance=2.2cm]
\node (s0) [pnode] {$M$};
\node (s1) [onode, right of=s0] {$x = 0$};
\node (s2) [pnode, right=1.5cm of s1] {$x = 0$};
\node (s3) [onode, right=1.5cm of s2] {$x = 1$};
\node (s4) [pnode, right=1.5cm of s3] {$x = 1$};
\node (s5) [onode, right=1.5cm of s4] {$x = 2$};
\node[draw=none] (end) [right=1.5cm of s5] {$\cdots$};

\draw [->] 
(s0) edge node[anchor=south]{$\pret{g}$} (s1)
(s1) edge node[anchor=south]{$\oapp{g}{()}$} (s2)
(s2) edge node[anchor=south]{$\pret{true}$} (s3)
(s3) edge node[anchor=south]{$\oapp{g}{()}$} (s4)
(s4) edge node[anchor=south]{$\pret{true}$} (s5)
(s5) edge node[anchor=south]{$\oapp{g}{()}$} (end)
;
\end{tikzpicture}
}
\end{center}
We can see above that each call to $g$ increments the state, which makes it hard to find cycles. We annotate the function with an invariant as follows:
\[M = \text{\lstinline{ref x = 0 in fun () \{ w | x as w | w >= 0 \} -> x++; !x > 0}}\]
The invariant---shown in the curly braces---states that the value in $x$ can be abstracted by any $w$ such that $w \geq 0$, so long as $x \geq 0$ is also valid.
\begin{center}
\scalebox{0.7}{
\begin{tikzpicture}[node distance=2.2cm]
\node (s0) [pnode] {$M$};
\node (s1) [onode, right of=s0] {$x = 0$};
\node (s2) [pnode, right=1.5cm of s1,label=north:$w \geq 0$,label=south:$s_1$] {$x = w$};
\node (s3) [onode, right=1.5cm of s2,label=north:$w \geq 0$] {$x = w$};
\node (s4) [pnode, right=1.5cm of s3,label=north:$w \geq 0$,label=south:$s_1$] {$x = w$};
\node[draw=none] (end) [right=0.8cm of s4] {};
\node[draw=none] (s5) [right=3.3cm of end] {};

\draw [->] 
(s0) edge node[anchor=south]{$\pret{g}$} (s1)
(s1) edge node[anchor=south]{$\oapp{g}{()}$} (s2)
(s2) edge node[anchor=south]{$\pret{true}$} (s3)
(s3) edge node[anchor=south]{$\oapp{g}{()}$} (s4)
;
\draw [->,strike through] (s4) -- (end);
\end{tikzpicture}
}
\end{center}
Shown above, we see that the states labelled $s_1$ are identical, which lets us prune it, e.g. via memoisation, to end the game.

\subsection{Landin's Fixpoint}
\begin{example}
  The following equivalence relates Landin's imperative fixpoint operator with a fixpoint with letrec.
  The type of the two expressions is $((\Int \arrow \Int) \arrow \Int \arrow \Int) \arrow \Int \arrow \Int$.

\noindent
    \begin{tabular}{llll}
      $M =\;$
      &\begin{minipage}[t]{0.47\textwidth}\vspace{-1.3em} \begin{lstlisting}[boxpos=t]
let landinsfixpoint f =
  ref x = fun z -> z in
  x:= (fun y {} -> f !x y); !x
in landinsfixpoint
\end{lstlisting}
    \end{minipage}
  &
      $N =\;$ 
      &\begin{minipage}[t]{0.3\textwidth}\vspace{-1.3em} \begin{lstlisting}[boxpos=t]
let rec fix f =
  (fun y -> f (fix f) y)
in fix
\end{lstlisting}
    \end{minipage}
    \end{tabular}

In this example, up to separation removes the outer functions from the $\Gamma$ environments, thus they are only applied once.
However the inner functions 
\lstinline|(fun y {} -> f !x y)|
and
\lstinline{(fun y -> f (fix f) y)},
which are provided as arguments to opponent function \lstinline{f},
cannot be removed from the $\Gamma$ environments by up to separation because of the access to location 
\lstinline{x}, and are arbitrarily nested in the bisimulation transition system.
The up to re-entry technique removes the need for this nesting.
The syntax \lstinline|{}| serves as the flag to apply this technique to these inner functions.
\qed
\end{example}

\subsection{Full Example~\ref{ex:meyer-sieber-e6}}\label{sec:meyer-sieber-e6-full}

\begin{example}
  The following is the full example by Meyer and Sieber~\cite{MeyerS88} featuring location passing,
  adapted to \lang where locations are local.
  \begin{lstlisting}
$M =$ let loc_eq loc1loc2 =
	  let (l1,l2) = loc1loc2 in
	  let (r1,w1) = l1 in
	  let (r2,w2) = l2 in
	  let val1 = r1 () in let val2 = r2 () in 
	  w2(val2+1);
	  let res = if r1() = val1+1 then true else false
	  in w1(val1); w2(val2); res in
	fun q ->
	  ref x = 0 in
	  let locx = (fun () -> !x) , (fun v -> x := v) in
	  let almostadd_2 locz {w | x as w | w mod 2 == 0} =
	    if loc_eq (locx,locz) then x := 1 else x := !x + 2
	  in q almostadd_2; if !x mod 2 = 0 then _bot_ else ()

$N =$ fun q -> _bot_
  \end{lstlisting}
%
%

\end{example}
